\documentclass[aps,prx,twocolumn,nofootinbib,superscriptaddress,notitlepage]{revtex4-2}
\usepackage{color, xcolor, colortbl}
\usepackage{graphicx}
\usepackage{amsthm,amsmath,amssymb,amsfonts}
\usepackage{dcolumn}% Align table columns on decimal point
\usepackage{url}
\usepackage{epstopdf}
\usepackage{enumitem}
\usepackage{algorithm}
\usepackage{algpseudocode}
\usepackage{bm}
\usepackage[caption=false]{subfig}
\usepackage{appendix}
\usepackage{multirow}
\usepackage{braket}
\usepackage[english]{babel}
\usepackage{hyperref}
\usepackage[capitalize]{cleveref}
\usepackage{xpatch}
\usepackage{tikz}
\usepackage{adjustbox}
\usepackage{xspace}
\usetikzlibrary{quantikz}

\renewcommand{\Re}{\operatorname{Re}}
\renewcommand{\Im}{\operatorname{Im}}

\newcommand{\Tr}{\operatorname{Tr}}

\newcommand{\I}{\mathrm{i}}

\newcommand{\mc}[1]{\mathcal{#1}}
\newcommand{\mf}[1]{\mathfrak{#1}}
\newcommand{\wt}[1]{\widetilde{#1}}

\newcommand{\abs}[1]{\left\lvert#1\right\rvert}
\newcommand{\norm}[1]{\left\lVert#1\right\rVert}

\newcommand{\ud}{\,\mathrm{d}}
\newcommand{\Or}{\mathcal{O}}
\newcommand{\EE}{\mathbb{E}}
\newcommand{\NN}{\mathbb{N}}
\newcommand{\RR}{\mathbb{R}}
\newcommand{\CC}{\mathbb{C}}
\newcommand{\ZZ}{\mathbb{Z}}

\newtheorem{thm}{\protect\theoremname}
\newtheorem{lem}[thm]{\protect\lemmaname}

\newtheorem{cor}[thm]{\protect\corollaryname}

\newtheorem{defn}[thm]{\protect\definitionname}

\newtheorem*{problem*}{Problem}

\providecommand{\definitionname}{Definition}
\providecommand{\assumptionname}{Assumption}
\providecommand{\corollaryname}{Corollary}
\providecommand{\lemmaname}{Lemma}
\providecommand{\propositionname}{Proposition}
\providecommand{\remarkname}{Remark}
\providecommand{\theoremname}{Theorem}
\usepackage{bbm}



\newcommand{\TOR}{\mathcal{T}}

\usepackage{mathrsfs}

\usepackage{tabularx}

\theoremstyle{definition}

\newcommand{\SP}{{Supplementary Information}}

\usepackage{soul}

\newcommand{\apq}[1]{\bar{#1}}

\begin{document}

\title{In Situ Quantum Analog Pulse Characterization via Structured Signal Processing}
\author{Yulong Dong}
\email[Electronic address: ]{dongyl@umich.edu}
\affiliation{Department of Electrical Engineering and Computer Science, University of Michigan, Ann Arbor, MI 48109, USA}
\author{Christopher Kang}
\affiliation{Department of Computer Science, University of Chicago, Chicago, IL 60637, USA}
\author{Murphy Yuezhen Niu}
\email[Electronic address: ]{murphyniu@ucsb.edu}
\affiliation{Google Quantum AI, Venice, California 90291, USA.}
\affiliation{Department of Computer Science, University of California, Santa Barbara, California, 93106, USA.}

\date{\today}

%TC:ignore

\begin{abstract}
Analog quantum simulators can directly emulate time-dependent Hamiltonian dynamics, enabling the exploration of diverse physical phenomena such as phase transitions, quench dynamics, and non-equilibrium processes. Realizing accurate analog simulations requires high-fidelity time-dependent pulse control, yet existing calibration schemes are tailored to digital gate characterization and cannot be readily extended to learn continuous pulse trajectories. We present a characterization algorithm for in situ learning of pulse trajectories by extending the Quantum Signal Processing (QSP) framework to analyze time-dependent pulses. By combining QSP with a logical-level analog-digital mapping paradigm, our method reconstructs a smooth pulse directly from queries of the time-ordered propagator, without requiring mid-circuit measurements or additional evolution. Unlike conventional Trotterization-based methods, our approach avoids unscalable performance degradation arising from accumulated local truncation errors as the logical-level segmentation increases. Through rigorous theoretical analysis and extensive numerical simulations, we demonstrate that our method achieves high accuracy with strong efficiency and robustness against SPAM as well as depolarizing errors, providing a lightweight and optimal validation protocol for analog quantum simulators capable of detecting major hardware faults.
\end{abstract}

\maketitle

%TC:endignore

\section{Introduction}\label{sec:intro}

Analog quantum simulators offer a powerful platform for exploring non-equilibrium quantum dynamics and emergent many-body behavior~\cite{bairey2019learning,qi2019determining,Hsieh2020,hangleiter2021precise,Zoller2021,Zoller2022}. By directly emulating continuous-time Hamiltonian evolution, they bypass the need for digital compilation and enable the study of rich dynamical phenomena. However, this same analog nature makes the underlying control pulses difficult to characterize: they are continuous, device-specific, and oftentimes incompatible with mid-circuit measurements. Consequently, standard gate-based calibration and tomography techniques~\cite{kimmel,neill_accurately_2021,PhysRevA.85.042311,PhysRevA.77.012307,PhysRevLett.106.180504,blume2017demonstration}, which assume discretized control access and Markovian noise, do not apply to realistic analog hardware subject to drift, non-Markovian crosstalk, and contextual errors~\cite{acharya2022suppressing,tls,rol2020time,NiuBoixoSmelyanskiyEtAl2019}.
 
Existing approaches rely either on Trotterized digital surrogates~\cite{Zoller2022} or global fitting of experimental observables~\cite{Zoller2021,andersen2024thermalization}. Trotterization introduces accumulating discretization errors that mask fine pulse structure, while black-box optimization over continuous control parameters is non-convex, unstable, and not scalable in time or system size. As a result, there remains no practical \textit{in situ} method for reconstructing the true analog control trajectory directly from hardware data without halting or perturbing the ongoing evolution.

Here, we introduce a characterization framework that bridges digital learning tools with analog quantum control using a structured signal-processing approach. The method follows an analog-to-digital-to-analog workflow: first, the continuous control pulse is mapped to a set of discrete parameters through a logical analog-to-digital mapping within the Quantum Signal Processing (QSP) framework, with rigorously bounded digitization bias that decreases with finer segmentation rather than accumulating as in Trotterization. Second, these digital surrogate parameters are estimated directly from experimental data using a QSP-based direct learning method, avoiding the usage of black-box optimization. To ensure robustness against SPAM and depolarizing noise, we further develop an enhanced tomography subroutine to guarantee robustness against State Preparation and Measurement~(SPAM)  and depolarizing errors. Finally, a digital-to-analog conversion reconstructs a smooth pulse using spline interpolation. We prove that the end-to-end procedure achieves bounded bias and variance under limited shots, with error scalings that are arguably optimal.

By embedding pulse learning within the QSP formalism, this framework unifies digital and analog control characterization, enabling tools developed for digital gate characterizations~\cite{LowYoder2016,ni2023low,dong2025optimal,GilyenSuLowEtAl2018,Haah2019} to analyze continuous-time dynamics. Because estimation takes place entirely at the level of the time evolution operator, the protocol operates fully \textit{in situ}, without mid-circuit access, ancilla qubits, or hardware modifications. The resulting reconstruction is accurate, efficient, and robust to realistic noise sources, providing a lightweight and experimentally feasible validation method for analog quantum simulators and hybrid architectures.

\section{Results}\label{sec:results}

\begin{figure*}
    \centering
    \includegraphics[width=\textwidth]{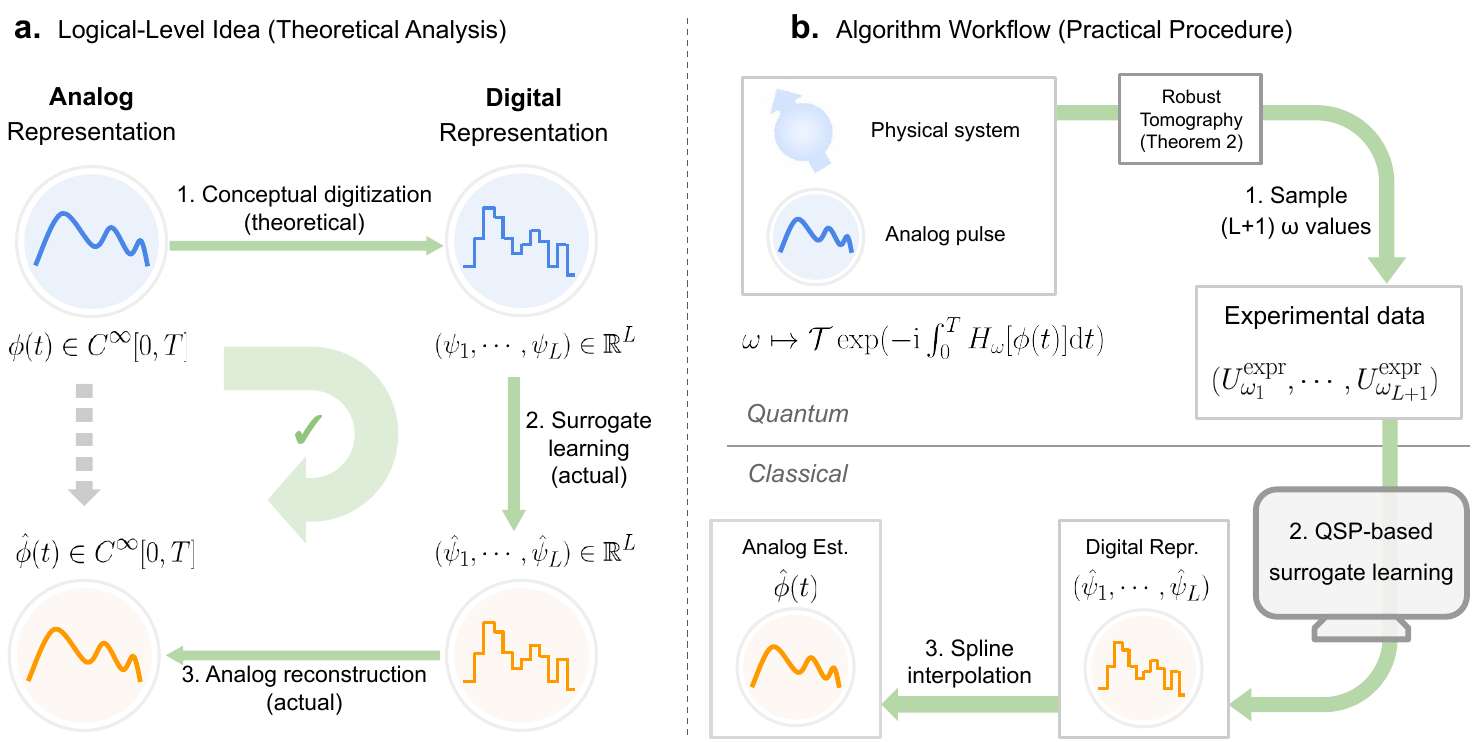}
    \caption{Overview of the proposed approach for learning smooth control pulses. Blue denotes ground-truth information, and orange denotes learned estimators. (a) Logical-level idea for theoretical analysis. A smooth analog function is conceptually mapped to a piecewise-constant surrogate representation, followed by tractable learning and reconstruction through spline interpolation. (b) Practical algorithm workflow. Experimental unitaries are collected at sampled parameter values and processed using a QSP-based surrogate-learning method. Finally, spline interpolation yields the final analog pulse estimate.}
    \label{fig:main_figure}
\end{figure*}

We begin by formalizing the analog pulse learning problem and introducing a digital surrogate model that enables tractable estimation of continuous pulse trajectories. We then describe the theoretical analog-to-digital and actual digital-to-analog conversion procedures, provide rigorous bounds on learning accuracy, analyze statistical optimality, and establish robustness under realistic SPAM and decoherence errors. Together, these results demonstrate a provably accurate and experimentally feasible end-to-end methodology for reconstructing smooth analog pulses from finite noisy measurements.

\subsection*{Quantum System Driven by an Analog Control Pulse}

Many universal quantum computing platforms, including superconducting qubits, neutral atoms, and trapped ions, implement dynamics through time-dependent Hamiltonians whose amplitudes or phases are programmed via pulse functions. In such settings, the quantum device is accessed only through end-to-end evolution over a fixed time interval, while the underlying control waveform remains hidden inside the time-ordered propagator. This motivates a general formulation of the analog pulse-learning problem that does not rely on any particular physical model, dimension, or parametric pulse family.

\begin{problem*}[Characterizing analog quantum pulses]
Let $H_{\omega}[\phi(t)]$ be a time-dependent Hamiltonian whose instantaneous form is determined by a smooth, unknown control pulse $\phi(t)$ on $[0,T]$, and let $\omega$ denote a set of experimentally tunable parameters (e.g., drive amplitudes, detunings, or global scaling factors). Given access only to the resulting time-evolution operator
\begin{equation}
U(T,0;\omega) = \mathcal{T}\exp\left(-\I\int_0^T H_{\omega}[\phi(s)]\ud s\right),
\end{equation}
or to expectation values of physical observables after this evolution, how can one reconstruct the underlying analog pulse $\phi(t)$ with bounded error using only a finite number of experiments? 
\end{problem*}

This formulation is deliberately non-parametric: we assume no particular functional form for $\phi(t)$ beyond smoothness, and make no assumption about the structure of $H_{\omega}$ other than that its dependence on $\phi(t)$ is known. Such generality indicates that the metrology could in principle be adapted to multi-qubit systems when the accessible inputs and measurements are confined to appropriate symmetry subspaces.

Having introduced the general problem, we now illustrate the key ideas in the simplest nontrivial setting: a driven two-level system. This model provides a clean stage for analysis while capturing the essential challenges of analog pulse learning. In this case, we consider the family of Hamiltonians
\begin{equation}\label{eqn:single_qubit_H_model}
    H_\omega[\phi(t)] = \omega \left(\cos \phi(t) X + \sin \phi(t) Y\right),
\end{equation}
where $\omega$ is an experimentally tunable drive amplitude and $\phi(t)$ is the analog pulse to be learned. The goal remains to recover $\phi(t)$ using only queries to the black-box propagator $U(T,0;\omega)$, without the ability to interrupt the evolution mid-pulse. This two-level model will serve as a running example for the analog-to-digital conversion, surrogate-model learning, and reconstruction procedures developed in the subsequent subsections. Furthermore, this Hamiltonian is commonplace in a variety of analog quantum simulators, e.g., neutral atoms \cite{wurtz2023aquilaqueras256qubitneutralatom} and ion traps~\cite{barreiro2011open,davoudi2020towards}

\subsection*{Learning Analog Pulse from Time Evolution}

Learning general analog pulses is challenging due to the complex functional dependence in time evolution. To explain our method, we start with a simple case: piecewise constant pulse functions. This setting highlights the kind of structure we aim to exploit, where the full-time evolution can be split into piecewise-learnable digital surrogates over short time intervals. We use this case to illustrate the core idea behind our method before generalizing to smooth analog pulses.

When the pulse function is piecewise constant, the propagator can be written as the product of time-discretized operators. Consider an equidistant time partition $t_j := j \tau, \tau := T / L$, and a piecewise constant pulse function:
\begin{equation}
    \phi(t) = \psi_j \in \RR,\ \text{when } (j-1) \tau \le t < j \tau,\ j = 1, \cdots, L.
\end{equation}
Then, the time-dependent Hamiltonian evolution can be explicitly characterized by combining multiple constant-pulse-driven pieces
\begin{equation}
    U(T, 0; \omega) = e^{- \I \tau H_\omega[\psi_L]} \cdots e^{- \I \tau H_\omega[\psi_2]} e^{- \I \tau H_\omega[\psi_1]}.
\end{equation}
This special class of pulse function coincides with 
many well-known variational digital quantum algorithms~\cite{farhi2014quantum,peruzzo_variational_2014}, whose parameters can be learned by matching experimentally observed data via optimization-based approaches. Furthermore, similar digital gate sequence models are also learned through the angle of gate learning. There exists a plethora of well-studied digital gate learning techniques, including randomized benchmarking~\cite{PhysRevA.77.012307}, gate set tomography~\cite{nielsen2021gate}, and phase estimation-based approaches~\cite{neill_accurately_2021,arute_observation_2020,kimmel,dong2025optimal}.

Given the success established for the characterization of digital gates, an intriguing question arises:
\begin{center}
    Can we characterize an \textit{analog} pulse via a learnable digital surrogate, i.e., detouring through the commutative diagram in \cref{fig:main_figure}(a)?
\end{center}

In this paper, we provide an affirmative answer to this question. To bridge the gap between the underlying analog evolution and its digital surrogate model, we propose logical-level analog-to-digital and digital-to-analog converters. These converters ensure the validity of surrogate model learning by providing a bounded-error, invertible mapping between the analog and digital representations of smooth control pulses. For surrogate model learning, we introduce a direct algebraic method that achieves high efficiency without relying on any black-box optimization. We further analyze its estimation variance and statistical optimality through Fisher information analysis. By performing a detailed analysis of these components and leveraging the structural properties of our design, we establish rigorous bounds that guarantee the method’s performance in terms of small and controllable bias and variance. The algorithms are presented in the remainder of this section.

\subsection*{QSP-Based Digital Surrogate Model Learning}

In the single-qubit case of \cref{eqn:single_qubit_H_model}, the short-time Hamiltonian evolution driven by a constant pulse amplitude can be expressed as
\begin{equation}\label{eq:unitaryVthetapsi}
\begin{split}
    V(\theta, \psi) &= e^{-\I \tau H_\omega[\psi]} = e^{- \I \theta (\cos(\psi) X + \sin(\psi)Y)}\\
    &= e^{- \I \psi/2 Z}   e^{- \I \theta X}   e^{\I \psi/2 Z},
\end{split}
\end{equation}
where $\theta := \omega \tau = \omega T / L$ is a tunable parameter. Assembling these short-time segments into a long-time evolution yields an interleaved sequence of $X$- and $Z$-rotations:
\begin{equation}\label{eqn:QSP_surrogate_model}
    W(\theta, \Psi) = V(\psi_L, \theta) \cdots V(\theta, \psi_1) \approx U(T, 0; \omega),
\end{equation}
where the approximation occurs due to Trotterization error.
The surrogate model $W(\theta, \Psi)$ coincides with the structure of \emph{quantum signal processing} (QSP), a powerful quantum algorithmic primitive that underlies many key quantum algorithms~\cite{chuang2021grand}.

In the context of analog control, this model corresponds to truncating the Magnus series at first order. Our analysis in \cref{sec:app:surrogate_model} shows that the associated truncation error scales as $\Or(\tau)$. Importantly, this first-order approximation avoids introducing higher-order commutator terms, which would otherwise result in complicated nonlinear dependencies on $\theta$ and $\phi$. This simplicity is crucial for efficient post-processing of experimental data.

The computation of phase factors $(\psi_1, \cdots, \psi_L)$ has been extensively studied in idealized numerical settings where only floating-point errors are considered. Here, we adapt a classical phase-computation method from Refs.~\cite{Haah2019,WangDongLin2022} to learn the digital representation of control pulses from experimental data using Fourier analysis. Since experimental measurements are collected with a finite number of shots, the data are inevitably affected by statistical fluctuations. To rigorously quantify the resulting estimation variance in the discretized pulse, we provide the first semi-analytical analysis of the QSP model from a statistics and learning perspective. The detailed analytical bounds are presented in the Methods section and in \cref{sec:app_Fourier_analysis_post_processing}.

A key distinction between our setting and conventional QSP phase-factor computation arises from a fundamental limitation of the Magnus expansion, which may fail to converge when $\theta \ge \pi$~\cite{moan2008convergence}. This constraint prevents direct sampling of $\theta$ values across the full unit circle. That is, we cannot experimentally query $U(T, 0; \omega)$ for arbitrary values of $\omega$, lest our digital surrogate will deviate greatly from the propagator. The missing information leads to an ill-posed estimation problem with poor numerical and statistical stability (see the Methods section and \cref{sec:fisher_info}). To overcome this challenge, we exploit the intrinsic symmetry of the surrogate model and propose a symmetry-based data augmentation method. This technique allows samples in the first quadrant, $\theta \in [0, \pi/2]$, to be extended to the full range $[0, 2\pi]$, thereby enabling subsequent analyses such as the fast Fourier transformation (FFT).

\subsection*{Analog Pulse Reconstruction via Spline Interpolation}

We now turn to the problem of reconstructing the continuous pulse function $\phi(t)$ from experimentally learned discretized data. Given a set of point values $(\psi_1, \ldots, \psi_L) \in \RR^L$ obtained from experiments, our goal is to recover the underlying analog signal $\phi(t)$.

To accurately reconstruct the pulse, we apply polynomial interpolation techniques upon the discretized data. Direct interpolation of discrete data using global methods such as Lagrange interpolation can suffer from instability, known as the Runge phenomenon~\cite{kress2012numerical}, which leads to large deviations in the reconstructed pulse. To avoid this, we employ spline-based interpolation methods, which provide stable and accurate reconstruction from discrete point data.

The QSP-based learning procedure yields discretized data that approximate the averaged pulse values with first-order error, since the surrogate model is based on a linear approximation. Interpolating the data at the midpoints $\{((j - 1/2)\tau, \psi_j) : j = 1, \ldots, L\}$ reconstructs the original pulse function with a first-order error scaling as $\Or(T / L)$. This reconstruction accuracy can be further improved to second order via Richardson extrapolation. By combining reconstructed pulse functions derived from two resolutions, $L$ and $2L$, the leading-order bias can be canceled, yielding an overall second-order systematic error while preserving the same estimation variance. This establishes an end-to-end performance guarantee for the reconstruction pipeline.

To simplify the presentation of theoretical bounds, we assume the pulse function $\phi$ belongs to a Bernstein-type class $\mathscr{P}_\beta([0, T])$ with some smoothness parameter $\beta$ (see \cref{sec:app_Bernstein_pulse}). Specifically, for some smoothness parameter $\beta > 0$, its derivatives satisfy $\norm{\phi^{(k)}} \le \Or(\beta^k)$ for all $k \in \NN$. This class is broad and includes, for instance, band-limited and polynomial pulse functions. 

\begin{thm}[End-to-end reconstruction performance]\label{thm:end-to-end-algorithm}
    Let $\phi \in \mathscr{P}_\beta([0, T])$, and let $L, M$ be positive integers. Suppose $M$ is the number of measurement shots per experiment, and that there are $L + 1$ experiments implementing Hamiltonian magnitudes $\omega_j = \frac{(2j+1)\pi}{4L+4}$ for $j = 0, 1, \ldots, L$. Let $\mc{I}^\circ = [1/L, T - 1/L]$ denote the interior region. Let $\hat\phi(t)$ denote the reconstructed pulse obtained from the end-to-end pipeline in \cref{fig:main_figure}(b). Then, there exist constants $C_1, C_2$ such that
    \begin{equation}
        \sup_{t \in \mc{I}^\circ} \mathbb{E}\left(\abs{\hat\phi(t) - \phi(t)}\right) \le C_1 \frac{\beta^2 T^2}{L^2} + C_2 \frac{1}{\sqrt{M}}.
    \end{equation}
    At the boundaries $t \in [0, 1/L] \cup [T - 1/L, T]$, the reconstruction error scales as $\Or(\beta T / L + 1 / \sqrt{M})$.
\end{thm}

The proof is outlined in \cref{sec:app:proof_main_thm}, where we also state the performance guarantee in a stronger error metric. The second term in the bound arises from statistical uncertainty due to finite measurement shots, while the first term corresponds to the systematic error of the end-to-end method. The first error term comes from the linearized QSP-based surrogate model. Yet, the simplicity of the QSP-based model allows fast and direct computation of discretized data but introduces a tradeoff in accuracy. In addition to the rigorous theoretical analysis, we numerically confirm the $L$-scaling of the bias, as shown in \cref{fig:scaling_L}.

Our QSP-based approach avoids any black-box optimization, relying instead on direct algebraic computation. It makes our QSP-based method both transparent and computationally efficient. Despite the first-order model misspecification in the QSP surrogate, we demonstrate that the overall systematic bias of the end-to-end reconstruction can be reduced to second order via Richardson extrapolation. Details of this bias reduction are provided in the following subsections and in \cref{sec:ho_bias_reduction}, and we justify the optimality of this second-order scaling in the Methods section.

\begin{figure*}
    \centering
    \includegraphics[width = .75\textwidth]{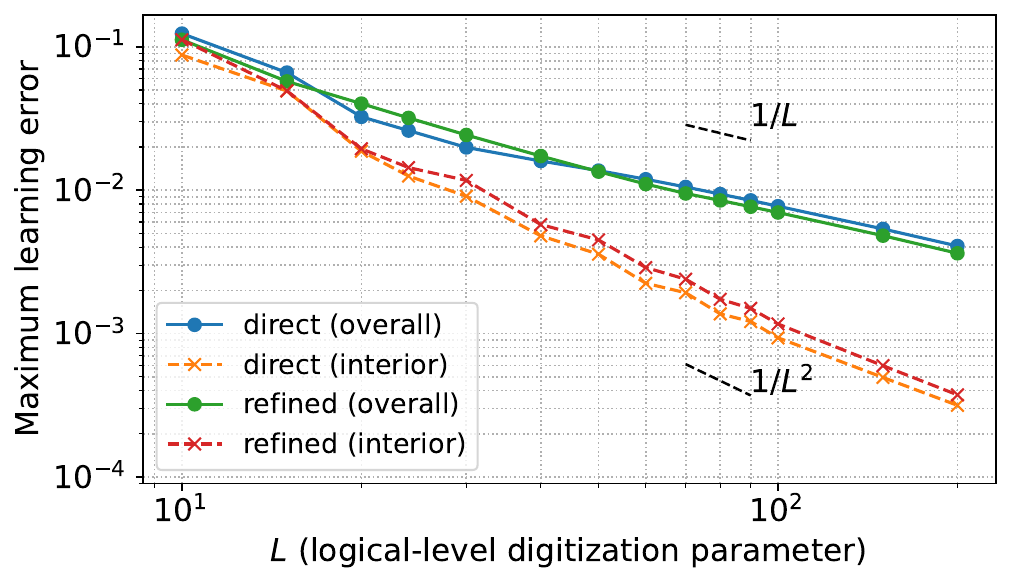}
    \caption{Scaling of the overall estimation bias. We fix the pulse function to $\phi(t) = \sin(3\pi t)$ with $T = 1$. To isolate the scaling of the estimation bias, we suppress sampling error in the numerical test by setting $M = \infty$. We report the maximal estimation error over the full interval $[0, T]$ (solid lines) and over the interior subinterval $[0.1T, 0.9T]$ (dashed lines). Two spline-based approaches are used, as outlined in \cref{sec:app_midpoint_Hermite}. The ``direct'' method identifies the raw data $\psi_j$ as midpoint values, while the ``refined'' method first applies a de-averaging filter to the raw data. The two methods perform similarly, indicating that higher-order digital-to-analog schemes offer little benefit because accuracy is limited by upstream steps such as the linearized QSP-based surrogate model. The interior reconstruction error is significantly reduced when using large $L$.}
    \label{fig:scaling_L}
\end{figure*}

\subsection*{Optimality of Our Analog Pulse Learning Method}

\cref{thm:end-to-end-algorithm} establishes the bias and variance scaling of our analog pulse learning method. A remaining question is whether these scalings are information-theoretically optimal. The efficiency and accuracy of any statistical estimation protocol are fundamentally determined by the amount of information that can be extracted from experimental data. In this subsection, we analyze the optimality of our method in both bias and variance scaling, as characterized in \cref{thm:end-to-end-algorithm}.

In the Methods section and \cref{sec:app_analytical_properties}, we show that the full nonlinearity of short-time evolution can, in principle, be exploited to improve the estimation accuracy of the averaged pulse values over short-time intervals. However, as shown in \cref{thm:bounded_bias}, this improvement is ultimately limited to second order. Although our QSP-based surrogate model captures only linear-order information, the application of Richardson extrapolation in post-processing effectively recovers the leading nonlinear correction. Consequently, the second-order $L$-dependence in \cref{thm:end-to-end-algorithm} is  optimal.

Turning to the variance analysis, the Fisher information of the digital surrogate model reveals a distinctive structure in the attainable estimation variance. Specifically, the boundary modes exhibit the smallest possible variance, $\mathrm{Var}(\hat{\phi}_1), \mathrm{Var}(\hat{\phi}_L) = \Or(1/(ML))$, whereas the central modes experience larger variance, with $\mathrm{Var}(\hat{\phi}_{\lceil L/2 \rceil}) = \Or(1/M)$. These statistical features are faithfully captured by our QSP-based surrogate learning method, as shown in \cref{thm:qsp_estimation_variance}. Therefore, upon taking the square root, the resulting standard deviation scaling $\Or(1/\sqrt{M})$ in \cref{thm:end-to-end-algorithm} is also optimal.

It is worth noting that this variance scaling appears worse than the conventional Standard Quantum Limit (SQL). However, this distinction arises from the structured statistical correlations inherent in our problem. The discretized pulse values are not independent measurements. Instead, they are derived from a single smooth control pulse evaluated over distinct but causally ordered intervals. As shown in \cref{sec:fisher_info}, this leads to a strong positive correlation in the estimation problem, which effectively reduces the number of independent degrees of freedom. Consequently, the achievable precision improvement is constrained with increasing $L$. As a result, increasing the number of digitized segments does not proportionally reduce the statistical uncertainty. A detailed analysis of this correlation effect and its connection to the Fisher information structure is provided in \cref{sec:fisher_info}.

\subsection*{Robustness against Realistic Quantum Errors}

In practical quantum experiments, decoherence and   SPAM errors inevitably affect the system’s evolution and the accuracy of reconstructed unitaries. To enhance the robustness of our framework, we introduce a \textit{robust unitary tomography} procedure as a data preprocessing subroutine. As analyzed in detail in \cref{sec:app:robust_preproc_tomography}, both theoretical justification and numerical evidence show that this method significantly improves the error resilience of the refined unitary data.

Our approach builds upon the standard unitary tomography protocol~\cite{chuang_prescription_1997}, complemented by two key enhancements: (1) a \emph{sandwich transformation} using a reference experiment, and (2) a \emph{projection via polar decomposition}.

The sandwich transformation suppresses depolarizing noise by preconditioning the tomographically reconstructed unitaries. The subsequent projection step enforces unitarity and removes the antisymmetric component of SPAM-induced errors in the generator representation thanks to the structure of polar decomposition. The remaining symmetric component arises from causality, specifically, the non-interchangeability of preparation and measurement processes. If one can perform an additional mirrored experiment with the roles of state preparation and measurement reversed, this residual symmetric component can also be canceled. Without such experimental access, the best achievable reconstruction error typically scales linearly with the magnitude of SPAM errors. In our method, the unitary reconstruction error also scales linearly in general, but improves to quadratic scaling when the symmetric SPAM component vanishes. In addition, depolarizing errors can be further mitigated by increasing the number of measurement shots. These properties jointly demonstrate the robustness and effectiveness of our approach. The formal result is summarized in the following theorem.

\begin{thm}
    Let $\alpha$ denote the fidelity of the depolarizing channel, and let $\mf{g}_S$ and $\mf{g}_M$ be the infinitesimal generators associated with state-preparation and measurement errors, respectively. The magnitude of SPAM error is $\delta = \max\{\norm{\mf{g}_S}, \norm{\mf{g}_M}\}$. Define the SPAM generator difference as $\Delta_\mathrm{SPAM} = \frac{1}{2}(\mf{g}_M - \mf{g}_S)$. Then, using $M = \Or(\delta^{-2}\alpha^{-2})$ measurement shots, the target unitary can be reconstructed via the robust unitary tomography procedure with reconstruction error at most $\Or(\delta)$. Moreover, if $\Delta_\mathrm{SPAM}$ is symmetric (i.e. $\Delta_\mathrm{SPAM}^\top = \Delta_\mathrm{SPAM}$), the reconstruction error further improves to $\Or(\delta^2)$.
\end{thm}

The complete theoretical proofs and numerical validations are presented in \cref{sec:app:robust_preproc_tomography}. The incorporation of a tomography subroutine that is intrinsically resilient to SPAM and depolarizing noise is pivotal for preserving the overall robustness and reliability of our pulse-learning algorithm.~(see Fig.~\ref{fig:main_figure} for its role in overall workflow).

\subsection*{Numerical Results of Pulse Calibration}

In quantum experiments, analog control pulses are unavoidably influenced by stochastic control errors~\cite{NiuBoixoSmelyanskiyEtAl2019}, control crosstalk~\cite{sheldon2016procedure,abrams2019methods,fan2025calibrating}, time-dependent control drift~\cite{manuel-endres-2024benchmarking}, and sampling noise. These imperfections distort the ideal pulse shape, introducing deviations from the intended smooth profile and complicating accurate characterization. As we demonstrate below, our pulse-learning methods are capable of recovering these actuated control imperfections on top of the ideal target waveform, providing an effective tool for verification and calibration. This capability serves as a critical test of both the practical feasibility and the high accuracy achievable with our approach.

In this subsection, we numerically evaluate the end-to-end performance of our method in the context of pulse calibration. Specifically, we consider three types of ideal control pulses that are representative of  practically relavant pulse families: linear ($\phi(t) = t$), sinusoidal ($\phi(t) = \sin(2\pi t)$), and biharmonic ($\phi(t) = \tfrac{1}{2}[\sin(2\pi t) + \sin(4\pi t)]$). Following the procedures in Ref.~\cite{NiuBoixoSmelyanskiyEtAl2019} and in \cref{sec:app:numerical_results_perturbation}, we generate perturbed control pulses by introducing random control errors to these ideal pulses. These perturbed pulses are referred to as \emph{actual pulses}. In our end-to-end simulations, we further emulate depolarizing noise and SPAM errors, and apply our robust unitary tomography using a finite number of measurement samples.  

\Cref{fig:end-to-end-numerics} shows the numerical results of pulse calibration under various realistic noise conditions. Our method consistently achieves robust reconstruction of the actual implemented pulses. With SPAM and sampling errors on the order of $0.01$, the pointwise reconstruction error remains tightly bounded within the same magnitude regardless of the pulse duration. Even when the smoothness of the actuated pulses is degraded due to control level noise, the proposed approach performs reliably.  

Two key features contribute to the effectiveness and accuracy of our method. First, it requires no prior knowledge of the underlying pulse shape. Second, it is a fully interpretable approach that relies only on direct algebraic operations rather than black-box iterative optimization. These properties jointly ensure the efficiency, accuracy, and robustness of our pulse characterization framework.

\begin{figure*}
    \centering
    \includegraphics[width=\textwidth]{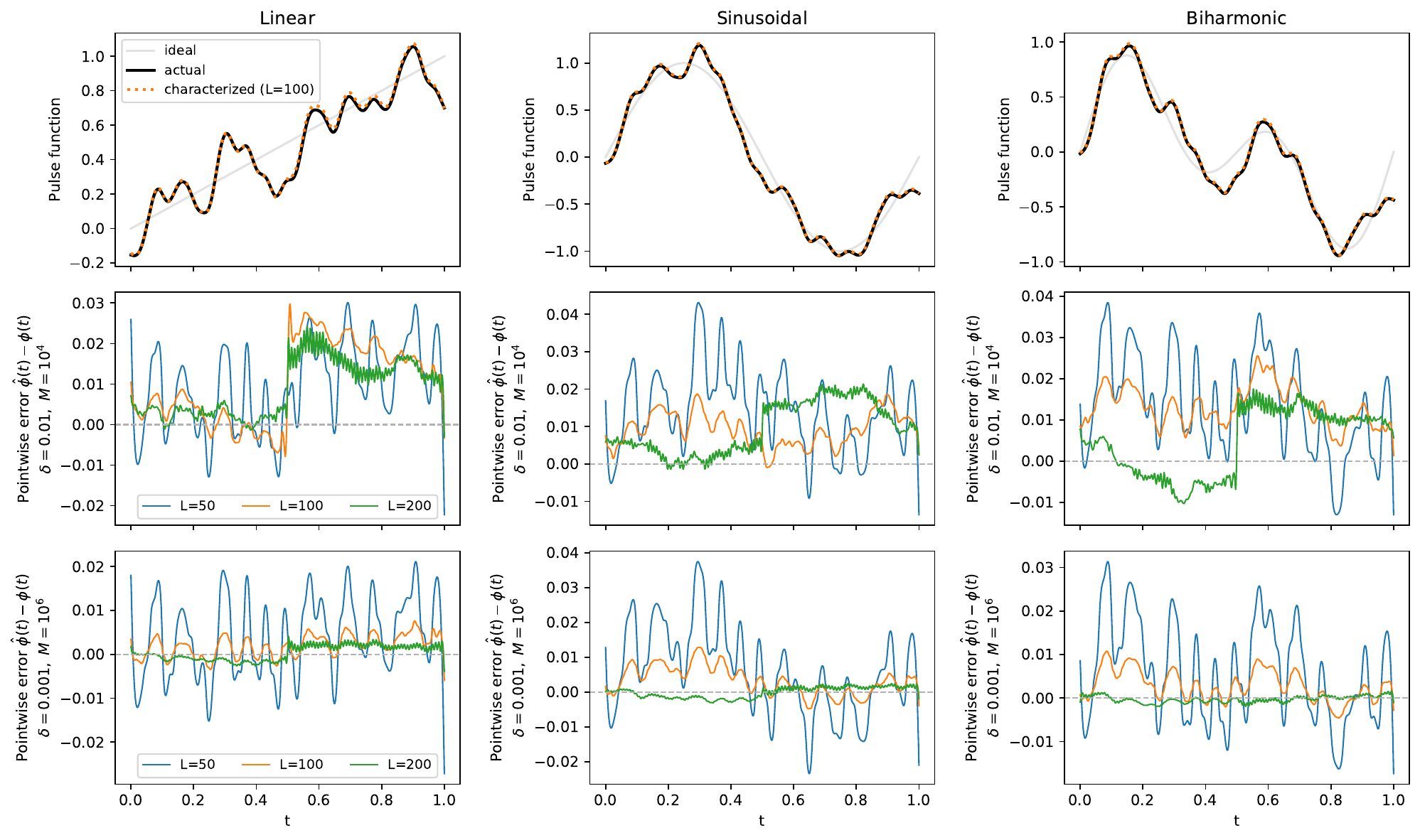}
    \caption{Performance of end-to-end pulse learning. Sources of error include pulse perturbations, depolarizing noise with fidelity $\alpha$, random SPAM error of magnitude $\delta$, and sampling error from a finite number of measurement shots $M$. We fix $\alpha = 0.9$ in all simulations. Top row: pulse shapes. The ``actual'' pulse is obtained by perturbing the ``ideal'' pulse, and serves as the target for learning. The dotted curve denotes the pulse characterized by our proposed method using the configuration of the middle row. Other rows: pointwise reconstruction error between the characterized and actual pulses. In the middle row, we set $\delta = 0.01$ and $M = 10^4$, where these errors become the bottleneck of the overall reconstruction accuracy, even for large values of $L$. In the bottom row, we reduce the error magnitudes to $\delta = 0.001$ and $M = 10^6$, and the overall reconstruction error is then further improved by increasing $L$.}
    \label{fig:end-to-end-numerics}
\end{figure*}

\subsection*{Applications in Pulse Calibration}

There are two primary applications of our protocol in multi-qubit settings: calibration of \textit{classical} single-qubit control crosstalk and calibration of \textit{quantum} crosstalk.
In superconducting qubit systems, \textit{classical control crosstalk}~\cite{sheldon2016procedure,abrams2019methods,fan2025calibrating} is a dominant source of operational error, arising from unwanted leakage in microwave drives to neighboring qubits. Conceptually, an ideal single-qubit $X$-$Y$ control pulse on a certain qubit will generate additional drive amplitudes and phase profiles on neighbors that depend on both the hardware connectivity and the original control waveform. Rather than explicitly expanding each induced Hamiltonian term, we focus on the underlying functional dependencies: our method can be used to learns both how the leaked amplitudes scale with the intended drive $\omega_j$, and how the distorted pulse shapes relate to the target waveform $\phi_{\text{target}}(t)$. Our protocol can be used to reconstruct  crosstalk-induced pulse distortions directly at the waveform level. To the best of our knowledge, this has not been previously demonstrated for microwave control in superconducting architectures. This enables quantitative, pulse-resolved calibration of crosstalk with high efficiency and accuracy.

Our approach is also potentially applicable to \textit{quantum} crosstalk, where multi-qubit interactions produce time-dependent Hamiltonian amplitudes. Leveraging the underlying physical symmetry, the quantum-crosstalk structure can be effectively reduced to two-level subsystems~\cite{liu2025optimal}. With properly designed basis transformation and the associated unitary tomography, effective parameters can be learned by our protocol. For instance, \cite{liu2025optimal} shows that for the Mølmer-Sørensen interaction in trapped ions, collective spin-spin couplings and local fields can map onto the phase and amplitude parameters of the two-level surrogate model. For these systems where effect of many-body dynamics can be reduced to a two-level subspace, our protocol provides a means to calibrate time-dependent Hamiltonian amplitudes   at the pulse level. Such capability is essential for improving overall fidelity of analog quantum algorithms,  and mitigating coherent errors in error-correcting architectures, where even small crosstalk-induced distortions can significantly limit system performance.

\section{Discussion}\label{sec: discussion}

In this work, we introduced a rigorous and efficient framework for in situ learning of analog quantum control pulses through a logical-level analog-to-digital-to-analog paradigm. Our method establishes an analytical bridge between continuous Hamiltonian control and discrete algorithmic representations, embedding pulse learning within the QSP formalism. By mapping continuous pulse dynamics to a set of discrete, learnable surrogate parameters, and then reconstructing the analog waveform from the learned digital data, we demonstrated that the entire algorithm operates with bounded errors with mathematically rigorous guarantee. Crucially, our digital surrogate model replaces heuristic black-box optimization with an analytic, Fourier-structured analysis, yielding stable and transparent recovery of control parameters directly from experimental data. Without ancilla qubits, mid-circuit access, or interruption of hardware dynamics, our protocol operates entirely in situ and achieves highly accurate reconstruction of continuous control trajectories while maintaining robustness to realistic experimental imperfections.

An additional advance lies in our development of a SPAM- and depolarizing-noise-resilient tomography subroutine, which ensures the learned pulse parameters remain accurate under experimental noise typical of near-term hardware. Together, these contributions unify digital and analog control analysis under a single analytical framework, enabling digital gate-based tools such as QSP and Fisher-information analysis to be applied to analog calibration and validation. The resulting formalism establishes a scalable, model-agnostic route for characterizing and certifying analog quantum simulators subject to drift, bandwidth constraints, and contextual errors. Looking forward, this paradigm opens a number of promising directions, including the generalization of our method to multi-qubit analog Hamiltonian characterization, as well as the identification of pulse-level analog errors caused by control crosstalk in superconducting, trapped-ion, and Rydberg platforms. More broadly, by merging analytical digital learning theory with continuous-time control physics, this work studies a new frontier for mixed programming and calibration of quantum devices, where digital compilation and analog pulse shaping can be co-optimized within a single mathematically rigorous and noise-robust learning architecture.

\section{Methods}\label{sec:methods}

\subsection*{Learning Discretized Pulse Values via a QSP-Based Method}

Consider the surrogate model in \cref{eqn:QSP_surrogate_model}. According to the classical theory of quantum signal processing~\cite{Haah2019,WangDongLin2022}, this surrogate admits a matrix-valued Fourier expansion in $\theta := \omega \tau$:
\begin{equation}
\begin{split}
    W(\theta, \Psi) &= V(\theta, \psi_L) \cdots V(\theta, \psi_1) \\
    &= \sum_{j = -L}^L C_j e^{\I j \theta}, \quad C_j \in \CC^{2 \times 2},
\end{split}
\end{equation}
where the matrix coefficients $\{C_j\}$ are closely related to the phase factors $\{\psi_j\}$.

By multiplying $W(\theta, \Psi)$ with a suitably chosen $\theta$-dependent unitary, we can compensate for the boundary mode $V(\theta, \psi_1)$ and cancel the highest-degree terms in the Fourier expansion. This reduces the problem to a smaller instance involving only $(L-1)$ phase factors $(\psi_2, \ldots, \psi_L)$. As shown in Ref.~\cite{Haah2019} and detailed in \cref{sec:app_Fourier_analysis_post_processing}, this annihilation procedure yields the following relation:
\begin{equation}\label{eqn:Fourier_learning_P}
    \frac{C_L^\dagger C_L}{\Tr(C_L^\dagger C_L)} = \frac{1}{2}
\begin{pmatrix}
1 & -e^{-\I\psi_1} \\
- e^{\I\psi_1} & 1
\end{pmatrix}.
\end{equation}
Thus, $\psi_1$ can be estimated directly from the leading Fourier coefficient $C_L$. By subsequently applying the estimation to $W(\theta, \Psi) V^{-1}(\theta, \hat{\psi}_1)$, the remaining phase factors can be inferred sequentially from the set of Fourier coefficients $\{C_j\}_{j=-L}^L$. This process constitutes a direct algebraic learning method that avoids any black-box optimization. As a result, it is numerically stable, computationally efficient, and fully transparent. Moreover, processing data in Fourier space naturally mitigates experimental noise, since the Fourier transform acts as a contraction mapping. The rich Fourier structure also enables further noise-mitigation strategies, such as those proposed in~\cite{dong2025optimal}.

While this phase-factor computation method was originally designed for idealized numerical settings without quantum noise or hardware imperfections, we extend it to experimental data by introducing a symmetry-based data augmentation technique that satisfies the convergence requirements of the Magnus expansion. In addition, we perform a rigorous semi-analytical analysis of the estimation variance induced by sampling noise. These enhancements allow the originally idealized QSP algorithm to be effectively adapted to realistic experimental conditions. According to the analysis in \cref{sec:app_Fourier_analysis_post_processing}, the propagation of measurement noise through this learning procedure is characterized by the following theorem.

\begin{thm}\label{thm:qsp_estimation_variance}
    Suppose the noise variance in the experimental data is $\sigma^2 = \Or(1 / M)$, where $M$ denotes the number of measurement repetitions per experiment. When $\theta$-values are sampled at $N$ evenly spaced midpoints in $[0, \pi/2]$, the noise-induced estimation variance in the phase-factor learning algorithm satisfies the recurrence
    \begin{equation*}
    \begin{split}
        & \mathrm{Var}(\hat{\psi}_{j + 1}) \le \rho_j \, \mathrm{Var}(\hat{\psi}_{j}) + \Or\!\left(\frac{\alpha_j}{M N}\right), \quad j = 1, \ldots, \lceil L/2 \rceil,\\
        & \mathrm{Var}(\hat{\psi}_{j - 1}) \le \rho_j \, \mathrm{Var}(\hat{\psi}_{j}) + \Or\!\left(\frac{\alpha_j}{M N}\right), \quad j = L, \ldots, \lceil L/2 \rceil,
    \end{split}
    \end{equation*}
    where $\rho_j$ and $\alpha_j$ are constants determined by the smoothness of the analog pulse.
\end{thm}

When the pulse function is sufficiently smooth, the semi-analytical results in \cref{sec:app_Fourier_analysis_post_processing} show that the recurrence is well-conditioned with $\rho_j \approx 1$ and $\alpha_j = \Or(1)$. Consequently, the noise grows approximately linearly toward the center, yielding 
\[
\max_j \mathrm{Var}(\hat{\psi}_j) = \Or(1 / M), \quad \min_j \mathrm{Var}(\hat{\psi}_j) = \Or(1 / (M L)).
\]
This variance scaling is proven optimal as it saturates the Cram\'{e}r-Rao lower bound derived from the Fisher information (see \cref{sec:fisher_info}).

\subsection*{Optimality of Achievable Variance }

The statistical optimality of an estimation protocol can be characterized using the Fisher information matrix (FIM) and the Cram\'er-Rao lower bound (CRLB). When the quantum experimental settings are fixed, the FIM quantifies the amount of information that can be extracted from noisy measurements, and the CRLB establishes the best achievable precision for estimation. Together, they reveal the theoretical limits of post-processing accuracy.

In \cref{sec:fisher_info}, we show that the Fisher information matrix for the digital surrogate model admits the following compact form:
\begin{equation}
    \mc{F}_{ij}(\Psi) = \frac{2M(L+1)}{\pi} \int_0^{\pi/2} 
    \Re\left(
        \braket{0 | 
        \frac{\partial W}{\partial \psi_i} 
        \frac{\partial W^\dagger}{\partial \psi_j} 
        | 0}
    \right) 
    \ud \theta.
\end{equation}

To gain analytical insight, we examine a representative class of surrogate problems in which all phase factors are identical. This translational symmetry allows the Fisher matrix to be expressed in a solvable form. Specifically, the resulting FIM takes a tridiagonal Toeplitz structure,
\begin{equation}
\mc{F}_{ij}^\mathrm{id} = \frac{M (L + 1)}{4} ( 2 \delta_{ij} - \delta_{i, j+1} - \delta_{i, j-1} ),
\end{equation}
whose inverse can be obtained in closed form~\cite{Kay1989}. From this expression, we derive the minimum achievable estimation variance given by the CRLB and identify the strong statistical correlations between neighboring phase factors. These correlations explain the observed saturation behavior in \cref{thm:end-to-end-algorithm}. That is, increasing the number of discretized segments $L$ does not further improve estimation precision because of the intrinsic correlation among estimators.

\subsection*{Optimality of Achievable   Bias}

We justify the optimality of our bias scaling by analyzing how accurately the discretized pulse values can be approximated from the short-time evolution generators.

We begin by expressing the full time-evolution operator $U(T, 0; \omega)$ as an exact product of short-time segments:
\begin{equation}\label{eqn:Trotterized_time_evolution}
    U(T, 0; \omega) = U(L \tau, (L - 1) \tau; \omega) \cdots U(\tau, 0; \omega),
\end{equation}
where each short-time segment $U(t_j, t_{j-1}; \omega)$ depends only on the pulse function within the local interval $[t_{j-1}, t_j]$.  
The key question is: how much information about the pulse can be extracted from each segment if we fully exploit its nonlinear structure?

For the single-qubit case, each segment can be expressed exactly in terms of Lie-algebra generators:
\begin{equation}
    U(j \tau, (j-1) \tau; \omega) = \exp(-\I (a_j X + b_j Y + c_j Z)),
\end{equation}
where the coefficients $a_j, b_j, c_j$ depend nonlinearly on $(\omega, \tau, \phi)$.  
Our QSP-based surrogate model corresponds to a first-order linearization of these coefficients with respect to $\omega$.  
A natural question follows: if we retained the full nonlinear dependence, how much additional information about the pulse could we gain?

In \cref{sec:app_analytical_properties}, we show that the following estimator can approximate the average pulse value with a remarkably small error:
\begin{equation}\label{eqn:data_error}
    \wt{\psi}_j := \arctan\!\left(\frac{b_j}{a_j}\right)
    \quad \text{s.t.} \quad
    \abs{\wt{\psi}_j - \bar{\phi}_j} \le \Or(\beta^2 (\beta^2 + \omega^2) \tau^4).
\end{equation}
The $\Or(\tau^4)$ scaling arises from the trigonometric structure of the Hamiltonian. Since statistical estimation requires nonzero Fisher information, it requires $\omega_{\max}\tau \approx \pi/2$. Consequently, this bound corresponds to a second-order optimal approximation when the full nonlinearity in the short-time evolution is utilized. We summarize this result below.

\begin{thm}\label{thm:bounded_bias}
    Suppose $U(T, 0; \omega)$ is discretized into $L$ short-time segments with step size $\tau = T / L$.  
    Let $\wt{\psi}_j$ be the estimator defined in \cref{eqn:data_error}, which exploits the full nonlinear dependence in the short-time evolution.  
    When learning from experimental data obtained with finite measurement shots, the systematic bias satisfies
    \begin{equation}
        \max_{j = 1, \ldots, L} \abs{\wt{\psi}_j - \bar{\phi}_j}
        \le \Or\!\left( \beta^2 T^2 / L^2 \right),
    \end{equation}
    where $\bar{\phi}_j = \frac{1}{\tau} \int_{(j - 1)\tau}^{j \tau} \phi(s) \, \mathrm{d}s$ is the average pulse value on each subinterval.
\end{thm}

Although our QSP-based surrogate model alone achieves only first-order accuracy, combining it with Richardson extrapolation allows the leading-order nonlinear effects to be captured through post-processing.  
As a result, our reconstruction achieves a second-order error scaling as in \cref{thm:end-to-end-algorithm}. This agreement with the fully nonlinear analysis confirms that the quadratic $L$-dependence in our result is arguably optimal.

 \bigskip

\noindent{\large \textbf{Data availability}}\\
All data presented in this work are visualized in the figures within the main text and the Supplementary Information file.\\

\noindent{\large \textbf{Code availability}}\\
The codes that support the finding are available at \url{https://github.com/dongsnaq/qsp_analog_pulse_characterization}. \\
 
\noindent {\large \textbf{Acknowledgments}}\\

C.K. thanks Shraddha Anand and Tom Manovitz for useful conversations and Fred Chong for his support. 
M. N. is
supported by the U.S. National Science Foundation grant CCF-2441912(CAREER), Air Force Office of Scientific Research under award number FA9550-25-1-0146, and   the U.S. Department of Energy,  Office of  Advanced Scientific Computing Research  under Award Number DE-SC0025430.
C.K. is funded in part by the STAQ project under award NSF Phy-232580; in part by the US Department of Energy Office of Advanced Scientific Computing Research, Accelerated 
Research for Quantum Computing Program; and in part by the NSF Quantum Leap Challenge Institute for Hybrid Quantum Architectures and Networks (NSF Award 2016136), in part by the NSF National Virtual Quantum Laboratory program, in part based upon work supported by the U.S. Department of Energy, Office of Science, National Quantum 
Information Science Research Centers, and in part by the Army Research Office under Grant Number W911NF-23-1-0077. The views and conclusions contained in this document are those of the authors and should not be interpreted as representing the official policies, either expressed or implied, of the U.S. Government. The U.S. Government is authorized to reproduce and distribute reprints for Government purposes notwithstanding any copyright notation herein.

 \bigskip
 
\noindent {\large \textbf{Author contributions}}\\

M. N. and Y. D. proposed the initial QSP-based estimation scheme. C. K. helped improve and implement the numerical simulations of the original scheme. Y. D. further developed the full algorithm and carried out the theoretical analysis and numerical simulations supporting the study. All authors contributed to the writing of the manuscript.

 \bigskip
 
\noindent{\large \textbf{Competing interests}}\\
\noindent The authors declare no competing interests.

\bibliographystyle{plain}

\newpage
\clearpage
\appendix
\onecolumngrid

%TC:ignore
 \begin{center}
     {\Large \bf \SP}
 \end{center}

\tableofcontents

\section{Notations and Preliminaries}
\subsection{Notations}
Throughout the paper, we use $\norm{A}$ and $\norm{A}_F$ to denote the matrix 2-norm and Frobenius norm, respectively. For a function $f$, $\norm{f}_\infty$ denotes its maximum norm over its domain $\mathrm{Dom}(f)$. We write $a \lesssim b$ to indicate that $a$ is bounded above by $b$ up to a small higher-order remainder, i.e., $a \le b(1 + o(1))$. In Supplementary Information, we use $h := T / L$ to denote the time step size.

\subsection{Bernstein-Type Pulse Functions}\label{sec:app_Bernstein_pulse}

In this work, we consider a class of smooth pulse function whose smoothness and derivative growth are controlled by a constant parameter. This bounded behavior mimics and generalizes that of classic bandlimited functions in signal processing according to Bernstein's inequality \cite{butzer2013basic}. Hence, it is referred to as the Bernstein class in the literature \cite{roytwarf1997bernstein}.

We remark that the regularity requirement on the target pulse function can be substantially relaxed: it suffices for $\phi \in C^k[0, T]$ to be $k$-times differentiable, rather than of Bernstein type. The smoothness parameter $k$ depends on the specific spline-based reconstruction scheme employed. To simplify the preconstants in the presented theoretical bounds, however, we assume the target pulse function is of Bernstein type.

\begin{defn}[Bernstein-type pulse]
Let $\beta, T > 0$ be two constant parameters. We define the class of Bernstein-type pulse functions as
\begin{equation}
    \mathscr{P}_\beta([0, T]) := \Bigl\{ \phi \in C^\infty([0,T]) : \exists C>0 \text{ s.t. } \sup_{t\in[0,T]}\abs{\phi^{(k)}(t)} \le C \beta^{k},
\ \forall k\in\mathbb{N} \Bigr\}.
\end{equation}
Unless otherwise noted, when a pulse function is defined on $[0, T]$, we use the simplified notation $\phi \in \mathscr{P}_\beta$ to indicate that it is a Bernstein-type pulse.
\end{defn}

This Bernstein-type class of pulse functions is very general. We discuss two instances in this class.

\vspace{0.5em}\noindent
\textbf{(1) Bandlimited pulse functions}
\vspace{0.5em}

A smooth pulse function $\phi \in C^\infty[0, T]$ is said to be bandlimited with bandwidth $\beta > 0$ if there exists a smooth extension $g \in C^\infty(\mathbb{R})$ such that
\begin{equation}\label{eqn:bandlimited_condition}
    \phi = g|_{[0, T]}, \quad \norm{g}_\infty \le \mc{O}(1), \text{ and }\quad  \hat{g}(\xi) = \int_{-\infty}^{\infty} g(t)  e^{-\I \xi t}  \ud t = 0,\ \forall  |\xi| > \beta.
\end{equation}
Turning derivatives to multiplication in the frequency domain, we have the standard Bernstein's inequalities for bandlimited pulse functions
\begin{equation}
    \sup_{t\in[0,T]}\abs{\phi^{(k)}(t)} \le \sup_{t\in \RR}\abs{g^{(k)}(t)} \le C \beta^{k}.
\end{equation}
Therefore, bandlimited functions are in the Bernstein-type pulse class $\mathscr{P}_\beta$.

In our simulations, when $\phi$ is expressed as a finite sum of sine and cosine modes, it satisfies the bandlimited condition \cref{eqn:bandlimited_condition}. Consequently, $\phi$ is both bandlimited and Bernstein-type.

\vspace{0.5em}\noindent
\textbf{(2) Polynomial pulse functions}
\vspace{0.5em}

Let $\phi \in \RR_d[t]$ be a polynomial pulse function of degree at most $d$. Suppose it is bounded on its domain $[0, T]$ as $\sup_{0 \le t \le T} \abs{\phi(t)} \le C$. We show that it is a Bernstein-type pulse function with $\beta = 2 d^2 / T$. 

We first perform a scale transformation so we have an equivalent polynomial function defined on $[-1, 1]$:
\begin{equation}
    x := \frac{2 t - T}{T},\ q(x) := \phi(T(x+1) / 2) \in \RR_d[x].
\end{equation}
Applying Markov brother's inequalities \cite{Markov1890} to the alternative polynomial, we have
\begin{equation}
    \sup_{-1 \le x \le 1} \abs{q^{(k)}(x)} \le d^{2k} \sup_{-1 \le x \le 1} \abs{q(x)}.
\end{equation}
Undoing the scale transformation, we have
\begin{equation}
    \sup_{0 \le t \le T} \abs{\phi^{(k)}(t)} \le C  \left(\frac{2 d^2}{T}\right)^k.
\end{equation}
Hence, polynomial pulse functions are Bernstein-type.

\section{Proof of \cref{thm:end-to-end-algorithm} and Related Results in Estimation Performance}\label{sec:app:proof_main_thm}
The technical results in the following sections are centered on proving \cref{thm:end-to-end-algorithm}. We provide a roadmap of the proof in this section.

\subsection{Performance metrics of pulse learning}\label{subsec:app:metrics}
We evaluate the estimation error using two performance metrics. In \cref{thm:end-to-end-algorithm} of the main text, the guarantee is stated in terms of the 
\emph{uniform mean error}, which keeps the presentation clean and focuses on the optimal scaling. For completeness, we also introduce a stronger notion of accuracy to establish a high-probability bound on the maximal deviation.

\begin{defn}[Performance metrics]
    Let $\hat{\phi}(t)$ be an estimator of the pulse function $\phi(t)$ defined on the interval 
    $[0,T]$. 
    Let $\mathcal{I}\subset[0,T]$ be the interval on which the estimation performance is evaluated. 
    We consider the following two metrics:
    \begin{enumerate}
        \item \emph{Uniform mean error}:
        \begin{equation}
            \sup_{t\in\mathcal{I}} 
            \mathbb{E}\left( |\hat{\phi}(t)-\phi(t)| \right)
            \le \epsilon.
        \end{equation}

        \item \emph{Uniform maximal error}:
        \begin{equation}
            \sup_{t\in\mathcal{I}} 
            |\hat{\phi}(t)-\phi(t)|
            \le \epsilon
            \quad\text{with high probability}.
        \end{equation}
    \end{enumerate}
\end{defn}

Assuming that $\hat{\phi}(t)$ is uniformly bounded, a high-probability 
uniform maximal error bound of the form $\mathbb{P}(\sup_{t\in\mathcal{I}} |\hat{\phi}(t)-\phi(t)| \le \epsilon) \ge 1 - \delta$ immediately implies a corresponding bound on the expected max-norm error. That is $\EE\sup_t \abs{\hat{\phi}(t) - \phi(t)} \le \epsilon + C \delta$ for some constant $C$. Consequently, the second error metric (uniform maximal error) is stronger than the uniform mean error. Intuitively, this follows from the basic inequality $\sup_t \EE\abs{\hat{\phi}(t) - \phi(t)} \le \EE\sup_t \abs{\hat{\phi}(t) - \phi(t)}$. 

\subsection{Structure of the estimator}\label{subsec:app:structure_estimator}
Before proving results, we restate the structure of the problem that we derived in the following sections.

In \cref{sec:app:reconstruct_analog}, we show the following decomposition to separate the analysis of bias and variance in the estimator:
\begin{equation}
    \hat{\phi}(t) = \mathbb{E}(\hat{\phi}(t)) + \zeta^\top(t) \vec{\varepsilon}.
\end{equation}
Here, $\zeta(t) \in \RR^L$ is a coefficient vector with $\norm{\zeta(t)}_1, \norm{\zeta(t)}_2 \le \Or(1)$, and $\vec{\varepsilon} \in \RR^L$ is the normal-distributed vector of data noise with zero mean and $\mathrm{Cov}(\vec{\epsilon}) \preceq \sigma^2 I$.
Then, we have the following bias-variance decomposition
\begin{equation}\label{eqn:bias_variance_decomposition}
    \hat{\phi}(t) - \phi(t) = \left(\mathbb{E}(\hat{\phi}(t)) - \phi(t)\right) + \zeta^\top(t) \vec{\varepsilon} =: \mathrm{bias}(t) + \zeta^\top(t) \vec{\varepsilon}.
\end{equation}
In \cref{sec:ho_bias_reduction}, we show that the bias term can be reduced to second order using Richardson extrapolation. Suppose $L$ is the algorithm parameter as in \cref{thm:end-to-end-algorithm}. \cref{thm:Richardson_extrapolation_estimation_error_bound} indicates that
\begin{equation}\label{eqn:bias_bound}
    \abs{\mathrm{bias}(t)}  \le C_1 \frac{\beta^2 T^2}{L^2}
\end{equation}
for some constant $C_1$. 

According to the analysis of Fisher information in \cref{sec:fisher_info}, we have
\begin{equation}\label{eqn:sigma_sq_bound}
    \sigma^2 \le \Or(1/M).
\end{equation}

These suffice to prove our main technical result.

\subsection{Proof of \cref{thm:end-to-end-algorithm}}\label{subsec:app:proof_main_theorem}
\begin{proof}
    The result follows the structure in \cref{subsec:app:structure_estimator}. 

    In terms of mean square error, the bias-variance decomposition in \cref{eqn:bias_variance_decomposition} becomes
    \begin{equation}
        \mathbb{E}\left( (\hat{\phi}(t) - \phi(t))^2 \right) = \mathrm{bias}^2(t) + \zeta^\top(t) \mathbb{E}(\vec{\varepsilon}^\top \vec{\varepsilon}) \zeta(t) \le \mathrm{bias}^2(t) + \sigma^2 \norm{\zeta(t)}_2^2 \le \left(C_1 \frac{\beta^2 T^2}{L^2}\right)^2 +  C_2^2 \frac{1}{M}.
    \end{equation}
    for some constant $C_2$. Here, the last expression uses \cref{eqn:bias_bound,eqn:sigma_sq_bound}. This bound holds for any interior point $t \in \mc{I}_\circ$.

    According to Jensen's inequality, we have
    \begin{equation}
        \mathbb{E}^2\abs{\hat{\phi}(t) - \phi(t)} \le \mathbb{E}\left( (\hat{\phi}(t) - \phi(t))^2 \right) \le \left(C_1 \frac{\beta^2 T^2}{L^2}\right)^2 +  C_2^2 \frac{1}{M}.
    \end{equation}

    Note that $\sqrt{a + b} \le \sqrt{a} + \sqrt{b}$ holds for any $a, b \ge 0$. Then, we have
    \begin{equation}
        \sup_{t \in \mc{I}_\circ} \mathbb{E}\abs{\hat{\phi}(t) - \phi(t)} \le C_1 \frac{\beta^2 T^2}{L^2}^2 +  C_2\frac{1}{\sqrt{M}}.
    \end{equation}
    Similarly, we can prove the bound in the boundary region $[0, T] \backslash \mc{I}_\circ$ in which the bias is first order. The proof is complete.
\end{proof}

\subsection{Performance in terms of uniform maximal error}
\begin{thm}\label{thm:uniform_maximal_error}
    Let $\phi \in \mathscr{P}_\beta([0, T])$, and let $L, M$ be positive integers. Suppose $M$ is the number of measurement shots per experiment, and that there are $L + 1$ experiments implementing Hamiltonian magnitudes $\omega_j = \frac{(2j+1)\pi}{4L+4}$ for $j = 0, 1, \ldots, L$. Let $\mc{I}^\circ = [1/L, T - 1/L]$ denote the interior region. Let $\hat\phi(t)$ denote the reconstructed pulse obtained from the end-to-end pipeline. Then, there exist constants $\wt{C}_1, \wt{C}_2$ such that for any inconfidence $\delta$
    \begin{equation}
        \mathbb{P}\left(\sup_{t \in \mc{I}^\circ} \abs{\hat\phi(t) - \phi(t)}^2 \le \left(\wt{C}_1 \frac{\beta^2 T^2}{L^2}\right)^2 + \wt{C}_2^2 \frac{\log (2 L / \delta)}{M}\right) \ge 1 - \delta.
    \end{equation}
    At the boundaries $t \in [0, 1/L] \cup [T - 1/L, T]$, a similar error bound can be derived by substituting the bias term with $\Or(\beta T / L)$.
\end{thm}
Applying the inequality $\sqrt{a + b} \le \sqrt{a} + \sqrt{b}$, the following uniform maximal error bound holds
\begin{equation}
    \sup_{t \in \mc{I}^\circ} \abs{\hat\phi(t) - \phi(t)} \le \wt{C}_1 \frac{\beta^2 T^2}{L^2} + \wt{C}_2 \sqrt{\frac{\log L}{M}}\quad \text{with high probability}.
\end{equation}
\begin{proof}
    Applying the inequality $(a+b)^2 \le 2a^2 + 2b^2$ to \cref{eqn:bias_variance_decomposition}, we have
    \begin{equation}
        \abs{\hat{\phi}(t) - \phi(t)}^2 \le 2\ \mathrm{bias}^2(t) + 2\abs{\zeta^\top(t) \vec{\varepsilon}}^2 \le 2\ \mathrm{bias}^2(t) + 2\norm{\zeta^\top(t)}_1^2 \norm{\vec{\varepsilon}}_\infty^2.
    \end{equation}

    Using union bound and the Gaussian tail bound, we have
    \begin{equation}
    \begin{split}
        \mathbb{P}(\norm{\vec{\varepsilon}}_\infty^2 \le 2 \sigma^2 \log(2L/\delta)) &\ge \mathbb{P}(\max_i \abs{\varepsilon_i / \sigma_i}^2 \le 2 \log(2L/\delta))\\
        &= 1 - \mathbb{P}(\max_i z_i^2 > 2 \log(2L/\delta))\quad \text{ where } z_i \sim N(0, 1)\\
        & \ge 1 - \sum_{i = 1}^L \mathbb{P}(z_i^2 > 2 \log(2L/\delta)) \ge 1 - 2 L e^{- \log(2L/\delta)} \ge 1 - \delta.
    \end{split}
    \end{equation}

    The result follows.

\end{proof}

\section{Digitizing Analog Pulse Function via Segmented Generators}\label{sec:app_analytical_properties}

We begin by outlining the technical results presented in this section before providing the detailed proofs.

Though the time-dependent Hamiltonian dynamics can not be treated simply like time-independent dynamics, the time evolution matrix can still be factored into the product of short-time evolutions, Given a time partition $0 = t_0 < t_1 < \cdots < t_L = T$, we have
\begin{equation}
    \TOR e^{- \I \int_0^T H(s; \omega) \ud s} = \TOR e^{- \I \int_{t_{L - 1}}^{t_L} H(s; \omega) \ud s} \cdots \TOR e^{- \I \int_{t_0}^{t_1} H(s; \omega) \ud s}.
\end{equation}
In general cases, pinning down a choice of Lie algebra allows us to represent each short-time evolution matrix. Specifically, in the single-qubit case, Pauli matrices can be set to a basis of the Lie algebra. For the time-dependent Hamiltonian evolution considered in this paper, we will show in \cref{lem:time_evolution_generator_repr} that they admit the following form:
\begin{equation}
    \TOR e^{-\I \int_{t_{j-1}}^{t_j} H(s) \ud s} = \exp\left(- \I \left(a(t_1, t_2; \omega) X + b(t_1, t_2; \omega) Y + c(t_1, t_2; \omega) Z\right) \right).
\end{equation}
Here, we refer these functions $a, b, c$ as the generators of each short-time evolution. Intuitively, when the evolution time is short enough $t_2 - t_1 \ll 1$, the pulse function involved in each segment becomes much simpler to characterize. Though this intuition is simple, the complex time-ordering operator makes the practically useful relation illusively hidden behind the convoluted nonlinear $\phi$-dependencies. How can one formalize this intuition as a tangible algorithm? In 
\cref{thm:estimation_error_psi_barphi}, we show that the following bound relates the generators and the characteristics of the pulse:
\begin{equation}\label{eqn:ADC_angle_error}
\begin{split}
    & \psi(t_1, t_2; \omega): = \arctan\left( \frac{b(t_1, t_2; \omega)}{a(t_1, t_2; \omega)} \right),\\
    & \abs{\psi(t_1, t_2; \omega) - \frac{1}{t_2 - t_1} \int_{t_1}^{t_2} \phi(s) \ud s} \le \Or\left((1 + \omega^2)(t_2 - t_1)^4\right).
\end{split}
\end{equation}
As a result of nontrivial analyses, this bound implies that if the generators of the short-time evolution can be learned, the averaged pulse value on such short-time interval can be approximately derived. 

Furthermore, this establishes an analog-to-digital conversion at the logical level. Specifically, we can define the following mapping through a bounded-error correspondence stated in \cref{eqn:ADC_angle_error}:
\begin{equation}\label{eqn:ADC_logical}
    \text{Analog Rep.:}\ \phi \in \mathscr{P}_\beta,\ \TOR e^{- \I \int_0^T H(s; \omega) \ud s} \quad \to \quad \text{Digital Rep.:}\ (\bar{\phi}_0, \cdots, \bar{\phi}_L) \in \RR^L,\ e^{-\I H[\bar{\phi}_L]} \cdots e^{-\I H[\bar{\phi}_1]}
\end{equation}
The reverse mapping is established in \cref{sec:app:reconstruct_analog}.

These constitute the preliminary steps of our data post-processing. In later sections, we introduce an algorithm for learning the generators in \cref{sec:app_Fourier_analysis_post_processing}, followed by another algorithm for reconstructing the analog pulse from the learned parameters in \cref{sec:app:reconstruct_analog}. Together, these three components form the complete procedure for learning analog pulse functions.

\subsection{Structure of short-time evolution operator}
For simplicity of notation, we define an anti-Hermitian operator $A := -\I H$.
The time-dependent Hamiltonian dynamics we consider can be characterized by the following differential equation
\begin{equation}\label{eqn:time_dep_H}
\begin{split}
    & \frac{\ud}{\ud t} U(t, 0; \omega) = A(t; \omega) U(t, 0; \omega), U(0, 0; \omega) = I, \\
    & \text{where } A(t; \omega) = - \I \omega (\cos(\phi(t)) X + \sin(\phi(t)) Y).
\end{split}
\end{equation}

Formally, the exact solution of this equation can be expressed using a time ordering operator or using Magnus series expansion:
\begin{equation}
    U(t, 0; \omega) = \TOR e^{\int_0^t A(s; \omega) \ud s} = e^{\Omega(t, 0; \omega)}.
\end{equation}
Here, $\Omega(t, 0; \omega)$ is the Magnus operator representing the exponential integrator from time $0$ to time $t$. When the initial time is equal to $0$, we simply write the Magnus operator as $\Omega(t) := \Omega(t, 0; \omega)$. We also drop the $\omega$-dependency in other operators for notational simplicity. 

As a generator in $\mf{su}(2)$, the Magnus operator admits a Pauli decomposition:
\begin{equation}\label{eqn:Magnus_generator}
    \Omega(t) = a(t) X + b(t) Y + c(t) Z.
\end{equation}

This operator can also be written as an infinite series:
\begin{equation}
    \Omega(t) = \sum_{n = 1}^\infty \Omega_n(t)
\end{equation}
where each term can be written in terms of nested commutator:
\begin{equation}\label{eqn:recursive_definition_Omega_n}
    \Omega_n(t) = \sum_{j = 1}^{n - 1} \frac{B_j}{j!} \sum_{\substack{k_1 + \cdots + k_j = n - 1\\ k_1 \ge 1, \cdots, k_j \ge 1}} \int_0^t [\Omega_{k_1}(s), [\Omega_{k_2}(s), \cdots [\Omega_{k_j}(s), A(s)]]] \ud s, n \ge 2.
\end{equation}
Here, $B_j$ is the Bernoulli numbers with $B_1 = - 1/ 2$. The first few terms are
\begin{equation}\label{eqn:Magnus_first_three_terms}
    \begin{split}
        & \Omega_1(t) = \int_0^t A(s) \ud s, \quad \Omega_2(t) = \frac{1}{2} \int_0^t \int_0^{t_1} [A(t_1), A(t_2)] \ud t_2 \ud t_1,\\
        & \Omega_3(t) = \frac{1}{6} \int_0^t \int_0^{t_1} \int_0^{t_2} [A(t_1), [A(t_2), A(t_3)]] + [A(t_3), [A(t_2), A(t_1)]] \ud t_3 \ud t_2 \ud t_1.
    \end{split}
\end{equation}

Note that the time-dependent Hamiltonian in \cref{eqn:time_dep_H} involves oscillatory trigonometric function dependencies. These oscillation contributes to the cancellation of the leading terms in the multi-fold integral in the Magnus term. As in our analyses outlined in \cref{sec:app:further_analysis_Magnus}, the property of the Magnus expansion is as follows.

\begin{lem}\label{lem:time_evolution_generator_repr}
    The $n$-th order Magnus expansion is
    \begin{align}
        \I \Omega_n(t) = \left\{ 
        \begin{array}{ll}
            c_n(t) Z & \text{when } n \text{ is even}, \\
            a_n(t) X + b_n(t) Y & \text{when } n \text{ is odd}.
        \end{array}
        \right.
    \end{align}
    Here, $a_n, b_n, c_n$ are some scalar functions. Furthermore, their magnitudes are
    \begin{equation}
        \norm{\Omega_{2k+1}(t)} = \Or(\beta^2 \omega^{2k+1} t^{2k+3}), k = 0, 1, \cdots, \quad \norm{\Omega_{2k}(t)} = \Or(\beta \omega^{2k} t^{2k+1}), k = 1, 2, \cdots.
    \end{equation}
\end{lem}

We aim to investigate the relationship between the Magnus generators in \cref{eqn:Magnus_generator} and the characteristics of the analog pulse function. For higher-order Magnus terms, the pulse dependence becomes highly intricate due to the multi-fold integrals involved. Therefore, it is more tractable to first analyze this relationship through the first-order Magnus generators $a_1$ and $b_1$ in \cref{lem:time_evolution_generator_repr}. In the following analysis, we begin by quantifying this connection using the first-order generators, and then establish bounds that characterize the corresponding relations derived from the full Magnus generators $a$ and $b$ in \cref{eqn:Magnus_generator}.

In the following analysis, we denote the step size parameter as $h := T / L$ and emphasize that it is small.

\begin{lem}\label{lem:first_order_Magnus_angle_estimation}
    Let 
    \begin{equation}
        \wt{\psi} := \arctan(b_1(h) / b_2(h)) = \arctan\left(\int_0^h \sin(\phi(s)) \ud s \middle/ \int_0^h \cos(\phi(s)) \ud s \right) \text{ and } \apq{\phi} := \frac{1}{h} \int_{0}^{h} \phi(s) \ud s.
    \end{equation}
    Then, the following estimation holds:
    \begin{equation}
        \wt{\psi} - \apq{\phi} = \Or(\beta^4 h^4).
    \end{equation}
\end{lem}
\begin{proof}
    We first note that 
    \begin{equation}
        \abs{\phi(s) - \apq{\phi}} \le \frac{1}{h} \int_0^h \abs{\phi(s) - \phi(u)} \ud u \le \Or(\beta h).
    \end{equation}
    Then, we can study the integrals by expanding the functions. We first have the following:
    \begin{equation}
    \begin{split}
        \phi(s) &= \phi_0 + \phi_0^\prime s + \frac{1}{2} \phi_0^{\prime\prime} s^2 + \frac{1}{6} \phi_0^{\prime\prime\prime} s^3 + \Or(\beta^4 h^4),\\
        \apq{\phi} &= \phi_0 + \frac{1}{2} \phi_0^\prime h + \frac{1}{6} \phi_0^{\prime\prime} h^2 + \frac{1}{24} \phi_0^{\prime\prime\prime} h^3 + \Or(\beta^4 h^4).
    \end{split}
    \end{equation}
    Then
    \begin{equation}
        \begin{split} 
            \cos(\phi(s)) =& \cos(\apq{\phi}) - \sin(\apq{\phi}) (\phi(s) - \apq{\phi}) - \frac{1}{2} \cos(\apq{\phi}) (\phi(s) - \apq{\phi})^2 + \frac{1}{6} \sin(\apq{\phi}) (\phi(s) - \apq{\phi})^3\\
            &+ \frac{1}{24} \cos(\apq{\phi}) (\phi(s) - \apq{\phi})^4 + \Or(\beta^5 h^5)\\
            =& \cos(\apq{\phi}) - \sin(\apq{\phi}) (\phi(s) - \apq{\phi}) - \frac{1}{2} \cos(\apq{\phi}) (\phi(s) - \apq{\phi})^2\\
            &+ \frac{1}{6} \sin(\apq{\phi}) \left( \phi_0^\prime (s - h/2) + \frac{1}{2} \phi_0^{\prime\prime} (s^2 - h^2 / 3) \right)^3 + \frac{1}{24} \cos(\apq{\phi}) \left( \phi_0^\prime (s - h/2)\right)^4 + \Or(\beta^5 h^5)\\
            =& \cos(\apq{\phi}) - \sin(\apq{\phi}) (\phi(s) - \apq{\phi}) - \frac{1}{2} \cos(\apq{\phi}) (\phi(s) - \apq{\phi})^2\\
            &+ \frac{1}{6} \sin(\apq{\phi}) \left( ( \phi_0^\prime)^3 (s - h/2)^3 + \frac{3}{2}  (\phi_0^\prime)^2 \phi_0^{\prime\prime} (s - h/2)^2 (s^2 - h^2 / 3) \right)\\
            &+ \frac{1}{24} \cos(\apq{\phi}) \left( \phi_0^\prime (s - h/2)\right)^4 + \Or(\beta^5 h^5)\\
        \end{split}
    \end{equation}
    Integrating this function over the interval $0 \le s \le h$, we have
    \begin{equation}\label{eqn:Taylor-int-cos}
        \int_0^h \cos(\phi(s)) \ud s = h \cos(\apq{\phi}) (1 - q(h)) + \frac{1}{720} \sin(\apq{\phi}) (\phi_0^\prime)^2 \phi_0^{\prime\prime} h^5 + \frac{1}{1920} \cos(\apq{\phi}) (\phi_0^\prime)^4 h^5 + \Or(\beta^5 h^6).
    \end{equation}
    Here
    \begin{equation}
        q(h) := \frac{1}{2 h} \int_0^h (\phi(s) - \apq{\phi})^2 \ud s.
    \end{equation}
    Similarly, we have
    \begin{equation}\label{eqn:Taylor-int-sin}
        \int_0^h \sin(\phi(s)) \ud s = h \sin(\apq{\phi}) (1 - q(h)) - \frac{1}{720} \cos(\apq{\phi}) (\phi_0^\prime)^2 \phi_0^{\prime\prime} h^5 + \frac{1}{1920} \sin(\apq{\phi}) (\phi_0^\prime)^4 h^5 + \Or(\beta^5 h^6).
    \end{equation}
    Hence
    \begin{equation}
    \begin{split}
        \wt{\psi} &= \arctan\left(\frac{\int_0^h \sin(\phi(s)) \ud s}{\int_0^h \cos(\phi(s)) \ud s}\right) = \arctan\left( \frac{h \sin(\apq{\phi}) (1 - q(h)) + \Or(\beta^4 h^5)}{h \cos(\apq{\phi}) (1 - q(h)) + \Or(\beta^4 h^5)} \right)\\
        &= \arctan(\tan(\apq{\phi}) + \Or(\beta^4 h^4)) = \apq{\phi} + \Or(\beta^4 h^4).
    \end{split}
    \end{equation}
    This is a remarkable result. \cref{eqn:Taylor-int-cos,eqn:Taylor-int-sin} indicate that the integrals of sine and cosine are approximated up to an error $\Or(h q(h)) = \Or(h^3)$. However, because the third-order errors in the numerator and denominator cancel in the quotient, the error to $\psi - \apq{\phi}$ is improved to the fourth-order. 
\end{proof}
Now, we can lift this relation to that of the full Magnus generator series.
\begin{thm}\label{thm:estimation_error_psi_barphi}
    Let $\phi \in \mathscr{P}_\beta$ be a Bernstein-type pulse, $\Omega(h) := a(h) X + b(h) Y + c(h) Z$,  
    \begin{equation}
        \psi := \arctan(b(h) / a(h)) \text{ and } \apq{\phi} := \frac{1}{h} \int_{0}^{h} \phi(s) \ud s.
    \end{equation}
    The following error bound holds:
    \begin{equation}
        \abs{\psi - \apq{\phi}} \le \Or((\beta^2 + \omega^2) \beta^2 h^4).
    \end{equation}
\end{thm}

\begin{proof}
    For the simplicity of presenting the proof, we assume $\cos(\apq{\phi}) = \Omega(1)$ is away from zero by at least a nonvanishing constant. Otherwise, we can alternatively implement the estimator using $\psi = \mathrm{atan2}(b(h), a(h))$ to relax this condition.
    
    We first construct a reference generator by keeping the first-order term and all even-order terms:
    \begin{equation}
        \I \wt{\Omega}(h) := a_1(h) X + b_1(h) Y + c(h) Z = \omega \int_0^h \cos(\phi(s)) \ud s X + \omega \int_0^h \sin(\phi(s)) \ud s Y + c(h) Z.
    \end{equation}
    Note that we have
    \begin{equation}
        \norm{\wt{\Omega}(h) - \Omega(h)} = \norm{\sum_{k=1}^\infty \Omega_{2k+1}(h)} \le \Or(\beta^2 \omega^3 h^5).
    \end{equation}
    Then, using H\"{o}lder's inequality, we have
    \begin{equation}
        \abs{a(h) - a_1(h)} = \frac{1}{2} \abs{\Tr\left(X(\Omega(h) - \wt{\Omega}(h)\right)} \le \frac{1}{2} \norm{X}_* \norm{\Omega(h) - \wt{\Omega}(h)} \le \Or(\beta^2 \omega^3 h^5).
    \end{equation}
    Here, $\norm{X}_* = 2$ is the nuclear norm. Similarly, we have
    \begin{equation}
        \abs{b(h) - b_1(h)}  \le \Or(\beta^2 \omega^3 h^5).
    \end{equation}
    Following \cref{lem:first_order_Magnus_angle_estimation}, we have
    \begin{equation}
        \begin{split}
            \abs{\psi - \apq{\phi}} &\le \abs{\psi - \wt{\psi}} + \abs{\wt{\psi} - \apq{\phi}} \le \abs{\arctan(b(h) / a(h)) - \arctan(b_1(h) / a_1(h))} + \Or(\beta^4 h^4)\\
            &\le \abs{\frac{b(h)}{a(h)} - \frac{b_1(h)}{a_1(h)}} \sup_\xi \frac{1}{1 + \xi^2} + \Or(\beta^4 h^4) \le \abs{\frac{b(h)}{a(h)} - \frac{b_1(h)}{a_1(h)}} + \Or(\beta^4 h^4).
        \end{split}
    \end{equation}
    Note that the assumption $\cos(\apq{\phi}) = \Omega(1)$ implies that $\abs{h^{-1} \int_0^h \cos(\phi(s)) \ud s} \ge \abs{\cos(\apq{\phi})} - \Or(h^3) = \Omega(1)$. Then
    \begin{equation}
        \abs{a_1(h)} = \abs{\omega h \cdot \frac{1}{h} \int_0^h \cos(\phi(s)) \ud s} \ge \Omega(\omega h),\  \abs{a(h)} \ge \abs{a_1(h)} - \abs{a_1(h) - a(h)} \ge \Omega(\omega h).
    \end{equation}
    Consequently, using the bound $b_1(h) \le \omega h$, we also have
    \begin{equation}
        \abs{\tan(\wt{\psi})} = \abs{\frac{b_1(h)}{a_1(h)}} \le \Or(1).
    \end{equation}
    Then, the following inequality holds
    \begin{equation}
        \abs{\frac{b(h)}{a(h)} - \frac{b_1(h)}{a_1(h)}} \le \frac{\abs{b_1(h) - b(h)}}{\abs{a(h)}} + \abs{\frac{b_1(h)}{a_1(h)}} \frac{\abs{a_1(h) - a(h)}}{\abs{a(h)}} \le \Or(\beta^2 \omega^2 h^4).
    \end{equation}
    To conclude, we have
    \begin{equation}
        \abs{\psi - \apq{\phi}} \le \Or((\beta^2 + \omega^2) \beta^2 h^4).
    \end{equation}
    The proof is complete.
\end{proof}

\subsection{Justifying analytical results with fully solvable models using linear pulse functions}\label{app:subsec:linear_pulse}
In this subsection, we justify the theoretical results derived in the previous subsection by analyzing a fully solvable model that uses linear pulse functions. We first derive the exact expressions for the time-evolution matrices.
\begin{thm}
    Consider a time-dependent Hamiltonian generated by a linear pulse $H(t) = \omega \left(\cos(\phi(t)) X + \sin(\phi(t)) Y\right)$, where $\phi(t) = \alpha t$. The associated time evolution unitary admits a closed-form expression
\begin{equation}
U(t_2, t_1) = R(t_2) V(t_2 - t_1) R(-t_1)
\end{equation}
where $R(t) = e^{- \I \alpha t / 2 Z}$ and $V(t_2 - t_1) = e^{- \I (t_2 - t_1) (\omega X - \alpha / 2 Z)}$.
\end{thm}

\begin{proof}
The quantum system evolves under a time-dependent Schr\"odinger equation of the form
\begin{equation}
\frac{\ud}{\ud t} U(t) = -\I \omega 
\begin{pmatrix}
0 & e^{- \I \alpha t} \\
e^{\I \alpha t} & 0
\end{pmatrix} U(t).
\end{equation}
Our objective is to derive an explicit expression for the evolution operator $U(t_2, t_1)$
which propagates the state from time \( t_1 \) to \( t_2 \). We rotate the coordinates to eliminate the explicit time dependence in the differential equation. Let
\begin{equation}
R(t) = \begin{pmatrix}
e^{- \I \alpha t/2} & 0 \\
0 & e^{\I \alpha t/2}
\end{pmatrix} \text{ and } U(t) = R(t) V(t).
\end{equation}

Then, we have
\begin{equation}
\frac{\ud}{\ud t} U(t) = \dot{R}(t)\, V(t) + R(t)\, \dot{V}(t)
\end{equation}
and
\begin{equation}
R^{-1}(t) \dot{R}(t)\, V(t) + \dot{V}(t) = -\I \omega R^{-1}(t) H(t) R(t)\, V(t).
\end{equation}

Direct computation gives
\begin{equation}
R^{-1}(t) H(t) R(t) = X,
\text{ and }
R^{-1}(t) \dot{R}(t) = - \I \frac{\alpha}{2} Z.
\end{equation}

Therefore, the effective equation in the rotating frame becomes
\begin{equation}
\dot{V}(t) = -\I \left( \omega X - \frac{\alpha}{2} Z \right) V(t) =: -\I \wt{H} V(t).
\end{equation}
The effective Hamiltonian $\wt{H}$ is time-independent.
Then, the time evolution in the rotating frame is given by
\begin{equation}
V(t_2, t_1) = \exp\left( -\I (t_2 - t_1) \wt{H} \right)
= \cos(\Gamma \Delta t)\, I - \I \frac{\sin(\Gamma \Delta t)}{\Gamma} \wt{H},
\end{equation}
where $\Delta t = t_2 - t_1$ and $\Gamma = \sqrt{\omega^2 + \left( \frac{\alpha}{2} \right)^2}$.

Therefore, we have
\begin{equation}
U(t_2, t_1) = R(t_2)\, V(t_2, t_1)\, R^{-1}(t_1).
\end{equation}
Note that $R^{-1}(t) = R(-t)$. The proof is complete.

\end{proof}

For a linear pulse, the systematic error of representing the averaged pulse is vanishing.
\begin{cor}
    When considering a linear pulse, for any time segment $t_1 < t_2$ it holds that 
    \begin{align}
    \arctan\left(\frac{b(t_2, t_1)}{ a(t_2, t_1)} \right) = \frac{\int_{t_1}^{t_2} \phi(s) \ud s}{(t_2 - t_1)}.
    \end{align}
\end{cor}
\begin{proof}
    Note that the time evolution operator can be expressed as $R(t_2) V(t_2 - t_1)R^{-1}(t_1)$. Here, the expression of $V$ enables a special structure which can be written as
    \begin{equation}
        V(t_2 - t_1) = \begin{pmatrix}
            u & \I v\\
            \I v & u^*
        \end{pmatrix} \text{ where } u := \cos(\Gamma (t_2 - t_1)) + \I \sin(\Gamma(t_2-t_1)) \alpha / (2 \Gamma), v := - \sin(\Gamma(t_2-t_1)) \omega / \Gamma.
    \end{equation}
    We have 
    \begin{equation}
    \begin{split}
        \frac{b(t_2, t_1)}{a(t_2, t_1)} &= \frac{\Tr(U(t_2, t_1) Y)}{\Tr(U(t_2, t_1) X)} = \frac{\Tr(R(t_2) V(t_2 - t_1) R^{-1}(t_1) Y)}{\Tr(R(t_2) V(t_2 - t_1) R^{-1}(t_1) X)} = \frac{\Tr(R(t_2) V(t_2 - t_1) Y R(t_1))}{\Tr(R(t_2) V(t_2 - t_1) X R(t_1))}\\
        &= \frac{\Tr(R(t_1 + t_2) V(t_2 - t_1) Y)}{\Tr(R(t_1 + t_2) V(t_2 - t_1) X)} = \frac{v \sin(\alpha (t_1 + t_2) / 2)}{v \cos(\alpha (t_1 + t_2) / 2)} = \tan(\alpha (t_1 + t_2) / 2).
    \end{split}
    \end{equation}
    The conclusion follows.
\end{proof}
Furthermore, we can justify the difference between the exact and approximate Magnus generators.
\begin{cor}
    The difference between the exact generator and the approximate generator is bounded as follows:
    \begin{equation}
        \norm{\Omega(t_2, t_1) - \wt{\Omega}(t_2, t_1)} \lesssim \frac{1}{60} \alpha^2 \omega^3 h^5, \text{ with } \abs{c(t_2, t_1)} \lesssim \frac{1}{6} \alpha \omega^2 h^3. 
    \end{equation}
\end{cor}
\begin{proof}
    We first recall the property of $\mathrm{SU}(2)$: $e^{- \I \vartheta \hat{\boldsymbol{n}} \cdot \boldsymbol{\sigma}} = \cos \vartheta - \I \sin \vartheta \hat{\boldsymbol{n}} \cdot \boldsymbol{\sigma}$ where $\hat{\boldsymbol{n}} \cdot \boldsymbol{\sigma} := (\hat{n}_x X + \hat{n}_y Y + \hat{n}_z Z)$ and $\hat{\boldsymbol{n}}$ is a normalized vector. Note that we have $U := U(t_2, t_1) = e^{\Omega(t_2, t_1)} = e^{- \I (a X + b Y + c Z)}$. Hence, the procedure for determining these unknowns follows,
    \begin{equation}
    \begin{split}
        & \cos \vartheta = \frac{1}{2} \Tr(U) ,\  a = \vartheta \hat{n}_x = - \frac{\vartheta}{2 \I \sin \vartheta} \Tr(U X),\\
        & b = \vartheta \hat{n}_y = - \frac{\vartheta}{2 \I \sin \vartheta} \Tr(U Y),\  c = \vartheta \hat{n}_z = - \frac{\vartheta}{2 \I \sin \vartheta} \Tr(U Z).
    \end{split}
    \end{equation}
    Using the analytical formulae derived above, we have the following equalities and estimates:
    \begin{align}
        & \cos \vartheta = \cos(\alpha h / 2) \cos(\Gamma h) + \sin(\alpha h / 2) \sin(\Gamma h) \alpha / (2 \Gamma),\\
        & a = \frac{\vartheta}{\sin \vartheta} \frac{\omega}{\Gamma} \sin(\Gamma h) \cos(\alpha \bar t), \quad \wt{a} = \omega \int_{t_1}^{t_2} \cos(\alpha t) \ud t = \frac{2\omega}{\alpha} \sin(\alpha h / 2) \cos(\alpha \bar t), \\
        & b = \frac{\vartheta}{\sin \vartheta} \frac{\omega}{\Gamma} \sin(\Gamma h) \sin(\alpha \bar t), \quad \wt{b} = \omega \int_{t_1}^{t_2} \sin(\alpha t) \ud t = \frac{2\omega}{\alpha} \sin(\alpha h / 2) \sin(\alpha \bar t), \\
        & c = \frac{\vartheta}{\sin \vartheta} \left( \sin(\alpha h / 2) \cos(\Gamma h) - \frac{\alpha}{2\Gamma} \cos(\alpha h / 2) \sin(\Gamma h) \right) = - \frac{1}{6} \alpha \omega^2 h^3 + \Or((\alpha \omega^4 + \alpha^3 \omega^2 ) h^5)
    \end{align}
    where $\bar t := (t_1 + t_2) / 2$ and $h := t_2 - t_1$.

    Note that the following estimate holds
    \begin{equation}
        \frac{\vartheta}{\sin \vartheta} \frac{\omega}{\Gamma} \sin(\Gamma h) - \frac{2\omega}{\alpha} \sin(\alpha h / 2) = - \frac{1}{60} \alpha^2 \omega^3 h^5 + \Or((\alpha^2 \omega^5 + \alpha^4 \omega^3) h^7).
    \end{equation}
    
    Thus, the difference between the exact generator and the approximate generator $\wt{\Omega} := - \I (\wt{a} X + \wt{b} Y + c Z)$ is bounded as
    \begin{equation}
        \begin{split}
            & \norm{\Omega(t_2, t_1) - \wt{\Omega}(t_2, t_1)} = \norm{(a - \wt{a}) X + (b - \wt{b}) Y} = \sqrt{(a - \wt{a})^2 + (b - \wt{b})^2}\\
            &= \abs{\frac{\vartheta}{\sin \vartheta} \frac{\omega}{\Gamma} \sin(\Gamma h) - \frac{2\omega}{\alpha} \sin(\alpha h / 2)} = \frac{1}{60} \alpha^2 \omega^3 h^5 + \Or((\alpha^2 \omega^5 + \alpha^4 \omega^3) h^7).
        \end{split}
    \end{equation}
\end{proof}

\subsection{Deferred proofs and computations}\label{sec:app:further_analysis_Magnus}
In this subsection, we analyze Magnus terms (e.g. \cref{eqn:Magnus_first_three_terms}) in detail. 

Recall that $A(t) = - \I H(t) = - \I \omega (\cos(\phi(t)) X + \sin(\phi(t)) Y)$. We first compute the second-order term via the commutator.

\begin{lem}\label{lem:commutator_A}
The following equality holds
    \begin{align}
        [A(t_1), A(t_2)] = 2 \I \omega^2 \sin(\phi(t_1) - \phi(t_2)) Z.
    \end{align}
\end{lem}
\begin{proof}
Using the commutation relation of Pauli matrices, we have
\begin{align}
\begin{split}
    [A(t_1), A(t_2)] &= - \omega^2 [\cos(\phi(t_1)) X + \sin(\phi(t_1)) Y, \cos(\phi(t_2)) X + \sin(\phi(t_2)) Y]\\
    &= \omega^2 (\sin(\phi(t_1)) \cos(\phi(t_2)) - \cos(\phi(t_1)) \sin(\phi(t_2))) [X, Y]\\
    &= 2 \I \omega^2 \sin(\phi(t_1) - \phi(t_2)) Z.
\end{split}
\end{align}
\end{proof}

\begin{lem}\label{lem:commutator_3_A}
The following equality holds
    \begin{align}
        [A(t_1), [A(t_2), A(t_3)]] = - 4 \I \omega^3 \sin(\phi(t_3) - \phi(t_2)) \left( - \cos(\phi(t_1)) Y + \sin(\phi(t_1)) X \right).
    \end{align}
\end{lem}
\begin{proof}
Using \cref{lem:commutator_A}, we have
\begin{align}
    \begin{split}
        [A(t_1), [A(t_2), A(t_3)]] &= 2 \I \omega^2 \sin(\phi(t_2) - \phi(t_3)) [A(t_1), Z]\\
        &= 2 \omega^3 \sin(\phi(t_2) - \phi(t_3)) \left( \cos(\phi(t_1)) [X, Z] + \sin(\phi(t_1)) [Y, Z] \right)\\
        &= - 4 \I \omega^3 \sin(\phi(t_3) - \phi(t_2)) \left( - \cos(\phi(t_1)) Y + \sin(\phi(t_1)) X \right).
    \end{split}
\end{align}
\end{proof}

We then analyze the third-order term using nested commutators.

\begin{lem}
Let $\phi \in \mathscr{P}_\beta$ be a Bernstein-type pulse function, and $h$ be the step size parameter. We have 
    \begin{align}
        \Omega_3(h) = \Or(\beta^2 \omega^3 h^5).
    \end{align}
\end{lem}
\begin{proof}

Using \cref{lem:commutator_3_A}, we have
\begin{align}
    \begin{split}
        & [A(t_1), [A(t_2), A(t_3)]] + [A(t_3), [A(t_2), A(t_1)]]\\
        =& - 4 \I \omega^3 \left( \sin(\phi(t_3) - \phi(t_2)) \sin(\phi(t_1)) + \sin(\phi(t_1) - \phi(t_2)) \sin(\phi(t_3)) \right) X\\
        & + 4 \I \omega^3 \left( \sin(\phi(t_3) - \phi(t_2)) \cos(\phi(t_1)) + \sin(\phi(t_1) - \phi(t_2)) \cos(\phi(t_3)) \right) Y\\
        =& -4 \I \omega^3 \left(f(t_1, t_2, t_3) X - g(t_1, t_2, t_3) Y\right).
    \end{split}
\end{align}
The rest of the proof analyzes the behavior of the three-layer integral of $f$ and $g$ in the limit of small $h$.

We expand $f(t_1, t_2, t_3)$ using Taylor series around $0$, retaining terms up to $\Or(h^3)$. The expansion of $f(t_1, t_2, t_3)$ is
\begin{equation}
f(t_1, t_2, t_3) = a b t_{1} - 2 a b t_{2} + a b t_{3} + a c t_{1}^{2} - 2 a c t_{2}^{2} + a c t_{3}^{2} - b^{2} t_{1} t_{2} + 2 b^{2} t_{1} t_{3} - b^{2} t_{2} t_{3} + \Or(h^3).
\end{equation}
Here, $a = \phi(0), b = \phi^\prime(0), c = \phi^{\prime\prime}(0)$. Integrating over $0 \leq t_3 \leq t_2 \leq t_1 \leq h$, we have
\begin{equation}
\int_0^h \int_0^{t_1} \int_0^{t_2} f(t_1, t_2, t_3) \, dt_3 \, dt_2 \, dt_1 = \left(\frac{a c}{60} - \frac{b^{2}}{40}\right) h^5 + \Or(h^6).
\end{equation}

We can similarly analyze the other function 
\begin{equation}
\begin{split}
    g(t_1, t_2, t_3) &= \sin(\phi(t_1) - \phi(t_2)) \cos(\phi(t_3)) + \sin(\phi(t_3) - \phi(t_2)) \cos(\phi(t_1))\\
    &= b t_{1} - 2 b t_{2} + b t_{3} + c t_{1}^{2} - 2 c t_{2}^{2} + c t_{3}^{2} + \Or(h^3).
\end{split}
\end{equation}
Integrating over $0 \leq t_3 \leq t_2 \leq t_1 \leq h$ and noting that the first term vanishes after integration, we get
\begin{equation}
    \int_0^h \int_0^{t_1} \int_0^{t_2} g(t_1, t_2, t_3) \ud t_3 \ud t_2 \ud t_1 = \frac{c}{60} h^5 + \Or(h^6).
\end{equation}
These complete the proof.

\end{proof}

The fifth-order term can be analyzed similarly.

\begin{lem}
Let $\phi \in \mathscr{P}_\beta$ be a Bernstein-type pulse function, and $h$ be the step size parameter. We have 
    \begin{equation}
        \Omega_5(h) = \Or(\beta^2 \omega^5 h^7).
    \end{equation}
\end{lem}
\begin{proof}
    There are 22 nested commutator terms in the fifth-order Magnus term. Formally, it can be written as
    \begin{equation}
        \Omega_5(h) = - \I \frac{2}{5!} \int_0^h \int_0^{t_1} \int_0^{t_2} \int_0^{t_3} \int_0^{t_4} f X + g Y \ud t_5 \ud t_4 \ud t_3 \ud t_2 \ud t_1.
    \end{equation}
    Using symbolic computation packages, such as \textsf{sympy} in \textsf{python}, we can evaluate its Taylor expansion around zero:
    \begin{equation}
    \begin{split}
        f(t_1, t_2,& t_3, t_4, t_5) = - 32 a b t_{1} - 32 a b t_{2} + 128 a b t_{3} - 32 a b t_{4} - 32 a b t_{5} - 32 a c t_{1}^{2} - 32 a c t_{2}^{2}\\
        & + 128 a c t_{3}^{2} - 32 a c t_{4}^{2} - 32 a c t_{5}^{2} - 160 b^{2} t_{1} t_{2} + 288 b^{2} t_{1} t_{3} - 160 b^{2} t_{1} t_{5} - 80 b^{2} t_{2} t_{3}\\
        & + 16 b^{2} t_{2} t_{4}  + 192 b^{2} t_{2} t_{5} - 32 b^{2} t_{3} t_{4} - 48 b^{2} t_{3} t_{5} - 16 b^{2} t_{4} t_{5} + \Or(h^3)
    \end{split}
    \end{equation}
and
    \begin{equation}
    \begin{split}
        g(t_1, t_2, t_3, t_4, t_5) =&32 b t_{1} + 32 b t_{2} - 128 b t_{3} + 32 b t_{4} + 32 b t_{5} + 32 c t_{1}^{2} + 32 c t_{2}^{2}\\
        & - 128 c t_{3}^{2} + 32 c t_{4}^{2} + 32 c t_{5}^{2} + \Or(h^3).
    \end{split}
    \end{equation}

    Here, $a = \phi(0), b = \phi^\prime(0), c = \phi^{\prime\prime}(0)$. Their integrals are as follows:
\begin{equation}
    \begin{split}
        \int_0^h \int_0^{t_1} \int_0^{t_2} \int_0^{t_3} \int_0^{t_4} f \ud t_5 \ud t_4 \ud t_3 \ud t_2 \ud t_1 = \left(- \frac{4 a c}{63} - \frac{4 b^{2}}{105}\right) h^7 + \Or(h^8) = \Or(h^7).
    \end{split}
\end{equation}
and
\begin{equation}
    \begin{split}
        \int_0^h \int_0^{t_1} \int_0^{t_2} \int_0^{t_3} \int_0^{t_4} g \ud t_5 \ud t_4 \ud t_3 \ud t_2 \ud t_1 = \frac{4 c}{63} h^{7} + \Or(h^8) = \Or(h^7).
    \end{split}
\end{equation}

These complete the proof.

\end{proof}

Following this pattern, we conjecture that the general term takes the form
\begin{equation}
    \norm{\Omega_{2k+1}(h)} = \Or(\beta^2 \omega^{2k+1} h^{2k+3}), k = 0, 1, \cdots, \quad \norm{\Omega_{2k}(h)} = \Or(\beta \omega^{2k} h^{2k+1}), k = 1, 2, \cdots.
\end{equation}

We justify them numerically in \cref{fig:LTE}. The left panel indicates the agreement up to the leading order $\Omega_2 \sim h^3$ and $\Omega_3 \sim h^5$ when setting $\omega = \Or(1)$. The right panel indicates that the $\omega$ and $h$ scaling matches our conjecture when setting $\omega h = 1$.

\begin{figure}[htbp]
    \centering
    \subfloat[$\omega = 1$]{\includegraphics[width = .45\textwidth]{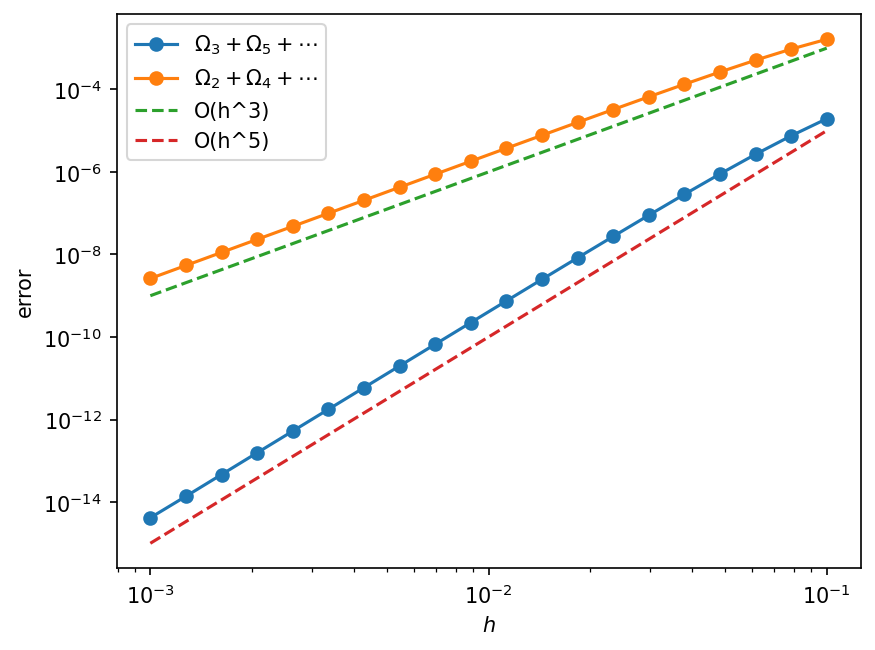}}
    \subfloat[$\omega = 1 / h$]{\includegraphics[width = .45\textwidth]{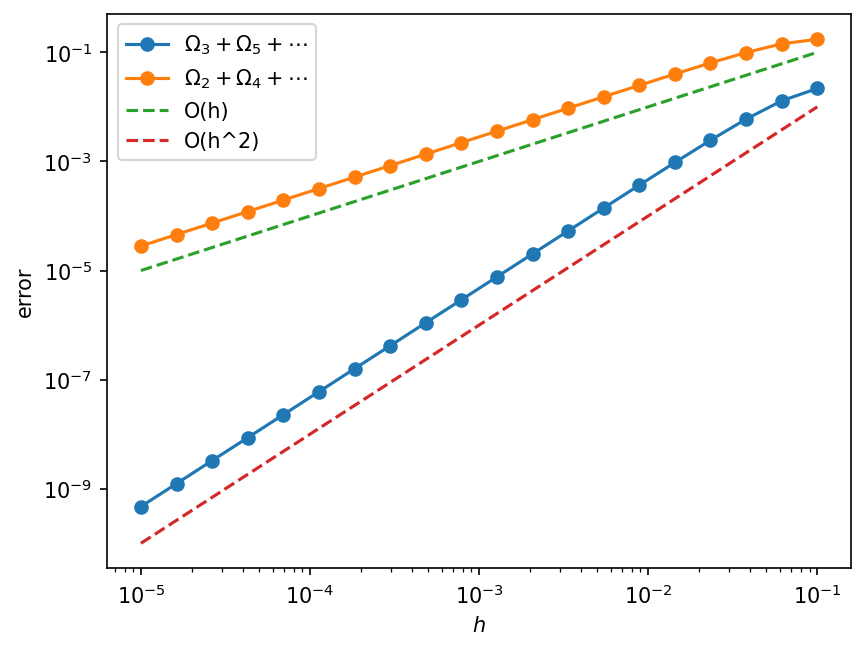}}
    \caption{Truncation error of the Magnus expansion. Blue dots stand for the sum of high-order Magnus terms in odd orders. Orange dots stand for the magnitude of the Pauli $Z$ components in the generator, which is equal to the sum of all even-order Magnus terms. The reference solution is generated by solving the time-dependent Hamiltonian dynamics using the fourth-order Runge-Kutta method (RK4).}
    \label{fig:LTE}
\end{figure} 

Another observation is that the even and odd Magnus terms admit distinct representations, as in \cref{lem:time_evolution_generator_repr}. We provide the proof as follows.
\begin{proof}[Proof of \cref{lem:time_evolution_generator_repr} (structure part)]
    We prove the structure by induction. In the base case $k = 1, 2$, we see that $\Omega_1 \in \mathrm{span}\{X, Y\}$ and $\Omega_2 \in \mathrm{span}\{Z\}$. Our induction hypothesis is $\Omega_{2k+1} \in \mathrm{span}\{X, Y\}$ and $\Omega_{2k+2} \in \mathrm{span}\{Z\}$ for $k \ge 1$. 
    
    It suffices to prove an intermediate relation. Let $A, B \in \mathrm{span}\{X, Y\}$ and $C \in \mathrm{span}\{Z\}$. We have
    \begin{equation}
        [A, B] = [a_x X + a_y Y, b_x X + b_y Y] = (a_x b_y - a_y b_x) [X, Y] \in \mathrm{span}\{Z\}
    \end{equation}
    and 
    \begin{equation}
        [A, C] = [a_x X + a_y Y, c_z Z] = z_x c_z [X, Z] + a_y c_z [Y, Z] \in \mathrm{span}\{X, Y\}.
    \end{equation}
    Let $\mc{M}$ be a classifier, which equals to $\mc{M}(\mathrm{span}\{X, Y\}) = 1$ and $\mc{M}(\mathrm{span}\{Z\}) = 0$. Following the previous relation, we have
    \begin{equation}
        \mc{M}([X_1, X_2]) = \mc{M}(X_1) \oplus \mc{M}(M_2)),\ \text{where } x \oplus y = (x+y) \ \mathrm{mod}\ 2.
    \end{equation}

    Then, our induction hypothesis is equivalent to
    \begin{equation}
        \mc{M}(\Omega_k) = k \ \mathrm{mod}\ 2, \ \forall k \le n - 1.
    \end{equation}
    
    Let us consider an integer $j \in [1, n]$ and $k_1, \cdots, k_j \ge 1$ so that $k_1 + \cdots + k_j = n - 1$. Then, we have
    \begin{equation}
        \mc{M}([\Omega_{k_j}, A]) = \mc{M}(\Omega_{k_j}) \oplus \mc{M}(A) = \mc{M}(\Omega_{k_j}) \oplus 1
    \end{equation}
    and
    \begin{equation}
        \mc{M}([\Omega_{k_1}, [\Omega_{k_2}, \cdots, [\Omega_{k_j}, A]]]) = \mc{M}(\Omega_{k_1}) \oplus \cdots \mc{M}(\Omega_{k_j}) \oplus 1 = (1 + \sum_{q=1}^j k_q) \ \mathrm{mod}\ 2 = n \ \mathrm{mod}\ 2.
    \end{equation}
    Following the recursive definition \cref{eqn:recursive_definition_Omega_n}, we have
    \begin{equation}
        \mc{M}(\Omega_n) = n \ \mathrm{mod}\ 2.
    \end{equation}
    By induction, the proof is complete.
\end{proof}

\section{Reconstructing Analog Pulse Functions from Learned Digital Data}\label{sec:app:reconstruct_analog}

In \cref{sec:app_analytical_properties}, we demonstrated that the analog evolution can be regarded as a digital representation with bounded error (see \cref{eqn:ADC_logical}). In this section, we address the reverse problem, i.e., reconstructing the analog pulse function from a given set of learned digital data. This corresponds to the following mapping:

\begin{equation}\label{eqn:DAC_logical}
    \text{Digital Rep. (Data):}\ (\psi_0, \cdots, \psi_L) \in \RR^L  \quad \to \quad \text{Analog Rep. (Pulse):}\ \hat\phi \in C^\infty.
\end{equation}

In this section, we address the reconstruction problem using spline-based interpolation methods. Splines are powerful tools in non-parametric machine learning, providing flexible and smooth function approximations. In contrast to high-degree polynomial interpolation methods such as Lagrange interpolation, which are prone to the instability known as Runge’s phenomenon, spline interpolation employs low-degree polynomial pieces that connect data points locally and smoothly, resulting in more stable and robust approximations. 

It is worth noting that our queryable data do not correspond to function values sampled at discrete points. Instead, they represent averaged function values over small intervals. This distinct problem setup differentiates our interpolation method from the standard textbook formulation of spline interpolation. In this section, we introduce two approaches for reconstructing the pulse function. The first method tailors the problem structure by differentiating an interpolation of the antiderivative of the pulse function, which we refer to as the \textit{Differentiating Spline (DS)} method. The second method, which is referred to as the \textit{Midpoint Spline (MS)} method, leverages the underlying data-generation process and carries a clearer information-theoretic interpretation.

\subsection{Pulse reconstruction method via differentiating spline}

Let the time grid point be $t_j := j h = j T / L$. The antiderivative of the pulse function can be written as
\begin{equation}
    \Phi(t_j) := \int_0^{t_j} \phi(s) \ud s = \sum_{k = 1}^j \int_{(k-1)h}^{k h} \phi(s) \ud s \approx h \sum_{k = 1}^j \psi_j =: \Psi_j.
\end{equation}
Hence, when the segmented generators can be learned, taking the prefix sum can assemble the antiderivative value of the pulse function. Let $\mc{S}_{[\Psi]}(t)$ be a spline interpolation on the dataset $\{(t_j, \Psi_j): j = 0, \cdots, L\}$ with $t_0 = \Psi_0 = 0$. Then, we have 
\begin{equation}
    \mc{S}_{[\Psi]}(t) \approx \mc{S}_{[\Phi]}(t) \approx \int_0^t \phi(s) \ud s \quad \Rightarrow \quad \hat{\phi}(t) := \frac{\ud }{\ud t} \mc{S}_{[\Psi]}(t) \approx \phi(t).
\end{equation}
Our interpolation data is constructed from the prefix sum, whose pointwise error is still bounded
\begin{equation}
    \abs{\Psi_j - \Phi(t_j)} \le h \sum_{k = 1}^j \abs{\psi_j - \bar\phi_j} \le T \cdot \max_{1 \le j \le L} \abs{\psi_j - \bar\phi_j}.
\end{equation}
When differentiating the spline for pulse recovery, one might expect the pointwise error of the reconstructed pulse function to be amplified by a factor of the inverse step size parameter $1/h$. However, the good news is that no such amplification occurs, because the derivative of a spline interpolation effectively relies only on local data for its computation.

The performance guarantee of this method is summarized in the following theorem.

\begin{thm}\label{thm:DS_spline_estimation_error}
    Let $\phi \in \mathscr{P}_\beta$ be a Bernstein-type pulse, $\mc{S}$ be a $q$-th order spline interpolant. Suppose each data point $\psi_j = \bar\phi_j + b_j + \varepsilon_j$ is subjected to a bounded error $\abs{b_j} \le \delta$ and a random normal-distributed error (not necessarily iid) $\vec{\varepsilon} \sim N(0, \Sigma), \Sigma \preceq \sigma^2 I$. Let the interior area of the interpolation be $\mc{I}^\circ_c := [ch, T - ch]$ where $h := T / L$ and $c \ge q$ is a constant. Then, there exist $(h, \sigma)$-independent constants $C_1, C_2, C_3 > 0$, so that the pointwise interpolation error in the interior is bounded as follows:
    \begin{equation}
        \hat{\phi}(t) := \frac{\ud }{\ud t} \mc{S}_{[\Psi]}(t),\quad \sup_{t \in \mc{I}^\circ_c} \mathbb{E} \abs{\hat{\phi}(t) - \phi(t)} \le C_1 \beta^q h^q + C_2 \delta + C_3 \sigma.
    \end{equation}
\end{thm}

\begin{proof}
    The first error term is the systematic truncation error of spline interpolation, where local degree-$q$ polynomials are used. We first show that the spline interpolant effectively relies only on local data for its computation. We start from the analysis of a simple case of cubic spline $q = 3$. 

    Let $y_j := h \sum_{i \le j} \psi_i$. When $t \in [t_j, t_{j + 1}]$, the interpolant is determined by the following equation:
    \begin{equation}\label{eqn:cubic_spline_defn}
        \mc{S}_{[\Psi]}(t) = \frac{m_j}{6h}(t_{j+1} - t)^3 + \frac{m_j}{6h} (t - t_j)^3 + (y_j - \frac{m_jh^2}{6}) \frac{t_{j+1}-t}{h} + (y_{j+1} - \frac{m_{j+1} h^2}{6}) \frac{t-t_j}{h}.
    \end{equation}
    Taking the natural boundary condition as an example, the second-order derivative parameters are derived from the following tridiagonal linear system:
    \begin{equation}
        m_{j-1} + 4 m_j + m_{j + 1} = \frac{6}{h^2} (y_{j+1} - 2y_j + y_{j - 1}) \Rightarrow \vec{m} = \frac{6}{h^2} K^{-1} D \vec{y},\ K = \mathrm{TriDiag}(1, 4, 1),\ D = \mathrm{TriDiag}(1, -2, 1).
    \end{equation}
    Taking derivative of \cref{eqn:cubic_spline_defn}, we have
    \begin{equation}\label{eqn:structure_of_pulse_DS}
        \hat{\phi}(t) = \alpha_j(t) y_j + \alpha_{j+1}(t) y_{j + 1} + \beta_j(t) m_j + \beta_{j+1}(t) m_{j + 1} = (\alpha_j e_j^\top + \alpha_{j+1} e_{j+1}^\top + (\beta_j e_j^\top + \beta_{j+1} e_{j+1}^\top) \frac{6}{h^2} K^{-1} D) \vec{y} =: \gamma^\top \vec{y}.
    \end{equation}
    Here, $\alpha \sim \Or(1/h)$ and $\beta \sim \Or(h)$. Now, we need to show that $\gamma$ is decaying, namely, $\abs{\gamma_i} := \abs{\gamma^\top e_i} \lesssim \Or(\rho^{|i - j|} / h)$ for some $\rho < 1$. From \cref{eqn:structure_of_pulse_DS}, we see that the first two $\alpha$-related terms naturally satisfy this decaying pattern. For $\beta$-related terms, it is suffices to show that $e_j^\top K^{-1} D e_i \lesssim \Or(\rho^{|i - j|})$. Note that $K$ is a banded positive definite matrix. According to \cite[Proposition 2.2]{demko1984decay}, there is a constant $\rho$ so that $(K^{-1})_{i, j} \lesssim \Or(\rho^{|i-j|})$. Then, by applying triangle inequality, we have
    \begin{equation}
        \abs{e_j^\top K^{-1} D e_i} = \abs{(K^{-1})_{j, i - 1} - 2 (K^{-1})_{j, i} + (K^{-1})_{j, i + 1}} \lesssim \Or(\rho^{|i-j|}).
    \end{equation}
    Hence, we show that $\gamma$ is decaying. Furthermore, $\gamma$ is derived from the derivative of a spline, it satisfies $\gamma^\top 1$ so the derivative is not sensitive to a constant shift in interpolation data.
    
    Note that $y$ is the prefix sum of the raw data $\psi$. We may write this relation in a compact matrix form as $\vec{y} = h Q \vec{\psi}$ where $Q$ is a lower triangular matrix whose lower-triangular elements are all equal to one. Hence, $\hat{\phi}(t) = \gamma^\top \vec{y} = h \gamma^\top Q \vec{\psi} =: \zeta^\top \vec{\psi}$. When $m \ge j$, we have
    \begin{equation}
        \abs{\zeta_m} = h \abs{\sum_{i \ge m} \gamma_i} \le \Or(\sum_{k=0}^\infty \rho^{k + |m - j|}) \le \Or(\rho^{|m - j|}).
    \end{equation}
    The other case can be shown similarly. Following the decaying property of $\zeta$, we have $\zeta$ is bounded by $\Or(1)$ in $1$- and $2$-norms. 

    This conclusion can be generalized to the spline method at a higher order $q$. The only difference in the proof is that the structure of matrices $K$ and $D$ spans a wider band, but \cite[Proposition 2.2]{demko1984decay} is still applicable.

    Note that the raw data admits an error decomposition $\vec{\psi} = \vec{\bar{\phi}} + \vec{b} + \vec{\varepsilon}$. We have
    \begin{equation}
        \abs{\mathbb{E}(\hat{\phi}(t)) - \phi(t)} \le C_1 \beta^q h^q + \norm{\zeta^\top \vec{b}}_1 \le C_1 \beta^q h^q + \norm{\zeta}_1 \norm{\vec{b}}_\infty \le C_1 \beta^q h^q + C_2 \delta.
    \end{equation}
    Here, the first error term comes from the intrinsic interpolation error because the spline uses $q$-th order polynomials locally. The second error term comes from the smooth error in the data.

    The variance of the interpolant is
    \begin{equation}
        \mathrm{Var}(\hat{\phi}(t)) = \zeta^\top \mathbb{E}(\vec{\varepsilon}\vec{\varepsilon}^\top) \zeta \le \sigma^2 \norm{\zeta}_2^2 \le \Or(\sigma^2).
    \end{equation}

    Hence, following the steps outlined in \cref{subsec:app:proof_main_theorem}, the error bound is proven. Note that this analysis only applies to interior points. The interpolation points close to the boundary are subjected to more subtle treatment of numerical stencils to align with proper physical boundary conditions, which usually cause potential downgrade in error order.

\end{proof}

\subsection{Pulse reconstruction method by de-averaging raw data}\label{sec:app_midpoint_Hermite}

The DS method in the previous subsection achieves a desirable error bound. However, it does not exploit any knowledge of the signal-generation process; instead, it relies solely on the observational relationship between the raw data and the antiderivative of the pulse function.

In fact, we know that the raw data can be thought as some bounded-error approximation to the averaged pulse value on sliding windows. Hence, we may first try to invert this averaging operation to recover the function value, then apply textbook-version spline interpolation methods to recover the full pulse function. In this subsection, we quantify this intuition and rigorously develop a more elegant reconstruction method with clearer information-theoretic meaning.

We begin by introducing a high-order stencil that allows us to recover the value of $\phi(t)$ at midpoints within the interior of the function's domain.

\begin{lem}[De-averaging raw data]\label{lem:deavg}
Let $\phi\in \mathscr{P}_\beta$ and consider a uniform partition $0=t_0<t_1<\cdots<t_L=T$ with step $h:=t_{i+1}-t_i$ and midpoints $m_i:=(t_i+t_{i+1})/2$. Suppose each data point $\psi_j = \bar\phi_j + b_j + \varepsilon_j$ is subjected to a bounded error $\abs{b_j} \le \delta$ and a random normal-distributed error (not necessarily iid) $\vec{\varepsilon} \sim N(0, \Sigma), \Sigma \preceq \sigma^2 I$. 
For any interior point in $\mc{I}^\circ := [h, T - h]$, define the central second difference $\Delta^2\psi_i:=\psi_{i-1}-2\psi_i+\psi_{i+1}$ and the de-averaged midpoint estimator
\begin{equation}\label{eqn:de-averaging-stencil}
\phi^{\mathrm{mpt}}(m_i)\;:=\;\psi_i-\frac{1}{24}\,\Delta^2\psi_i.
\end{equation}
Then for any $m_i \in \mc{I}^\circ$, we have the following holds
\begin{equation}\label{eq:main-claim}
\mathbb{E}\abs{\phi^{\mathrm{mpt}}(m_i) - \phi(m_i)} = \mc{O}(\beta^4 h^4) + \Or(\delta) + \Or(\sigma).
\end{equation}
At the boundary cells, the same accuracy holds if $\Delta^2$ is replaced by a one-sided fourth-order stencil.
\end{lem}

\begin{proof}
The averaging operator over the symmetric cell about $m_i$ has the even-derivative expansion
\begin{equation}\label{eq:avg-op}
\overline{\phi}_i \;=\; \phi(m_i)\;+\;\frac{h^2}{24}\,\phi''(m_i)\;+\;\frac{h^4}{1920}\,\phi^{(4)}(m_i)\;+\;\mc{O}(\beta^6 h^6),
\end{equation}
obtained by integrating the Taylor series of $\phi$ over $[t_i,t_{i+1}]$. Applying the central second difference operator $\Delta^2$, a discrete Taylor expansion about $m_i$ gives
\begin{equation}\label{eq:second-diff-op-new}
\Delta^2 \overline{\phi}_i
= h^2\,\phi''(m_i)\;+\;\frac{h^4}{8}\,\phi^{(4)}(m_i)\;+\;\mc{O}(\beta^6 h^6).
\end{equation}
Therefore, we have
\begin{equation}\label{eq:avg-minus-second}
\overline{\phi}_i - \frac{1}{24}\Delta^2\overline{\phi}_i
= \Big(\phi(m_i)+\frac{h^2}{24}\phi''(m_i)+\frac{h^4}{1920}\phi^{(4)}(m_i)+\mc{O}(\beta^6 h^6)\Big)
- \frac{1}{24}\Big(h^2\phi''(m_i)+\frac{h^4}{8}\phi^{(4)}(m_i)+\mc{O}(\beta^6 h^6)\Big).
\end{equation}
Because the $h^2\phi''$ terms cancel, we obtain
\begin{equation}\label{eq:avg-cancel}
\overline{\phi}_i - \frac{1}{24}\Delta^2\overline{\phi}_i
= \phi(m_i) - \frac{3}{640}h^4\,\phi^{(4)}(m_i) + \mc{O}(\beta^6 h^6)
= \phi(m_i) + \mc{O}(\beta^4 h^4).
\end{equation}

We apply the triangle equality to the systematic error terms
\begin{equation}
    \abs{b_j - \frac{1}{24} \Delta^2 b_j} \le \abs{\frac{13}{12} b_j - \frac{1}{24} b_{j - 1} - \frac{1}{24} b_{j + 1}} \le \frac{7}{6} \delta.
\end{equation}

The matrix representation of the de-averaging operation is a tridiagonal matrix $Q := \mathrm{TriDiag}(-1/24, 13/12, -1/24)$ up to some small modification according to the choice of boundary stencil. Then, the variance of the de-averaged data is
\begin{equation}
    \mathrm{Cov}(\vec{\phi}^\mathrm{mpt}) = Q^\top \mathbb{E}(\vec{\varepsilon} \vec{\varepsilon}^\top) Q \le \sigma^2 Q^\top Q, \ \text{where } G := Q^\top Q = \mathrm{PentaDiag}(1/576, -13/144, 113/96, -13/144, 1/576).
\end{equation}
Note that $\norm{G}_2 = \Or(1)$. Hence, the random error in the raw data contributes a $\Or(\sigma)$ uncertainty to the de-averaged data with high probability. The proof is complete.
\end{proof}

As a remark, we may include more neighboring data in \cref{eqn:de-averaging-stencil} to construct a higher-order stencil, which is adaptive to the order of the smooth error $\delta$. Furthermore, in a special case when $\delta = \Or(\beta^2 h^2)$, we do not need to keep finer structure in the de-averaging operation. This result is summarized in the following corollary.

\begin{cor}[Second-order de-averaging]\label{cor:second_order_deavg}
    Under the assumptions of \cref{lem:deavg}, when using the raw data as the midpoint function value $\phi^\mathrm{mpt}(m_i) := \psi_i$, we have the following holds in the interior area:
    \begin{equation}
        \max_{m_i \in \mc{I}^\circ} \mathbb{E} \abs{\psi_i - \phi(m_i)} = \Or(\beta^2 h^2) + \Or(\delta) + \Or(\sigma).
    \end{equation}
\end{cor}

The de-averaged midpoint data $\{\phi^{\mathrm{mpt}}(m_i)\}$ from Lemma~\ref{lem:deavg} enables us to reconstruct the pulse function by applying normal spline interpolants. By using de-averaged data, the pulse function is estimated from the spline interpolation itself, rather than its derivative. Consequently, the data locality of spline methods ensures that the pointwise estimation error is not amplified too much. We summarize this result in the following theorem.

\begin{thm}\label{thm:mpt_spline}
Let $\phi \in \mathscr{P}_\beta$ be a Bernstein-type pulse, $\mc{S}$ be a $q$-th order spline interpolant. Suppose each data point $\psi_j = \bar\phi_j + b_j + \varepsilon_j$ is subjected to a bounded error $\abs{b_j} \le \delta$ and a random normal-distributed error (not necessarily iid) $\vec{\varepsilon} \sim N(0, \Sigma), \Sigma \preceq \sigma^2 I$. Let $\Phi^\mathrm{mpt}$ be the midpoint function values by applying the de-averaging techniques in \cref{lem:deavg} or \cref{cor:second_order_deavg}. Let the interior area of the interpolation be $\mc{I}^\circ_c := [ch, T - ch]$ where $h := T / L$ and $c \ge q$ is a constant. Then, there exist $(h, \sigma)$-independent constants $C_1, C_2, C_3, C_4 > 0$, so that the pointwise interpolation error in the interior is bounded as follows:
    \begin{equation}
        \hat{\phi}(t) := \mc{S}_{[\Phi^\mathrm{mpt}]}(t),\quad \sup_{t \in \mc{I}^\circ_c} \mathbb{E} \abs{\hat{\phi}(t) - \phi(t)} \le C_1 \beta^{q+1} h^{q+1} + C_2 \delta + C_3 \sigma + C_4 \beta^\eta h^\eta.
    \end{equation}
    Here, $\eta$ is a constant depending on the choice of de-averaging method. We have $\eta = 2$ when using \cref{cor:second_order_deavg} and $\eta = 4$ when using \cref{lem:deavg}.
\end{thm}

We remark that this result differs from \cref{thm:DS_spline_estimation_error} in two error components. First, the systematic error caused by the spline interpolation is $(q+1)$-th order, due to the direct evaluation rather than taking derivative. Second, the fourth term is the systematic error caused by the de-averaging operation, which can be chosen adaptively to tradeoff other errors.

\section{Surrogate Model for Learning Pulse Characteristics}
\label{sec:app:surrogate_model}

In the last sections, we see that the total reconstruction error depends on the choice of spline interpolant, the error strength $\delta$ in the raw data $(\psi_1, \cdots, \psi_L)$, and the random uncertainty in the raw data $\sigma$. The magnitude of $\delta$ depends essentially on the choice of surrogate model and the learning procedure. In this section, we analyze a simple surrogate model which connects closely to the structure of a quantum algorithm called Quantum Signal Processing (QSP).

When keeping the Magnus expansion to the first order, according to \cref{sec:app_analytical_properties}, we have
\begin{equation}
    V(\theta, \bar\phi_j) := e^{\Omega_1(t_j + h, t_j; \omega)} = e^{- \I \bar\phi_j / 2 Z} e^{-\I \theta X} e^{\I \bar\phi_j / 2 Z}, \text{ where } \theta := \omega h \text{ and } \bar\phi_j := \frac{1}{h} \int_{t_j}^{t_j + h} \phi(s) \ud s .
\end{equation}

The truncation error of this surrogate model is
\begin{equation}
    \norm{e^{\Omega(t_j + h, t_j; \omega)} - V(\theta, \bar\phi_j)} \le \Or(\beta \omega^2 h^3) \le \Or(\beta h) \text{ when } \theta = \omega h = \Theta(1).
\end{equation}

The surrogate model matrix $V$ coincides with the primitive used in QSP. Combining these matrices in subsequent time intervals can form a long-time surrogate model. According to the theory of QSP in the literature \cite{Haah2019,WangDongLin2022}, this model has a clean matrix-valued Fourier representation
\begin{equation}\label{eqn:surrogate_model_QSP}
    V(\theta, \bar\phi_L) \cdots V(\theta, \bar\phi_1) = \sum_{j = -L}^L C_j e^{\I j \theta} \text{ where } C_j \in \CC^{2 \times 2}.
\end{equation}
This Fourier structure gives us a very elegant approach to process experimental data in Fourier space. We systematically study this Fourier-space data post-processing method in \cref{sec:app_Fourier_analysis_post_processing}.

Keeping the higher-order Magnus terms can reduce the truncation error of the surrogate model. However, these terms has a nonlinear dependency in $\theta$, i.e., $\Omega_k \sim \theta^k$. Consequently, the short-time surrogate model matrices have more complex exponents. These exponents contain nonlinear $\theta$ terms, which represent higher-order refinement to the accuracy of the surrogate model. When similarly assembling the long-time surrogate model, we no longer have the simple Fourier structure like that in \cref{eqn:surrogate_model_QSP}. Moreover, the $\phi$-dependency also becomes convoluted in these high-order correction terms due to the multi-fold integral involved in Magnus expansion. Consequently, more complex data post-processing methods are required. This might be addressed by using advanced optimization methods to solve surrogate parameters. We leave this as a direction for future work.

Because our simple surrogate model only has a first-order truncation accuracy, we might think the end-to-end pulse reconstruction error is bottlenecked at the first-order error $\Or(h)$. However, because of the smoothness of this truncation error, we are able to lift the overall end-to-end pulse reconstruction accuracy to second order $\Or(h^2)$ using Richardson extrapolation. We discuss this technique in \cref{sec:ho_bias_reduction}.

\section{Enhancing End-to-End Pulse Reconstruction Accuracy via Richardson Extrapolation}
\label{sec:ho_bias_reduction}

As outlined in \cref{sec:app:surrogate_model}, the accuracy of the surrogate model in \cref{eqn:surrogate_model_QSP} is first-order. This gives $\delta = \Or(\beta h)$, which seemingly implies that the end-to-end pulse reconstruction accuracy is also first-order according to the analysis in \cref{sec:app:reconstruct_analog}. In this section, we present a method lifting the end-to-end pulse reconstruction error to second-order via Richardson Extrapolation (RE).

We start from a lemma regarding the performance guarantee of RE.

\begin{lem}\label{lem:RE_performance}
    Suppose $h$ is a parameter of step size, and we have an estimation method $F_h(t)$. Let the performance of the estimator be factored in the following parts:
    \begin{equation}
        F_h(t) = f(t) + b(t) h + z(t, h) h^\gamma + \varepsilon + \Or(h^2), \ \forall t \in \mc{I}^\circ
    \end{equation}
    for some interior area $\mc{I}^\circ$. Here, $b(t)$ is a bounded error coefficient, $z(t, h)$ stands for a non-smooth bounded systematic error, and $\varepsilon \sim N(0, \sigma^2)$ is a random error. Furthermore, there is no neighborhood of $h = 0$ on which $z(t, h)$ extends to a differentiable (or even continuous) function in $h$. That means $z(t, h)$ does not admit a Taylor expansion in $h$. We have the following performance guarantee holds for the RE estimator
    \begin{equation}\label{eqn:RE_estimator_and_performance}
        F^\mathrm{RE}_h(t) = 2 F_{h/2}(t) - F_h(t), \quad \max_{t \in \mc{I}^\circ} \mathbb{E} \abs{F^\mathrm{RE}_h(t) - f(t)} \le \Or(h^2) + \Or(h^\gamma) + \Or(\sigma).
    \end{equation}
\end{lem}

The proof of the lemma is straightforward. We remark that \cref{lem:RE_performance} implies that systematic error that admits a smooth coefficient can be improved by using RE; however, systematic errors without a smooth coefficient cannot be improved via RE. In our application, these non-smooth errors include truncation errors from spline interpolations. To understand this rigidity, consider the simplest linear interpolation on the interval $[t_0, t_0 + h]$. Taylor expanding, we have
\begin{equation}
    F^\mathrm{lin}_h(t) = f(t_0)\frac{t_0 + h-t}{h} + f(t_0 + h) \frac{t - t_0}{h} = f(t) \underbrace{- \frac{1}{2} f^{\prime\prime}(t) \eta_h(t) (1 - \eta_h(t))}_{z(t, h)} h^2 + \Or(h^3)
\end{equation}
where
\begin{equation}
    \eta_h(t) = \mathrm{frac}\left(\frac{t - t_0}{h}\right) = \frac{t - t_0}{h} - \left[\frac{t - t_0}{h}\right] \in [0, 1).
\end{equation}
When extending to the piecewise linear interpolation to the full interval, we see that $\eta_h(t) (1 - \eta_h(t))$ is not differentiable. The factor $\eta_h(t)$ similarly involves in higher-order spline interpolation method. Consequently, the systematic truncation error due to spline interpolation is not improvable in RE. However, we can avoid this issue by choosing a spline method with sufficiently high error order. In this work, using cubic spline interpolation suffices.

It remains to show that the truncation error caused by the surrogate model is smooth. The Magnus expansion forms a smooth expansion series in $h$. Hence, the truncation error in the exponential generator is smooth. The mappings between Lie group and Lie algebra are smoothly defined by matrix exponential and logarithm \cite{hall2013lie}. Consequently, the truncation error in the surrogate model propagates to a smooth error in the raw data. Then, applying \cref{lem:RE_performance} gives the following estimation with improved accuracy.

\begin{thm}[Pulse reconstruction method with improved accuracy via Richardson extrapolation]\label{thm:Richardson_extrapolation_estimation_error_bound}
    Assume $\hat{\phi}_h(t)$ is the reconstructed pulse function derived from \cref{sec:app:reconstruct_analog} with cubic splines and step size $h$. Assume the raw data is learned from the surrogate model in \cref{eqn:surrogate_model_QSP} in \cref{sec:app:surrogate_model}. Let $\sigma$ be the standard deviation of the random error in the raw data. Then, there exist two $(h, \sigma)$-independent constants $C_1, C_2$ so that the following pulse reconstruction error holds:
    \begin{equation}
        \hat{\phi}^\mathrm{RE}_h(t) = 2 \hat{\phi}_{h/2}(t) - \hat{\phi}_h(t), \quad \sup_{t \in \mc{I}^\circ} \mathbb{E} \abs{\hat{\phi}^\mathrm{RE}_h(t) - \phi(t)} \le C_1 \beta^2 h^2 + C_2 \sigma.
    \end{equation}
\end{thm}

In \cref{fig:pulse_learning_RE_pulse_shape}, we visualize the reconstruction error in the absence of measurement noise and quantum errors. In this setting, the error reflects only the systematic bias. The results show that increasing $L$ substantially mitigates the bias in the interior region, while the improvement near the boundaries is more limited. This behavior is consistent with our theoretical analysis.

\begin{figure}[htbp]
    \centering
    \includegraphics[width=0.7\linewidth]{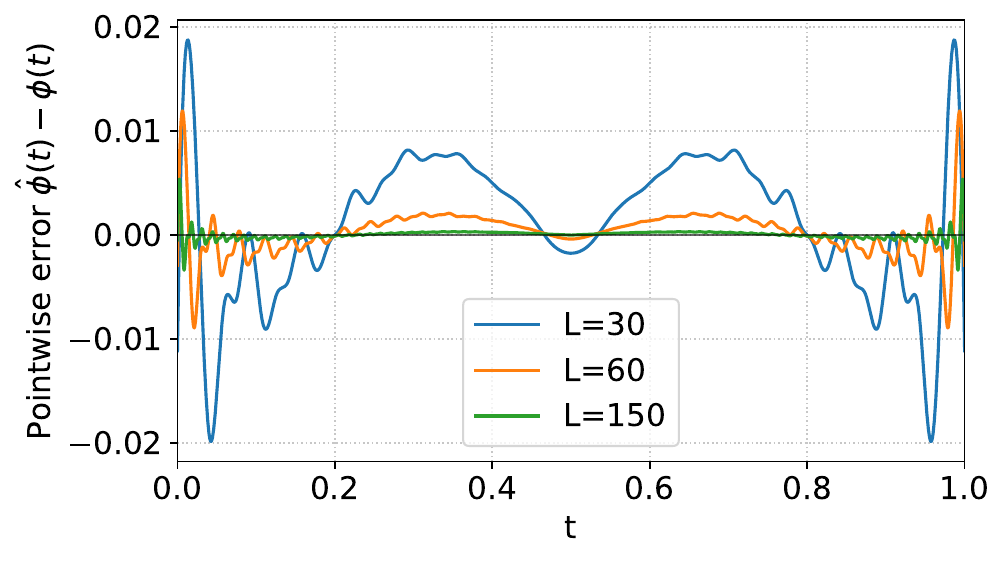}
    \caption{Pointwise reconstruction error of the pulse. The error profiles correspond to the example in \cref{fig:scaling_L} with $\phi(t) = \sin(3\pi t)$. The results obtained using the ``direct'' reconstruction method based on \cref{cor:second_order_deavg} are visualized.}
    \label{fig:pulse_learning_RE_pulse_shape}
\end{figure}

\section{Analysis of Fisher Information}\label{sec:fisher_info}

In the previous sections, our analysis mainly focuses on the order of systematic errors in terms of the scaling of the step size parameter $h = T / L$. Another major factor involved in the error bounds is the random uncertainty parameter $\sigma$. Because each quantum experiment is measured with only finitely many measurement samples, the experimentally derived unitary data are subject to random fluctuations rather than coinciding with the exact values from direct calculation. Such randomness propagates into the uncertainty of the pulse function estimation. To quantify this uncertainty, we comprehensively analyze the Fisher information of the estimation problem in this section. We begin with the structure of the Fisher information. Then, we analyze an interesting phase transition in its singularity when the sampling range in $\theta := \omega h$ does not fully cover the first quadrant.

\subsection{Fisher information matrix}
Fisher information can be used to lower bound the best estimation variance through the Cram\'er-Rao lower bound. In particular, it provides a proxy for how much information each measurement has. In this subsection, we derive its expression and structure in our setup.

Recall that $\theta := \omega T / L$ and our building block is $V(\theta, \psi) = e^{-\I \theta (\cos(\psi) X + \sin(\psi) Y)}$. Let $\sigma_\psi := \cos(\psi) X + \sin(\psi) Y$.
First, the following identity holds
\begin{equation}
    V(\theta, \psi) = e^{- \I \theta \sigma_\psi} = \cos \theta - \I \sin \theta \cdot \sigma_\psi = \begin{pmatrix}
        \cos \theta & - \I e^{- \I \psi} \sin \theta\\
        - \I e^{\I \psi} \sin \theta & \cos \theta
    \end{pmatrix} = e^{- \I \psi / 2 Z} e^{- \I \theta X} e^{\I \psi / 2 Z}.
\end{equation}

Given a sequence of digital phase factors $\Psi := (\psi_1, \cdots, \psi_L) \in \RR^L$, the combined time evolution matrix is
\begin{equation}
    W(\theta; \Psi) = \prod_{j = L}^1 V(\theta, \psi_j) = e^{- \psi_L / 2 Z} \prod_{j = L - 1}^{1} \left( e^{- \I \theta X} e^{\I (\psi_{j+1} - \psi_j) / 2 Z} \right) e^{- \I \theta X} e^{\I \psi_1 / 2 Z} = \prod_{j = L}^1 (\cos\theta - \I \sin\theta \sigma_\psi).
\end{equation}
Expanding the right-most expression, we see that the output of such a matrix-product sequence is a parametric family of trigonometric polynomials
\begin{equation}
    W(\theta ; \Psi) = \begin{pmatrix}
        a(\theta; \Psi) + \I b(\theta; \Psi) & c(\theta; \Psi) + \I d(\theta; \Psi) \\
        - c(\theta; \Psi) + \I d(\theta; \Psi) & a(\theta; \Psi) - \I b(\theta; \Psi)
    \end{pmatrix}.
\end{equation}
Here, $a, b, c, d$ are real trigonometric polynomials in $\theta$ which are parametrized by $\Psi$. 

Due to experimental fluctuation, each entry of the unitary is subjected to a complex normal distributed error $\delta W_{ij} \sim CN(0, 1/M)$ where $M$ is the number of measurement samples per experiment. Equivalently, the experimental unitary can be identified with a multivariate normal distribution $N((a, b, c, d)^\top, 1/M I_4)$. Following the standard Fisher information expression of the multivariate normal distribution, we have
\begin{equation}\label{eqn:Fisher-info-element}
    {F}_{ij}(\theta; \Psi) := M \sum_{f \in \{a, b, c, d\}} \frac{\partial f(\theta; \Psi)}{\partial \psi_i} \frac{\partial f(\theta; \Psi)}{\partial \psi_j} = M \Re\left(\braket{0 | \frac{\partial W(\theta; \Psi)}{\partial \psi_i} \frac{\partial W^\dagger(\theta; \Psi)}{\partial \psi_j} | 0}\right) .
\end{equation}

The last equality can be derived from direct calculation. The symmetry of $(i, j)$ indices can be seen from the fact that the right-hand side is real. Hence, we only need to consider the case where $i \le j$ as ${F}_{ij} = {F}_{ji}$. 

The exact form of the Fisher information involves convoluted relations, which are hard to analyze. To understand the structure of Fisher information, let us consider a special case where all phase parameters coincide as $\Psi^\psi_j = \psi_j = \psi \forall j$ and this equality is not known a priori. Direct calculation gives the following:
when $i = j$, we have
\begin{equation}
    {F}_{ii}(\theta; \Psi) = M \sin^2(\theta),
\end{equation}
and the off-diagonal elements are
\begin{equation}
    \begin{split}
       {F}_{ij}(\theta; \Psi^\psi) &= - \frac{M}{4} \Re\left( \braket{0 | e^{2 \I (i - j + 1) \theta \sigma_\psi} + e^{2 \I (i - j - 1) \theta \sigma_\psi} - 2 e^{2 \I (i - j) \theta \sigma_\psi} | 0} \right)\\
        &\stackrel{\star}{=} - \frac{M}{4} \left( \cos(2(i-j+1)\theta) + \cos(2(i-j-1)\theta) - 2 \cos(2(i-j)\theta) \right)\\
        &= \frac{M}{2} \cos(2(i-j)\theta) \left( 1 - \cos(2 \theta) \right)\\
        &= M \sin^2(\theta) \cos(2(i-j)\theta).
    \end{split}
\end{equation}

Equality $\star$ uses relations $e^{\I \theta \sigma_\psi} = \cos(\theta) + \I \sin(\theta) \sigma_\psi$ and $\braket{0 | \sigma_\psi | 0} = 0$. This structure indicates that $F$ and their linear combinations Toeplitz matrices. 

In the next subsection, we analyze the property of the Fisher information for this special class of phase factors $\Psi^\psi$ and show that there is an interesting phase transition in terms of the singularity of the Fisher information matrix.

\subsection{A phase transition of Fisher information for a special class of phase factors}\label{subsec:fisher_info_phase_transition}

Let $\nu$ be the maximum $\theta$ value in a set of experiments. When a set of equidistant $\theta$-values $\{\theta_n := \nu n / N : n = 1, \cdots, N\}$ is used and each $\theta$-experiment is measured $M$ times, the total Fisher information matrix (FIM) is
\begin{equation}\label{eq:Fisher_info_finite_diff_K}
\begin{split}
    \mc{F}_{ij}(\Psi^\psi) &= M \sum_{n = 1}^N F_{ij}(\theta_n; \Psi^\psi) = M \sum_{n = 1}^N \sin^2(\theta_n) \cos(2(i-j)\theta_n)\\
    &\stackrel{d = \abs{i - j}}{=} M \sum_{n = 1}^N \frac{1}{2} \cos(2d\theta_n) - \frac{1}{4} \cos(2(d + 1)\theta_n) - \frac{1}{4} \cos(2(d-1)\theta_n) \\
    &= \frac{M}{4} \left(2 K(d) - K(d + 1) - K(d - 1)\right) =: s_d
\end{split}
\end{equation}
where
\begin{equation}
    K(d) = \sum_{n = 1}^N \cos(2 d \theta_n) = \cos\left(\frac{(N+1) d \nu}{N}\right) \frac{\sin(d\nu)}{\sin(d\nu / N)}.
\end{equation}
Hence, the Fisher Information is a Toeplitz matrix whose first row is $(s_0, s_1, \cdots, s_{L-1})$. 

To investigate the spectral properties of the FIM, we numerically compute its determinant. The determinant of the FIM (DFI) quantifies the information volume on the statistical manifold. As shown in \cref{fig:phase_transition_DFI}, the DFI exhibits a sharp transition. When the maximum sampling-range parameter $\nu$ is below a certain threshold, the DFI remains vanishingly small. 
However, it suddenly increases to a significantly nonzero value and stabilizes once $\nu \ge \pi/2$. A vanishing DFI implies that the information volume represented by the FIM is nearly zero, meaning that little information can be effectively extracted from the experiments. To understand this sharp phase transition, we analyze the upper bounds of the FIM in the remainder of this subsection.

\begin{figure}[htbp]
    \centering
    \includegraphics[width=0.5\linewidth]{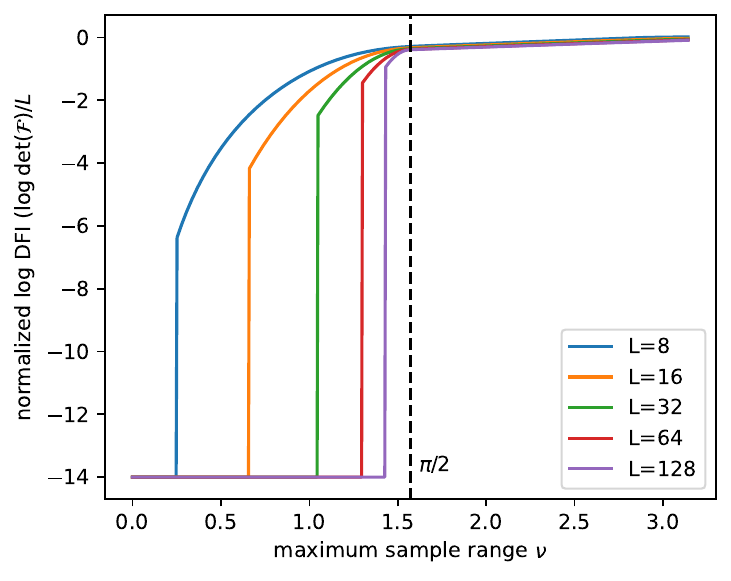}
    \caption{Phase transition of the DFI. The y-axis represents a normalized DFI metric with the inflation effect due to $L$ mitigated. Values below $10^{-14}$ are truncated for clarity. The parameter $M$ is set to $1$ for the simplicity of visualization.}
    \label{fig:phase_transition_DFI}
\end{figure}

The eigenspectrum of the FIM can be estimated by embedding it into a larger circulant matrix. Let $\mc{C}$ be a $(2L-2)\times(2L-2)$ circulant matrix whose first row is $(s_0, \cdots, s_{L - 2}, s_{L-1}, s_{L - 2}, \cdots, s_1)$. Due to the periodicity, the circulant matrix $\mc{C}$ can be diagonalized explicitly using Discrete Fourier Transformation (DFT). So its eigenvalue is
\begin{equation}\label{eq:circulant_eigenvalue}
    \Lambda_k(\mc{C}) = s_0 + 2 \sum_{d = 1}^{L - 2} s_d \cos\left( d \frac{2 \pi k}{2L - 2} \right) + s_{L - 1} (-1)^k, \ k = 0, 1, \cdots, 2L-3.
\end{equation}
Direct computation shows that $\mc{F}(\Psi)$ is the upper-left $L\times L$ submatrix of $\mc{C}$. Let $x \in \RR^L$ be any normalized vector and $\wt{x} := (x, 0_{L-2})^\top \in \RR^{2L-2}$ be its extension. We have
\begin{equation}
    x^\top \mc{F} x = \wt{x}^\top \mc{C} \wt{x} \le \sup_{y \in \RR^{2L-2}: \norm{y} = 1} y^\top \mc{C} y.
\end{equation}
Consequently, the maximum eigenvalues of these two matrices are bounded as follows
\begin{equation}
    \max_{k = 0, \cdots, L - 1} \Lambda_k(\mc{F}) \le \max_{k = 0, \cdots, 2L-3} \Lambda_k(\mc{C}).
\end{equation}
Hence, we upper bound the Fisher Information. Using Cram\'{e}r-Rao lower bound, we lower bound the variance of the phase-factor estimator as follows. Suppose the bias of the phase-factor estimator is sufficiently small, which holds in our case. For any unit coefficient vector $c \in \RR^L$ and the estimator of the phase factors $\hat{\Psi}^\psi$, we have
\begin{equation}
    \mathrm{Var}(c^\top \hat{\Psi}^\psi) \ge c^\top \mc{F}^{-1} c \ge \frac{1}{\max_{k = 0, \cdots, 2L-3} \Lambda_k(\mc{C})}. 
\end{equation}

We discuss two cases.
\begin{enumerate}
    \item \emph{Case 1: $\nu \ll 1$.} Applying Taylor expansion with respect to $\nu$ around zero, we have
    \begin{equation}
        K(d) = N -\frac{2N^2+3N+1}{3N}\,d^{2}\nu^{2} + \Or(\nu^4).
    \end{equation}
    Then
    \begin{equation}
        s_d =\frac{M}{4} \frac{2(2N^2+3N+1)}{3N}\,\nu^2 + \Or(\nu^4) \sim \frac{1}{3} M N \nu^2.
    \end{equation}
    The maximum Fisher Information is bounded as 
    \begin{equation}
        \mathrm{FI}_{\max} \le \max_{k = 0, \cdots, 2L-3} \Lambda_k(\mc{C}) \le \abs{s_0} + \abs{s_{L-1}} + 2 \sum_{d = 1}^{L - 2} \abs{s_d} \sim \frac{2}{3} M N L \nu^2.
    \end{equation}
    Specifically, when $N = \Theta(L)$ and $\nu = \omega_{\max} / L \sim 1 / L$, we have $\mathrm{FI}_{\max} \le \Or(M)$. Because it is $L$-independent, more eigenvalues will lie within the band $[0, \mathrm{FI}_{\max}]$. Hence, it is more likely that zero eigenvalues may occur, leading to a vanishing DFI.

    \item \emph{Case 2: $\nu = \Theta(1)$ is constantly large, and $N \gg L$.} In this case, we can approximate the expression as 
    \begin{equation}
        K(d) = \frac{d \nu / N}{\sin(d \nu / N)} N \frac{\cos(\sin(d \nu) \cos((1 + 1/ N) d \nu)}{d \nu} \approx N \frac{\sin(2 d \nu)}{2 d \nu}. 
    \end{equation}
    This function decays as $d$ increases.
    Extending the DFT in \cref{eq:circulant_eigenvalue} to be continuous, we have
    \begin{equation}
    \begin{split}
        \mc{K}(\omega) &= K(0) + 2 \sum_{d = 1}^{L-2} K(d) \cos(d \omega) + K(L - 1) \cos((L - 1) \omega)\\
        &\approx \int_0^{L - 1} K(x) \cos(\omega x) \ud x \approx \frac{N}{2\nu}\int_0^\infty \frac{\sin x}{x} \cos\left(\frac{\omega}{2\nu} x\right) \ud x = \frac{N}{2 \nu} \left( \pi \mathbb{I}_{0 \le \omega < 2 \nu} + \frac{\pi}{2} \mathbb{I}_{\omega = 2 \nu}  \right).
    \end{split}
    \end{equation}
    To form $\{s_d\}$, the defining equation \cref{eq:Fisher_info_finite_diff_K} applies a discrete central second-order difference stencil to the sequence $\{K(d)\}$. Hence, under DFT, the stencil operator becomes a multiplicative factor. Consequently, when extended to continuous regime, the eigenvalue in \cref{eq:circulant_eigenvalue} becomes
    \begin{equation}
        \Lambda_\omega(\mc{C}) = M \sin^2\left( \frac{\omega}{2} \right) \mc{K}(\omega).
    \end{equation}
    Then, the maximum Fisher Information is bounded as
    \begin{equation}\label{eqn:max_FI_large_nu}
        \mathrm{FI}_{\max} \le \sup_\omega \Lambda_\omega(\mc{C}) \le \frac{\pi}{2 \nu} \sin^2(\nu) M N = \Or(M N).
    \end{equation}
    Unlike Case 1, $\mathrm{FI}_{\max}$ grows linearly with $L$, providing more room for an increasing number of eigenvalues. Consequently, a vanishing DFI is less likely to occur.
\end{enumerate}

Although in case 2, we assume that $N$ is much larger than $L$, the derived bound in \cref{eqn:max_FI_large_nu} remains very tight when $N=\Theta(L)$. The tightness is observed numerically in \cref{fig:Fisher_info_upper_bound}(b). In \cref{fig:Fisher_info_upper_bound}(a), we numerically justify the bound in case 1 when a small $\nu$ value is assumed. The exact maximum Fisher Information value is half of the derived bound. The reason for this factor of two is that we embed the Toeplitz FIM into a large circulant matrix whose size nearly doubles. Alternatively, we may use techniques from Ref. \cite{gray2006toeplitz} to approximately calculate the Toeplitz eigenvalues with a circulant matrix of the same size. Although this alternative mapping agrees only in the asymptotic regime, it would bridge the constant factor-of-two gap.

\begin{figure}[htbp]
  \centering
  \subfloat[]{
    \includegraphics[width=0.45\linewidth]{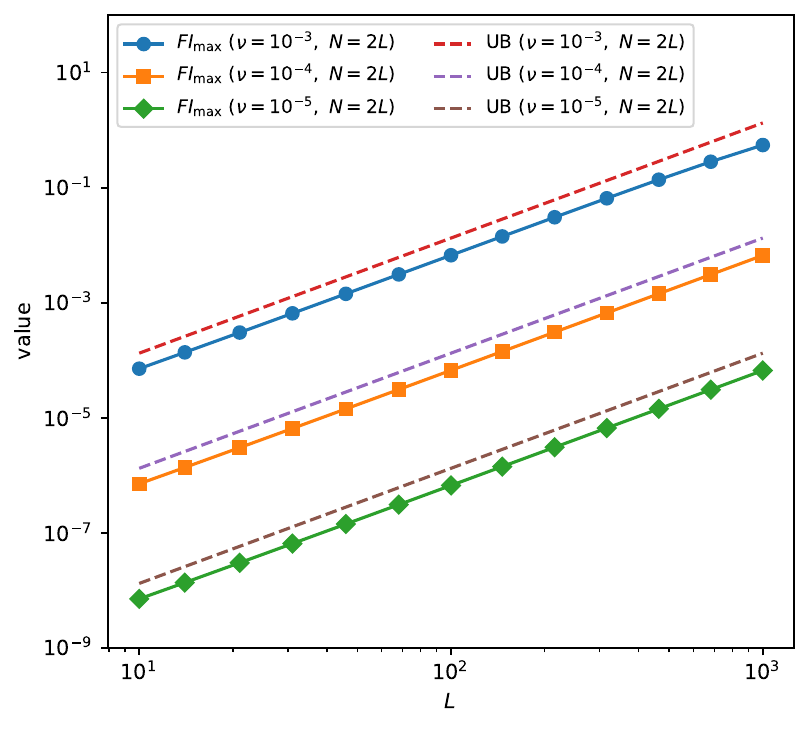}
  }\hfill
  \subfloat[]{
    \includegraphics[width=0.45\linewidth]{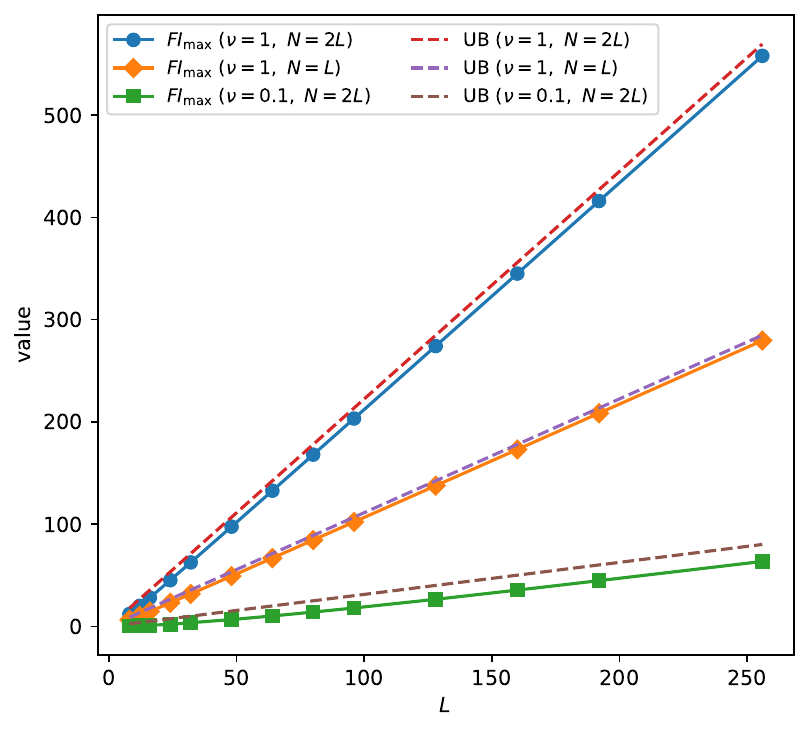}
  }
  \caption{Numerical validation of the derived upper bounds on maximum Fisher Information. (a) Case 1 with small $\nu$ values. The upper bound is $\mathrm{UB} = (2/3) MNL\nu^2$ (b) Case 2 with constantly large $\nu$ values. The upper bound is $\mathrm{UB} = (\pi / (2 \nu)) \sin^2(\nu) MN$. The parameter $M$ is set to one. }
  \label{fig:Fisher_info_upper_bound}
\end{figure}

It is worth noting that when $\nu$ is small, the tight upper bound indicates that the FIM may approach zero, as all eigenvalues are compressed into a narrow interval $[0, \mathrm{UB}]$. According to the theoretical bounds and numerical results in \cref{fig:Fisher_info_upper_bound}, this ill-conditioning can be alleviated as $\nu$ increases. These observations explain the phase transition observed in the DFI behavior.

\subsection{An exactly solvable case when $\nu$ is an integer multiple of $\pi/2$}\label{subsec:Fisher_info_exact_solvable}
The case considered in the previous subsection is exactly solvable when $\nu = r \pi / 2$ with $r \ge 1 \in \ZZ$. Let sample points be the midpoint values on the interval $[0, \nu]$, namely, $\theta \in \{(2n+1)\nu / (2N) : n = 0, \cdots, N - 1\}$. We can also write these points be the set of midpoints on each subinterval of length $\pi / 2$
\begin{equation}\label{eq: point distance}
    \theta_{k, n} = \frac{(2n+1) r \pi}{4 N} + \frac{k \pi}{2},\text{where } k = 0, \cdots, r - 1, \text{ and } n = 0, \cdots, \frac{N}{r} - 1.
\end{equation}
This corresponds to the choice of Chebyshev-Gauss quadrature on each subinterval. When $N = r \lceil L + 1/2 \rceil$, following discrete orthogonality, the following identity holds
\begin{equation}
\begin{split}
    \mc{F}_{ij}(\Psi^\psi) &= M \sum_{k, n} \sin^2(\theta_{k, n}) \cos(2(i-j)\theta_{k, n}) = \frac{2 MN}{r \pi} \int_0^{r \pi / 2} \sin^2(\theta) \cos(2(i-j)\theta) \ud \theta \\
    &= \left\{ \begin{array}{ll}
        M N / 2 & \text{ if } i = j, \\
        - M N / 4 & \text{ if } \abs{i - j} = 1,\\
        0 & \text{otherwise}.
    \end{array}\right.
\end{split}
\end{equation}

Thus, the FIM is a tridiagonal Toeplitz matrix. The inverse of such a discrete Laplacian matrix is well-studied, such as in Ref. \cite{Kay1989}. The element of the inverse FIM is
\begin{equation}
    [\mc{F}^{-1}(\Psi^\psi)]_{ij} = \frac{4}{MN} \left( \min(i, j) - \frac{ij}{L + 1} \right), \ \forall 1 \le i, j \le L.
\end{equation}

We discuss some interesting consequences of this exactly solvable model.
\begin{enumerate}
    \item Single parameter estimation variances. In the large-$L$ limit, the estimation bias of our problem is small enough. Then, the best estimation variance can be characterized by the Cram\'er-Rao lower bound. To evaluate the performance of each single-parameter estimator, we consider the diagonal element of the inverse FIM, which corresponds to the best single-parameter estimation variance:
    \begin{equation}
        \mathrm{Var}(\hat{\psi_i}) \ge [\mc{F}^{-1}(\Psi^\psi)]_{ii} = \frac{4}{MN} \left( i - \frac{i^2}{L + 1} \right), \forall 1 \le i \le L.
    \end{equation}
    It shows that the smallest best variance may be achieved at the boundary parameters:
    \begin{equation}
        \mathrm{Var}(\hat{\psi_1}), \mathrm{Var}(\hat{\psi_L}) \gtrsim \frac{4}{MN}.
    \end{equation}
    When the parameter index gets close to the center, i.e., $i^\star = \lceil (L + 1) / 2 \rceil$, the variance lower bound is maximized:
    \begin{equation}
        \mathrm{Var}(\hat{\psi}_{i^\star}) \gtrsim \frac{L + 1}{MN} = \Omega\left( \frac{1}{M} \right) \text{ if } N = \Theta(L).
    \end{equation}
    \item Correlations among estimators. Let $\sigma_i^2 := \mathrm{Var}(\hat{\psi}_i), \rho_{ij} := \mathrm{Cov}(\hat{\psi}_i \hat{\psi}_j) / (\sigma_i \sigma_j)$. Suppose the variance lower bound is attainable by some best estimation. We first compute the following results:
    \begin{equation}
    \begin{split}
        & \sum_{i = 1}^L \sigma_i^2 = \Tr(\mc{F}^{-1}(\Psi^\psi)) = \sum_{i = 1}^L \frac{4}{MN} \left( i - \frac{i^2}{L + 1} \right) = \frac{2L(L+2)}{3MN},\\
        & \sum_{i = 1}^L \sigma_i^2 + 2 \sum_{i < j} \rho_{ij} \sigma_i \sigma_j = \sum_{i, j} \mc{F}^{-1}_{ij}(\Psi^\psi) = \sum_{i, j} \left( \min(i, j) - \frac{ij}{L + 1} \right) = \frac{L (L+1) (L+2)}{3 MN}.
    \end{split}
    \end{equation}
    Hence, we have
    \begin{equation}
        \sum_{i < j} \rho_{ij} \sigma_i \sigma_j = \frac{L - 1}{4} \sum_{i = 1}^L \sigma_i^2.
    \end{equation}
    Note that Cauchy-Schwarz implies that
    \begin{equation}
        \sum_{i<j}\sigma_i\sigma_j=\frac{1}{2}\left(\left(\sum_{i=1}^{L}\sigma_i\right)^{2}-\sum_{i=1}^{L}\sigma_i^{2}\right) \le \frac{L - 1}{2} \sum_{i=1}^{L}\sigma_i^2.
    \end{equation}
    Combining them, we have the following inequality that holds for the average correlation:
    \begin{equation}\label{eqn:avg_corr_phase_factor}
        \bar{\rho} := \sum_{i < j} \frac{\sigma_i \sigma_j}{\sum_{i < j} \sigma_i \sigma_j} \rho_{ij} = \frac{(L - 1) \sum_{i = 1}^L \sigma_i^2}{4\sum_{i < j} \sigma_i \sigma_j} \ge \frac{1}{2}.
    \end{equation}
    This means that most parameter pairs in the system have a strong positive correlation. Though the number of segmentations $L$ increases, it does not effectively contribute to the reduction of the variance of estimating each segmented parameter. This also justifies the argument of Heisenberg limit scaling. Though there are more segments of a fixed pulse, each segment is a free degree of freedom in the joint estimation. Hence, there is no amplification effect, and the overall estimation variance does not show a decreasing result.
\end{enumerate}

\subsection{Estimation variance when $\pi/2 \le \nu \le \pi$}\label{subsec:fisher_info_general_nu_case}
In the last subsection, we show some exactly solvable cases of the estimation variance. In this subsection, we use semi-definite order (Loewner order) to analyze the estimation variance when $\pi/2 \le \nu \le \pi$.

Let $x := (x_1, \cdots, x_L) \in \mathbb{R}^L$ be any vector. From the defining equation, we see that
\begin{equation}
    x^\top \mc{F} x = \frac{MN}{\nu} \int_0^\nu \sin^2(\theta) \sum_{i, j} x_i \cos(2(i-j)\theta) x_j \ud \theta =  \frac{MN}{\nu} \int_0^\nu \sin^2(\theta) \abs{S(\theta)}^2 \ud \theta
\end{equation}
where $S(\theta) = \sum_{j = 1}^L x_j e^{2\I j \theta}$. Let $x_0 := 0, x_{L+1} := 0$. We also have the following identity 
\begin{equation}
    (1 - e^{2 \I \theta}) S(\theta) = \sum_{j = 1}^L x_j e^{2\I j \theta} - \sum_{j = 1}^L x_j e^{2\I (j + 1) \theta} = \sum_{j = 1}^{L + 1} (x_j - x_{j-1}) e^{2 \I j \theta} =: \sum_{j = 1}^{L + 1} \Delta x_j e^{2 \I j \theta}
\end{equation}
Then, we have
\begin{equation}
    \begin{split}
        & \int_0^\nu \sin^2(\theta) \abs{S(\theta)}^2 \ud \theta = \frac{1}{4} \int_0^\nu \abs{(1 - e^{2 \I \theta})  S(\theta)}^2 \ud \theta \le \frac{1}{4} \int_0^\pi \abs{(1 - e^{2 \I \theta})  S(\theta)}^2 \ud \theta\\
        &= \frac{1}{4} \sum_{i, j} \Delta x_i \Delta x_j \int_0^\pi e^{2 \I (i - j) \theta} \ud \theta = \frac{\pi}{4} \sum_{j = 1}^{L+1} (\Delta x_j)^2 = \frac{\pi}{4} x^\top \mf{D} x.
    \end{split}
\end{equation}
Here, $\mathfrak{D}\in\mathbb{R}^{L\times L}$ denotes the second-difference
matrix. Its entries are
\begin{equation}
    (\mathfrak{D})_{ij}=
\begin{cases}
2, & i=j,\\
-1, & |i-j|=1,\\
0, & \text{otherwise},
\end{cases}
\qquad i,j=1,\dots,L.
\end{equation}
On the other hand, since $\nu \ge \pi / 2$, we have
\begin{equation}
    \begin{split}
        & \int_0^\nu \sin^2(\theta) \abs{S(\theta)}^2 \ud \theta = \frac{1}{4} \int_0^\nu \abs{(1 - e^{2 \I \theta})  S(\theta)}^2 \ud \theta \ge \frac{1}{4} \int_0^{\pi / 2} \abs{(1 - e^{2 \I \theta})  S(\theta)}^2 \ud \theta\\
        &= \frac{1}{8} \int_0^\pi \abs{(1 - e^{2 \I \theta})  S(\theta)}^2 \ud \theta = \frac{\pi}{8} x^\top \mf{D} x.
    \end{split}
\end{equation}
Because these inequalities hold for any vector, the following Loewner order holds
\begin{equation}
    \frac{\pi MN}{8 \nu} \mf{D} \preceq \mc{F} \preceq \frac{\pi MN}{4 \nu} \mf{D}.
\end{equation}
This means the Fisher Information in the general case when $\pi/2 \le \nu \le \pi$ is equivalent to that of the exactly solvable case derived in the previous subsection.
Hence, the maximum estimation variance among these phase-factor estimators is bounded as
\begin{equation}
    \frac{\nu (L+1)}{\pi MN} \le \max_i \mc{F}^{-1}_{ii} \le \frac{2 \nu (L+1)}{\pi MN}.
\end{equation}

Moreover, the similarity in the sense of Loewner order also implies that the conclusions derived in the previous subsection also apply to the case when $\pi / 2 \le \nu \le \pi$. For example, in this case, the average correlation among phase-factor estimators is lower bounded by a constant, which indicates a strong positive correlation of estimation errors.

\section{Learning Surrogate Models via a Quantum-Signal-Processing-based Approach}\label{sec:app_Fourier_analysis_post_processing}

In the previous section, we analyzed the Fisher information of the estimation problem and justified the best estimation accuracy we may expect. The surrogate model in \cref{sec:app:surrogate_model} leads to a structure that coincides with a powerful quantum algorithm, Quantum Signal Processing (QSP). Leveraging this QSP structure, we propose a direct method for estimating digitized pulse values without using any black-box iterative methods such as optimization. We first provide a data augmentation approach so that the full unit circle $\theta \in [0, 2\pi]$ can be covered using the data sampled from the first quadrant. Then, we outline how to extract matrix-valued Fourier coefficients using the Fast Fourier Transform (FFT). Finally, we present the estimation algorithm and analyze its estimation variance.

\subsection{Extending samples from the first quadrant to the full unit circle}

Given a set of data with $\theta$ sampled from the first quadrant, in this subsection, we will show how to extend the data to cover the full unit circle by using the parity and symmetry.

Recall that for some phase factors $\Psi = (\psi_1, \cdots, \psi_L)$, we have
\begin{equation}
    W(\theta; \Psi) = V(\theta, \psi_L) \cdots V(\theta, \psi_1),\text{ where } V(\theta, \psi) = e^{- \I \psi/2 Z} e^{- \I \theta X} e^{\I \psi/2 Z}.
\end{equation}
For the simplicity of notation, we drop the explicit dependency of phase factors, and let
\begin{equation}
    U(\theta) := W(\theta; \Psi)
\end{equation}
denote the matrix $W(\theta; \Psi)$ at angle $\theta$. 

This construction enforces a parity pattern. Let us consider the data point in the second quadrant $\wt{\theta} = \pi - \theta$. We have
\begin{equation}
    V(\pi - \theta, \psi) = e^{- \I \psi/2 Z} e^{- \I \pi X} e^{\I \theta X} e^{\I \psi/2 Z} = - e^{- \I \psi/2 Z} Z e^{- \I \theta X} Z e^{\I \psi/2 Z} = - Z V(\theta, \psi) Z.
\end{equation}
Hence, we have the following mapping
\begin{equation}
    U(\pi-\theta)=(-1)^{L}  Z U(\theta) Z
\end{equation}
which extends first-quadrant data points to cover second quadrant.

Meanwhile, we have
\begin{equation}
    V(2 \pi - \theta, \psi) = e^{- \I \psi/2 Z} e^{- \I 2 \pi X} e^{\I \theta X} e^{\I \psi/2 Z} = e^{- \I \psi/2 Z} Z e^{- \I \theta X} Z e^{\I \psi/2 Z} = Z V(\theta, \psi) Z.
\end{equation}
Consequently, we have
\begin{equation}
    U(2 \pi-\theta)=(-1)^{L}  Z U(\theta) Z
\end{equation}
which mirrors data points in $[0, \pi]$ to cover $[\pi, 2\pi]$.

Thus midpoint samples on the first quadrant uniquely determine midpoint samples on the full circle. The data argumentation procedure is given in \cref{alg:extend-q1-to-full}.

\begin{algorithm}[H]
\caption{Extend first-quadrant midpoint samples to the full unit circle}
\label{alg:extend-q1-to-full}
\begin{algorithmic}
\Require First-quadrant samples $A_j$ for $j=0,\cdots,N-1$ at midpoints $\theta_j\in[0,\frac{\pi}{2}]$ (ascending), degree $L$
\Ensure Full set $B_j$ for $j=0,\cdots,4N-1$ on $[0,2\pi)$ (ascending)
\State Define $Z=\mathrm{diag}(1,-1)$ and $\mathrm{Flip}(O)\gets Z O Z$
\State \textit{First quadrant:} $B_j \gets A_j$ \textbf{for} $j=0,\cdots,N-1$
\State \textit{Second quadrant (mirror $[0,\frac{\pi}{2}]$ to $[\frac{\pi}{2},\pi]$):}
\For{$j=0$ \textbf{to} $N-1$}
  \State $B_{N+j} \gets (-1)^L \mathrm{Flip} \bigl(A_{N-1-j}\bigr)$
\EndFor
\State \textit{Third and fourth quadrants (mirror $[0,\pi]$ to $[\pi,2\pi]$):}
\For{$j=0$ \textbf{to} $2N-1$}
  \State $B_{2N+j} \gets \mathrm{Flip} \bigl(B_{2N-1-j}\bigr)$
\EndFor
\State \Return $B_0,\cdots,B_{4N-1}$
\end{algorithmic}
\end{algorithm}

\subsection{Fourier series structure and analysis}
Note that
\begin{equation}
    V(\theta, \psi) = e^{- \I \psi/2 Z} e^{- \I \theta X} e^{\I \psi/2 Z} = \cos\theta I - \I \sin\theta X e^{\I \psi Z} = \frac{1}{2}(I - X e^{\I \psi Z}) e^{\I \theta} + \frac{1}{2}(I + X e^{\I \psi Z}) e^{-\I \theta}.
\end{equation}
Consequently, directly expanding all $\theta$-valued $X$-rotations shows that the product matrix $U(\theta)$ is a  Fourier series whose degree is within $\pm L$ and coefficients are matrix-valued:
\begin{equation}
U(\theta)=V(\theta, \psi_L) \cdots V(\theta, \psi_1) = \sum_{k=-L}^{L} C_k e^{\I k\theta},\qquad C_k\in\mathbb{C}^{2\times 2}.
\end{equation}
Moreover, the parity condition derived in the previous subsection implies that $C_{-d + k} = 0$ for any $k \in 2 \ZZ + 1$.
Note that our data are sampled on the midpoint grid of the full circle
\begin{equation}
\theta_j=\frac{(2j+1)\pi}{\wt{N}},\quad j=0,\cdots,\wt{N}-1,\text{ where } \wt{N} := 4 N.
\end{equation}
Hence, the $j$-th sample admits the following form
\begin{equation}
    A_j := U(\theta_j) = \sum_{k = - L}^L \underbrace{C_k e^{\I k \pi / \wt{N}}}_{\wt{C}_k} e^{\I 2 \pi j k / \wt{N}} = \sum_{p = 0}^{\wt{N} - 1} \mf{F}_{jp} \wt{C}_p.
\end{equation}
Here, $\mf{F}_{kp} = e^{\I 2 \pi kp / \wt{N}}$ is the $\wt{N} \times \wt{N}$ DFT matrix element, and we define the following relation
\begin{equation}
    \wt{C}_p = \left\{
    \begin{array}{ll}
        C_p e^{\I p \pi / \wt{N}} & \text{ if } 0 \le p \le L, \\
        C_{p - \wt{N}} e^{\I (p - \wt{N}) \pi / \wt{N}} & \text{ if } \wt{N} - L \le p \le \wt{N} - 1,\\
        0 & \text{ otherwise}.
    \end{array}
    \right.
\end{equation}
Assembling $\mf{A} := (A_j : j = 0, \cdots, \wt{N}-1)$ as a $\wt{N} \times 2 \times 2$ tensor, the tensor $\wt{\mf{C}} := (\wt{C}_p : p = 0, \cdots, \wt{N} - 1)$ of the same size can be computed through applying FFT to the first axis
\begin{equation}
    \wt{\mf{C}} = \mf{F}^{-1} \mf{A} = \frac{1}{\wt{N}} \mf{F}^\dagger \mf{A}.
\end{equation}
It is worth noting that in \textsf{numpy}, the action of $\mf{F}^\dagger$ is implemented through \textsf{numpy.fft.fft}. With $\wt{N}=4N$, the condition $N > L$ guarantees no aliasing in Fourier modes.

\subsection{An iterative algorithm for estimating phase factors}

In this subsection, we will show how to compute phase factors through an iterative method in the Fourier domain, which is adapted from \cite{Haah2019}.

For the simplicity of presentation, we redefine the phase factors in a reversed index order:
\begin{equation}
    (\varphi_1, \varphi_2, \cdots, \varphi_L) = (\psi_L, \psi_{L-1}, \cdots, \psi_1),\quad W(\theta, \Psi) = V(\theta, \psi_L) \cdots V(\theta, \psi_1) = V(\theta, \varphi_1) \cdots V(\theta, \varphi_L).
\end{equation}

We first note the following expansion
\begin{equation}
    V(\theta, \varphi) = e^{-\I \varphi / 2 Z} e^{- \I \theta X} e^{\I \varphi / 2 Z} = e^{\I \theta} \underbrace{\left(e^{-\I \varphi / 2 Z} \ket{-}\bra{-} e^{\I \varphi / 2 Z}\right)}_{P_\varphi} + e^{- \I \theta} \left(e^{-\I \varphi / 2 Z} \ket{+}\bra{+} e^{\I \varphi / 2 Z}\right) = e^{\I \theta} P_\varphi + e^{- \I \theta} Q_\varphi.
\end{equation}
Here, $P_\varphi$ and $Q_\varphi = I - P_\varphi$ are two projections. Moreover, we have $V^\dagger(\theta, \varphi) = V(-\theta, \varphi)$. When the phase factor is chosen to be equal to the right-most one, we have
\begin{equation}\label{eqn:recursive_qsp_reduction_3}
    W(\theta, \Phi) V(-\theta, \varphi_L) = V(\theta, \varphi_1) \cdots V(\theta, \varphi_{L-1}) = \sum_{k = -L+1}^{L-1} C^\prime_k e^{\I k \theta}. 
\end{equation}
On the other hand, we have
\begin{equation}
    \begin{split}
        W(\theta, \Phi) V(-\theta, \varphi_L) &= \left(\sum_{k = -L}^L C_k e^{\I k\theta}\right) \left(e^{-\I\theta} P_{\varphi_L} + e^{\I \theta} Q_{\varphi_L}\right)\\
        &= C_L Q_{\varphi_L} e^{\I (L+1) \theta} + C_{-L} P_{\varphi_L} e^{-\I(L+1)\theta} + \sum_{k = -L+1}^{L-1} (C_{k-1} Q_{\varphi_L} + C_{k+1} P_{\varphi_L}) e^{\I k \theta},
    \end{split}
\end{equation}
where the parity condition is used to eliminate the terms, $C_{L - 1} = C_{-L + 1} = 0$. Thus, we have
\begin{equation}\label{eqn:recursive_qsp_reduction_1}
    C_L Q_{\varphi_L} = C_{-L} P_{\varphi_L} = 0 \quad \text{ and } \quad C_k^\prime = C_{k-1} Q_{\varphi_L} + C_{k+1} P_{\varphi_L}, \forall k = -L+1, \cdots, L-1.
\end{equation}
According to \cite[Theorem 2]{Haah2019}, these equations can be solved by choosing
\begin{equation}\label{eqn:recursive_qsp_reduction_2}
    P_{\varphi_L} = \frac{C_L^\dagger C_L}{\Tr(C_L^\dagger C_L)} \quad \text{ and } \quad Q_{\varphi_L} = \frac{C_{-L}^\dagger C_{-L}}{\Tr(C_{-L}^\dagger C_{-L})}.
\end{equation}
With the projector derived from \cref{eqn:recursive_qsp_reduction_2}, we can compute the phase factor by the following relation
\begin{equation}\label{eqn:phase-factor-compute}
    \begin{split}
        & (P_{\varphi_L})_{01} - (P_{\varphi_L})_{10} = \I \sin \varphi_L, \ (P_{\varphi_L})_{01} + (P_{\varphi_L})_{10} = - \cos\varphi_L\\
        & \Rightarrow \varphi_L = \mathrm{atan2}\left( \Im((P_{\varphi_L})_{01} - (P_{\varphi_L})_{10}), - \Re((P_{\varphi_L})_{01} + (P_{\varphi_L})_{10}) \right)
    \end{split}
\end{equation}
\cref{eqn:recursive_qsp_reduction_1,eqn:recursive_qsp_reduction_2,eqn:recursive_qsp_reduction_3} form an iterative reduction of the problem size, which allows us to compute the phase factors sequentially.

\begin{algorithm}[H]
\caption{Iterative phase-factor estimation in the Fourier domain}
\label{alg:fourier-reduction}
\begin{algorithmic}
\Require Fourier coefficients $\{C_k\}_{k=-L}^{L}\subset\mathbb{C}^{2\times 2}$ of $U(\theta)=\sum_{k=-L}^{L}C_k e^{ik\theta}$.
\Ensure Phase factors $\hat\Phi=(\hat\varphi_1,\cdots,\hat\varphi_L)$

\State $\hat\Phi \gets$ empty list
\State $\{C^{(L)}_k\} \gets \{C_k\}$ \Comment{initialize the working copy of coefficients}

\For{$j=L, L-1, \cdots, 1$}
  \State Choose projectors $P_{\varphi_j}, Q_{\varphi_j}$ from the extreme modes (see \cref{eqn:recursive_qsp_reduction_2}).
  \State Extract the phase estimator $\hat\varphi_j$ from $P_{\varphi_j}$ (see \cref{eqn:phase-factor-compute}).
  \State Append $\hat \varphi_j$ to $\hat\Phi$
  \State Reduce the Fourier degree by one on each side $\{C_k^{(j-1)}\}_{k = -j+1}^{j-1}$ (see \cref{eqn:recursive_qsp_reduction_1}).
\EndFor
\State \Return $\hat\Phi$
\end{algorithmic}
\end{algorithm}

A single right-to-left reduction accumulates error from the band edge $k=\pm L$ inward, so the last phase factor in the estimation procedure suffers most. To mitigate this, we also run a left-to-right reduction via transpose symmetry. Then, we stitch the two estimates at a midpoint. This two-sided stitching exploits the most reliable portion of each pass and noticeably reduces boundary-driven error accumulation.

\subsection{Analysis of estimation variance}

Note that the phase $\varphi$ is the intrinsic coordinate of the rank-one projector obtained by conjugating the reference state $\ket{-} \bra{-}$ by a $Z$-rotation. This point of view yields a simple differential description of the tangent space, which can be leveraged to analyze the propagation of estimation variance.

Recall the defining equality of the projector
\begin{equation}
  P_\varphi  =  e^{- \I\varphi Z/2} \ket{-} \bra{-} e^{ \I \varphi Z/2}
   =  \frac12 \left(I-\cos\varphi X-\sin\varphi Y\right).
\end{equation}
Differentiating with respect to $\varphi$ gives the tangent direction along the one-dimensional manifold of projectors generated by the $Z$-orbit:
\begin{equation}\label{eqn:relation_projection_norm_hermitian}
  P'_\varphi  =  \frac{\partial P_\varphi}{\partial\varphi}
   =  \frac12 \left(\sin\varphi X-\cos\varphi Y\right)
   =  \frac{\I}{2} [P_\varphi, Z], 
  \Rightarrow \norm{P'_\varphi}_F^2=\frac12 \text{ and } (P_\varphi^\prime)^\dagger = P_\varphi^\prime.
\end{equation}
Then, recall the estimation of the projector through the top Fourier coefficient matrix, e.g., $C := C_L$
\begin{equation}
  P_\varphi  =  \mc{P}(C) := \frac{C^\dagger C}{\Tr(C^\dagger C)}  =  \frac{C^\dagger C}{s},
 \quad s:=\Tr(C^\dagger C)=\|C\|_F^2.
\end{equation}
For a perturbation $\delta C$, the first-order variation of the normalization map is obtained by quotient differentiation:
\begin{equation}\label{eq:dP}
  \delta P_\varphi  =  \frac{C^\dagger \delta C + (\delta C)^\dagger C}{s}
   -  \frac{2 \Re \Tr(C^\dagger\delta C)}{s} P_\varphi.
\end{equation}

Note that $\delta P_\varphi = P^\prime_\varphi \delta \varphi$ and consequently, $(\delta P_\varphi)^\dagger = (P^\prime_\varphi)^\dagger \delta \varphi = P^\prime_\varphi \delta \varphi = \delta P_\varphi$. Projecting both sides onto the tangent direction $P'_\varphi$ in the Frobenius inner product ($\braket{A, B}_F := \Tr(A^\dagger B)$) gives 
\begin{equation}\label{eq:dphi-proj}
  \delta\varphi  =  \frac{\langle \delta P_\varphi, P'_\varphi\rangle_F}{\|P'_\varphi\|_F^2} = 2 \Tr( (\delta P_\varphi)^\dagger \delta P_\varphi^\prime) = 2 \Tr(\delta P_\varphi P_\varphi^\prime) 
   =  2 \Re \Tr(\delta P_\varphi P'_\varphi).
\end{equation}
Here, \cref{eqn:relation_projection_norm_hermitian} and the Hermicity of $\delta P_\varphi$ are used.

Another useful identity here can be derived by differentiating the identity:
\begin{equation}
    \frac{\partial}{\partial \varphi} \Tr(P_\varphi^2) = \frac{\partial}{\partial \varphi} \Tr(P_\varphi) = \frac{\partial 1}{\partial \varphi} = 0 \Rightarrow \Tr(P_\varphi P_\varphi^\prime) = 0.
\end{equation}

Using this identity and substituting \cref{eq:dP} into \cref{eq:dphi-proj}, we have
\begin{equation}\label{eqn:delta_phi_relation_1}
    \delta \varphi = \frac{4}{s} \Re \Tr \big((\delta C)^\dagger C P'_\varphi\big).
\end{equation}

To simplify the resulting expression, we use the following identities
\begin{equation}\label{eqn:P_derivative}
  P'_\varphi=\frac{\I}{2}(P_\varphi Z-ZP_\varphi), 
  \qquad ZP_\varphi=(I-P_\varphi)Z=:Q_\varphi Z, 
  \qquad P_\varphi Z=ZQ_\varphi,
\end{equation}
and the reduction relations $C Q_\varphi=0$ and $C P_\varphi = C(I - Q_\varphi) =C$. Left-multiplying $P'_\varphi$ by $C$ gives
\begin{equation}
  C P'_\varphi  =  \frac{\I}{2}\big(CP_\varphi Z-CZP_\varphi\big)
   =  \frac{\I}{2}\big(CZ - C(Q_\varphi Z)\big)
   =  \frac{\I}{2} C Z.
\end{equation}
Substituting it into \cref{eqn:delta_phi_relation_1}, we have
\begin{equation}\label{eq:dphi-compact}
  \delta\varphi
   =  \frac{4}{s} \Re \Tr \big((\delta C)^\dagger C P'_\varphi\big)
   =  - \frac{2}{s} \Im \Tr \big((\delta C)^\dagger C Z\big)
   =  \frac{2}{s} \Im \langle \delta C,  C Z\rangle_F,
  \qquad s=\|C\|_F^2.
\end{equation}

Now, we can analyze the propagation of the estimation variance based on this relation and the recurrence used in the estimation algorithm. Suppose we are at the stage when estimating the \(j\)-th phase factor.
We start from the exact reduction at this stage:
\begin{equation}\label{eq:red}
C'_{k}  =  C_{k-1} Q_{\varphi_j}  +  C_{k+1} P_{\varphi_j},
\qquad k=-(j-1),\dots,(j-1),
\end{equation}
where the boundary coefficients determine the projectors
\[
P_{\varphi_j}=\frac{C_{j}^\dagger C_{j}}{\Tr(C_{j}^\dagger C_{j})},
\qquad
Q_{\varphi_j}=\frac{C_{-j}^\dagger C_{-j}}{\Tr(C_{-j}^\dagger C_{-j})}=I-P_{\varphi_j}.
\]
According to \cref{eq:dphi-compact}, the error of the next phase is read out by the normalized functional
\begin{equation}\label{eq:ell}
\ell_{j-1}(X) = \frac{2}{s_{j-1}} \Im \big\langle X,  C'_{j-1} Z\big\rangle_F,
\qquad
s_{j-1}:=\|C'_{j-1}\|_F^2,
\end{equation}
so that \(\delta\varphi_{j-1}=\ell_{j-1}(\delta C'_{j-1})\) in the first-order linearization.

We linearize the recurrence relation \cref{eqn:recursive_qsp_reduction_1} by perturbing \(C_{\pm j}\mapsto C_{\pm j}+\delta C_{\pm j}\) and \(\varphi_j\mapsto \varphi_j+\delta\varphi_j\). Note that the Fr\'echet derivative of the projector map \(\mathcal P(C)=\frac{C^\dagger C}{\Tr(C^\dagger C)}\) at $P$ is 
\begin{equation}\label{eq:DP}
D_P[\Delta]  =  P \Delta  +  \Delta^\dagger P  -  2 \Re\Tr(P\Delta) P.
\end{equation}
Using \cref{eqn:P_derivative} and \cref{eq:red}, the first-order variation at \(j-1\) after the reduction is
\begin{equation}\label{eq:deltaCprime}
\delta C'_{ j-1}
= (\delta C_{ j-2}) Q_{\varphi_j}  +  (\delta C_{ j}) P_{\varphi_j}
 +  (C_{j}-C_{j-2}) P'_{\varphi_j} \delta\varphi_j
 +  C_{j} D_{P_{\varphi_j}}[\delta C_{j}]  +  C_{j-2} D_{Q_{\varphi_j}}[\delta C_{-j}] .
\end{equation}
Applying \cref{eq:ell} to \cref{eq:deltaCprime} yields the decomposition of the one-step phase variation
\begin{equation}\label{eq:decomp}
\delta\varphi_{j-1}
= \underbrace{\ell_{j-1} \big((\delta C_{ j})P_{\varphi_j}\big)}_{\text{right branch}}
 +  \underbrace{\ell_{j-1} \big((\delta C_{ j-2})Q_{\varphi_j}\big)}_{\text{left branch}}
 +  \underbrace{g_j \delta\varphi_j}_{\text{phase channel}}
 +  \underbrace{\ell_{j-1} \big(C_{j} D_{P_{\varphi_j}}[\delta C_{j}] + C_{j-2} D_{Q_{\varphi_j}}[\delta C_{-j}]\big)}_{\text{projector-estimation noise}} ,
\end{equation}
where the phase-channel gain is
\begin{equation}\label{eq:gain}
g_j  =  \ell_{j-1} \big((C_{j}-C_{j-2}) P'_{\varphi_j}\big)
 =  \frac{2}{s_{j-1}} \Im\Big\langle (C_{j}-C_{j-2})P'_{\varphi_j},  C'_{j-1} Z\Big\rangle_F .
\end{equation}

These terms are interpreted as follows.
\begin{enumerate}
  \item Right branch \(\ell_{j-1}((\delta C_{j})P_{\varphi_j})\): coefficient noise at the right boundary mode \(j\), filtered by \(P_{\varphi_j}\) and the readout geometry.
  \item Left branch \(\ell_{j-1}((\delta C_{j-2})Q_{\varphi_j})\): symmetric contribution from the left neighbor \(j-2\). 
  \item Phase channel \(g_j \delta\varphi_j\): multiplicative propagation of the previous-stage phase error through the tangent direction \(P'_{\varphi_j}\). 
  \item Projector-estimation noise \(\ell_{j-1}(C_{j}D_{P_{\varphi_j}}[\delta C_{j}] + C_{j-2}D_{Q_{\varphi_j}}[\delta C_{-j}])\): phase error propagated through the projectors due to the imperfectly estimated phase at the previous stage.
\end{enumerate}
These components can be written equivalently as a linear functional of the coefficient noise. For example, we have
\begin{equation}\label{eqn:Riesz_rep_1}
    \ell_{j-1}((\delta C_{j})P_{\varphi_j}) = \mathscr{L}_j^{(R)}(\delta C_j) = \Im \braket{\delta C_j, A_j^{(R)}} \text{ with } A^{(R)}_j = \frac{2}{s_{j-1}} C_{j-1}^\prime Z P_{\varphi_j}.
\end{equation}
Similarly, we have
\begin{equation}\label{eqn:Riesz_rep_2}
    \begin{split}
        & \mathscr{L}^{(L)}_j \leftrightarrow A^{(L)} = \frac{2}{s_{j-1}} C_{j-1}^\prime Z Q_{\varphi_j},\\
        & \mathscr{L}^{(P)}_j \leftrightarrow A^{(P)} = \frac{2}{s_{j-1}} \left( P_{\varphi_j} (G_{P, j} - G_{P, j}^\dagger) - 2 \I \Im \Tr(P_{\varphi_j} G_{P, j}) P_{\varphi_j} \right) \text{ where } G_{P, j} := C_j^\dagger C_{j-1}^\prime Z,\\
        & \mathscr{L}^{(Q)}_j \leftrightarrow A^{(Q)} = \frac{2}{s_{j-1}} \left( Q_{\varphi_j} (G_{Q, j} - G_{Q, j}^\dagger) - 2 \I \Im \Tr(Q_{\varphi_j} G_{Q, j}) Q_{\varphi_j} \right) \text{ where } G_{Q, j} := C_{j-2}^\dagger C_{j-1}^\prime Z.
    \end{split}
\end{equation}

Now, we have the recursive relation of the phase noise
\begin{equation}\label{eqn:recurrence_delta_phi}
    \delta \varphi_{j - 1} = g_j \delta \varphi_j + \sum_{k \in \{R, L, P, Q\}} \mathscr{L}_j^{(k)}(\delta C_j) = g_j \delta \varphi_j + \mathscr{L}_j(\delta C_j).
\end{equation}
Suppose the coefficient noise is proper complex normal distributed $(\delta C_j)_{k, l} \stackrel{\mathrm{IID}}{\sim} \mc{CN}(0, \sigma_j^2)$. Due to the linearity of the functional, $\mathscr{L}(\delta C_j)$ is also proper complex normal distributed. We have
\begin{equation}
    \mathrm{Var}(\mathscr{L}(\delta C_j)) = \mathrm{Var}(\Im \braket{\delta C_j, A}_F) = \frac{1}{2} \mathrm{Var}(\braket{\delta C_j, A}_F) = \frac{\sigma^2_j}{2} \norm{A}_F^2.
\end{equation}
Due to the linearity, the summation term in \cref{eqn:recurrence_delta_phi} can be written compact linear functional whose Riesz representative is the sum of individual ones, $\mathscr{L} \leftrightarrow A_j := A_j^{(R)} + A_j^{(L)} + A_j^{(P)} + A_j^{(Q)}$. To quantify variance propagation, we aggregate the additive contributions into a single parameter
\begin{equation}
\mathrm{Var}\left(\mathscr{L}_j(\delta C_j)\right) \le \frac{\sigma_j^2}{2} \norm{\sum_{k \in \{R, L, P, Q\}} A_j^{(k)}}_F^2 \le  \frac{\sigma_j^2}{2} \left(\norm{A_j^{(R)}}_F^2 + \norm{A_j^{(L)}}_F^2 + \norm{A_j^{(P)}}_F^2 + \norm{A_j^{(Q)}}_F^2\right).
\end{equation}
The cross-term covariance is
\begin{equation}
    2 \mathbb{E}\left(g_j \delta \varphi_j \mathscr{L}_j(\delta C_j)\right) = \sigma_j^2 \frac{2 g_j}{s_j} \Re\langle A_j, C_{j-1}^\prime Z P_{\varphi_j}\rangle_F =: \sigma_j^2 B_j.
\end{equation}
Let
\begin{equation}
    \rho_j := g_j^2, \quad \alpha_j := B_j + \frac{1}{2} \left(\norm{A_j^{(R)}}_F^2 + \norm{A_j^{(L)}}_F^2 + \norm{A_j^{(P)}}_F^2 + \norm{A_j^{(Q)}}_F^2\right).
\end{equation}
Then, the one-step variance recursion is
\begin{equation}\label{eq:var-rec}
\mathrm{Var}(\hat\varphi_{j-1})  \le  \rho_j \mathrm{Var}(\hat\varphi_j) + \alpha_j \sigma_j^2.
\end{equation}
We may assume the noise in the estimation system is stable and approximate $\sigma_j^2$ by the initial data noise level $\sigma^2$.

\begin{figure}[htbp]
  \centering
  \subfloat[]{
    \includegraphics[width=0.32\linewidth]{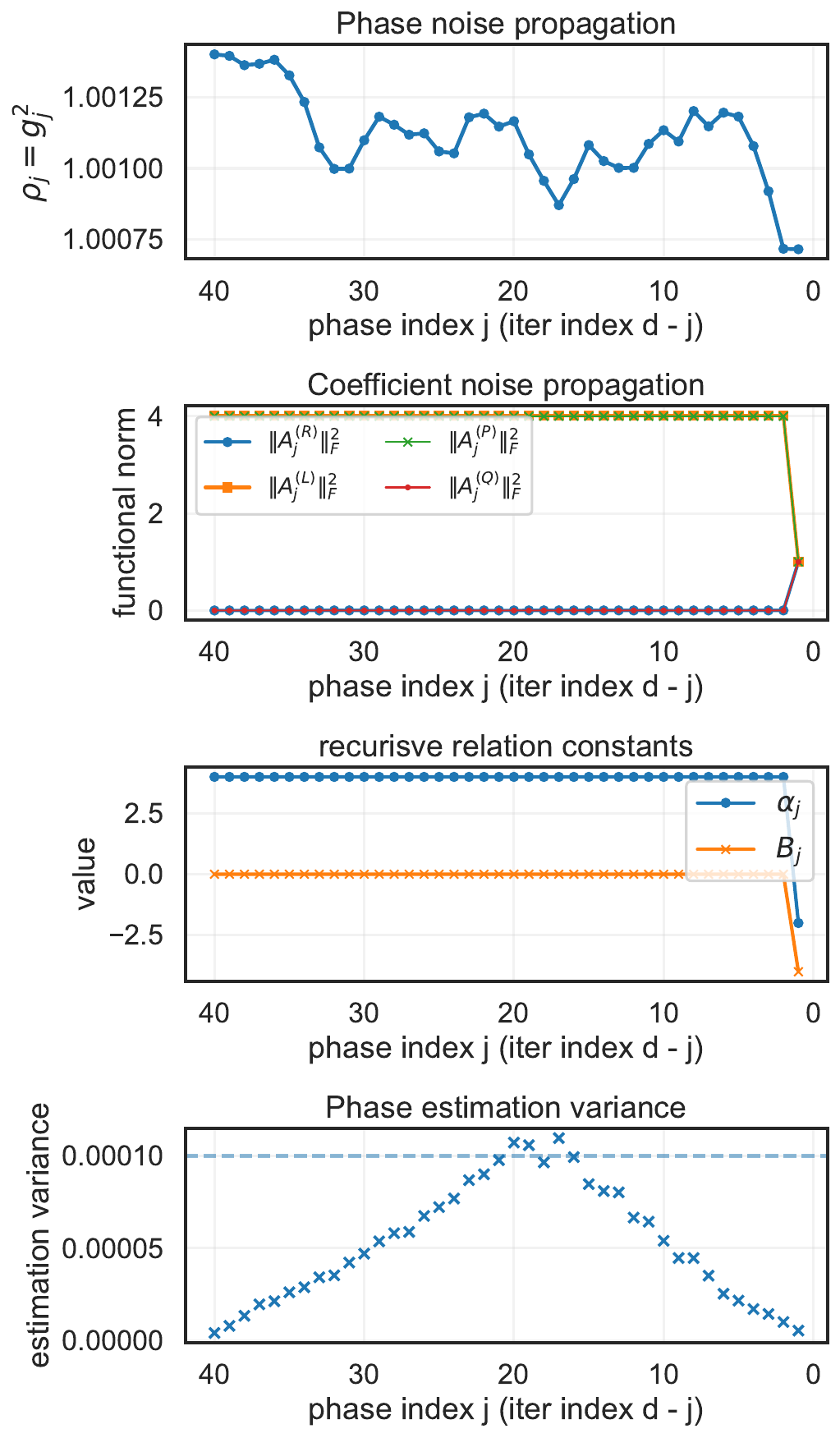}
  }\hfill
  \subfloat[]{
    \includegraphics[width=0.32\linewidth]{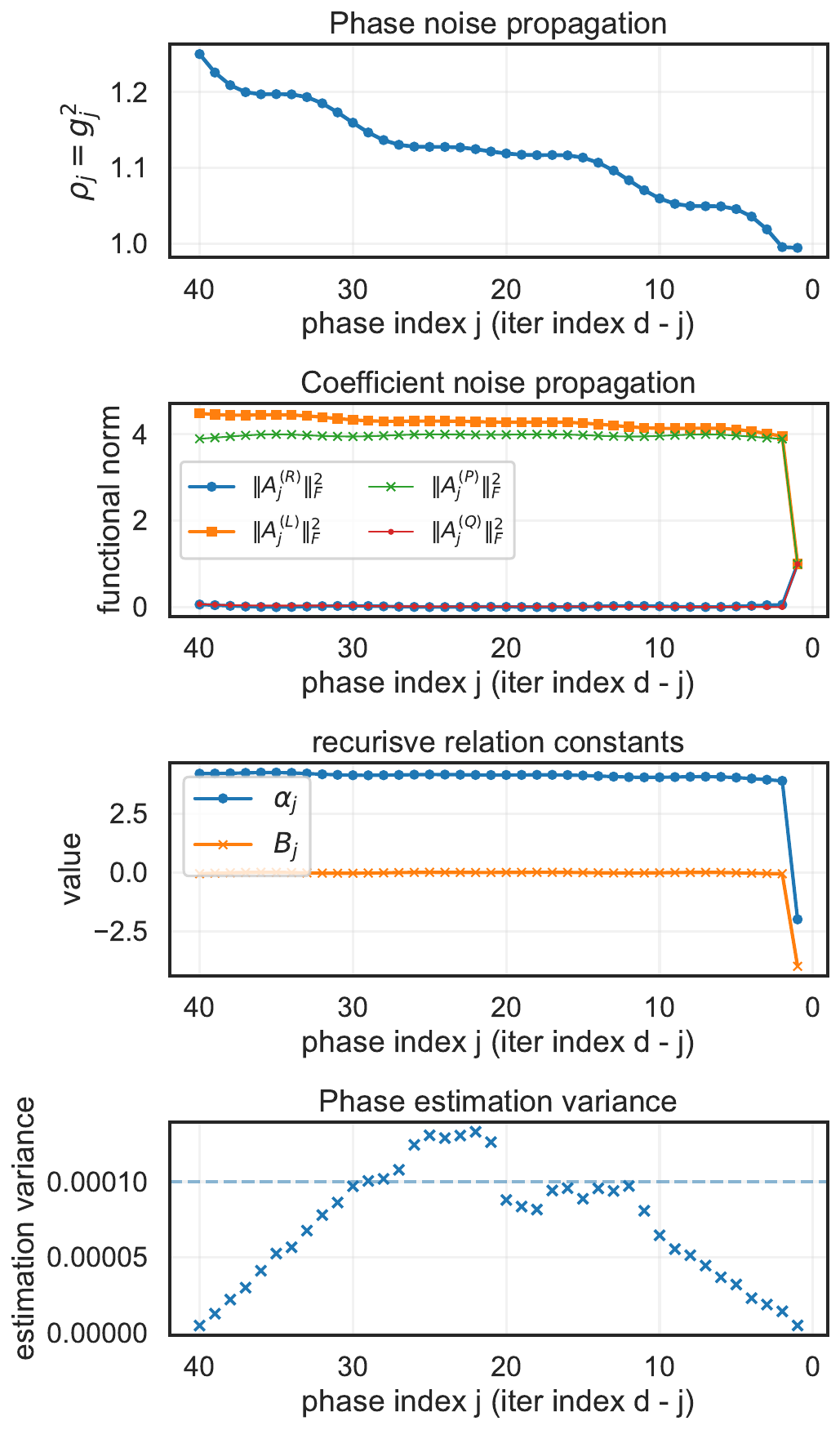}
  }\hfill
  \subfloat[]{
    \includegraphics[width=0.32\linewidth]{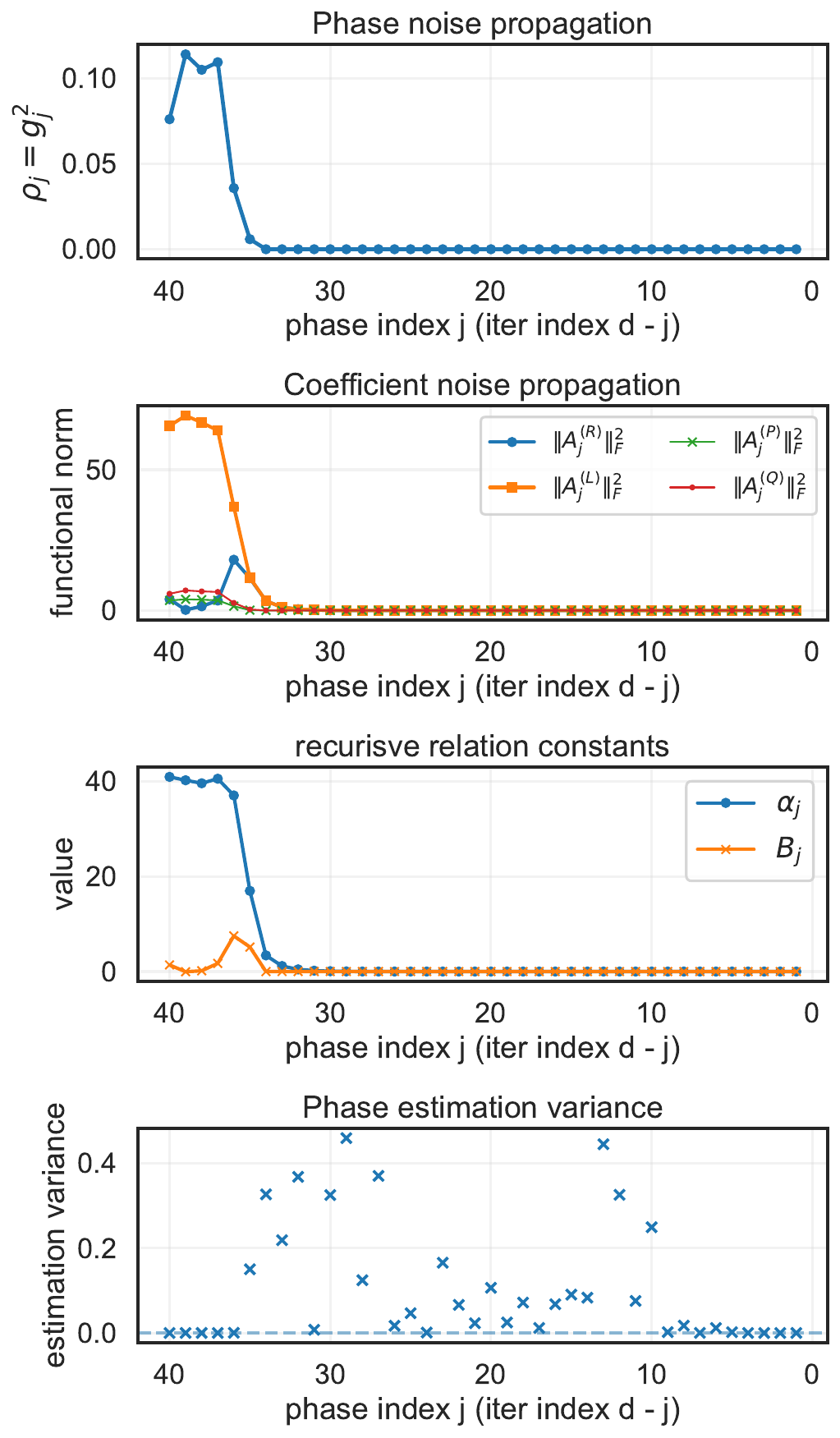}
  }\hfill
  \caption{Propagation of estimation errors. (a) Constant phase factors $\Phi = (1, \cdots, 1)$. (b) Phase factors derived from discretizing the bi-harmonic pulse. (c) Random phase factors $\Phi_j \stackrel{\mathrm{IID}}{\sim} \text{Unif}(-0.7, 0.7)$. Parameters are set to $M = 10,000$ and $d = 40$. For each case, $50$ repetitions of the double-sided estimation algorithm are used to estimate variance. In the bottom panel, the dashed line stands for the data noise level $1/M$.}
  \label{fig:noise_propagation_analysis}
\end{figure}

In \cref{fig:noise_propagation_analysis}, we perform numerical calculations to understand the constants in the recursive relation. To get some theoretical intuition, let us consider a simple case where $\Phi = \Phi^\phi = (\phi, \cdots, \phi)$ contains the same phase factors. The problem structure is drastically simplified
\begin{equation}
    U(\theta) = V^L(\theta, \phi) = P_\phi e^{\I L \theta} + Q_\phi e^{-\I L \theta}.
\end{equation}
Plugging this relation into the defining equations derived above, when $j \ge 1$, we have
\begin{equation}
    \begin{split}
        & g_j = \frac{2}{\norm{P_\phi}_F^2} \Im\braket{P_\phi P_\phi P^\prime_\phi, Z}_F = \Tr(P_\phi) = -1,\\
        & A_j^{(L)} = 2 P_\phi Z Q_\phi = 2 P_\phi Z \Rightarrow \norm{A_j^{(L)}}_F^2 = 4 \norm{P_\phi}_F^2 = 4,\\
        & A_j^{(P)} = 2 P_\phi Z Q_\phi = 2 P_\phi Z \rightarrow \norm{A_j^{(P)}}_F^2 = 4 \norm{P_\phi}_F^2 = 4,\\
        & A_j^{(R)} = 2 P_\phi Z P_\phi = 0, A_j^{(Q)} = 0 \text{ and } B_j = 0.
    \end{split}
\end{equation}

In the final round when $j = 1$, the left coefficient $C_{-1}^{(1)}$ and the right coefficient $C_{1}^{(1)}$ are all nonzero. Consequently, it gives $C_0^{(0)} = I$ and $s_0 = \norm{C_0^{(0)}}_F^2 = 2$. Hence, the structure differs from that in the case when $j > 1$, which has $C_{j - 2}^{(j)} = 0$ and $s_{j - 1} = 1$. When $j = 1$, similar calculations give
\begin{equation}
    \norm{A_1^{(L)}}_F^2 = \norm{P_\phi Z}_F^2 = 1,\ \norm{A_1^{(R)}}_F^2 = \norm{Z P_\phi}_F^2 = 1,\ \norm{A_1^{(P)}}_F^2 = \norm{P_\phi Z}_F^2 = 1,\ \norm{A_1^{(Q)}}_F^2 = \norm{Z P_\phi}_F^2 = 1,\ B_1 = -4.
\end{equation}

These exactly matches the observation in \cref{fig:noise_propagation_analysis} (a).

In \cref{fig:noise_propagation_analysis} (b), it shows that the pattern is almost preserved when the phase-factor sequence is derived from a smooth pulse function. In these cases, $\rho_j \approx 1$ and $\alpha_j \approx c \sigma^2$ where $c$ is a constant and $\sigma^2$ is the noise magnitude in the unitary data. Thus, we have $\mathrm{Var}(\hat\phi_j) \approx c (L - j + 1) \sigma^2$ which grows linearly. This is consistent with the estimation variance derived from Monte Carlo methods. 

In contrast, when the phase factors are completely random, \cref{fig:noise_propagation_analysis} (c) indicates that the propagation of the coefficient noise is very strong. Consequently, the estimation variance blows up to a level that is significantly larger than the noise level in the unitary data. The case of random phase factors may arise from the evolution of a highly unstructured and extremely oscillatory pulse function while the segmentation parameter $L$ is not chosen to be sufficiently large. This scenario is rare in real applications and may be resolved by choosing a larger $L$. 

It is also interesting that \cref{fig:noise_propagation_analysis} (a, b) indicates the phase error propagation and the coefficient error propagation are almost uncorrelated when the pulse function is smooth, i.e., $B_j \approx 0$ when $j > 1$. However, when the pulse function is highly unstructured, \cref{fig:noise_propagation_analysis} (c) shows that these two error channels are positively correlated, which significantly amplifies the phase estimation error further.

\subsection{Numerical simulations of the proved results}

In this subsection, we numerically validate the theoretical results established for phase-factor estimation. As shown in the left panel of \cref{fig:phase_factor_estimation_digital_surrogate}, the maximal estimation standard deviation scales as $1/\sqrt{M}$, while the minimal one scales as $1/\sqrt{LM}$, in agreement with our theoretical analysis. The middle panel further confirms that, in the absence of systematic bias introduced by surrogate modeling, the phase-factor estimation method does not incur any additional bias. Finally, the right panel demonstrates the strong average correlation among the estimated phase factors. This correlation limits the simultaneous reduction of estimation variance across all phase factors.

\begin{figure}[htbp]
    \centering
    \includegraphics[width=\linewidth]{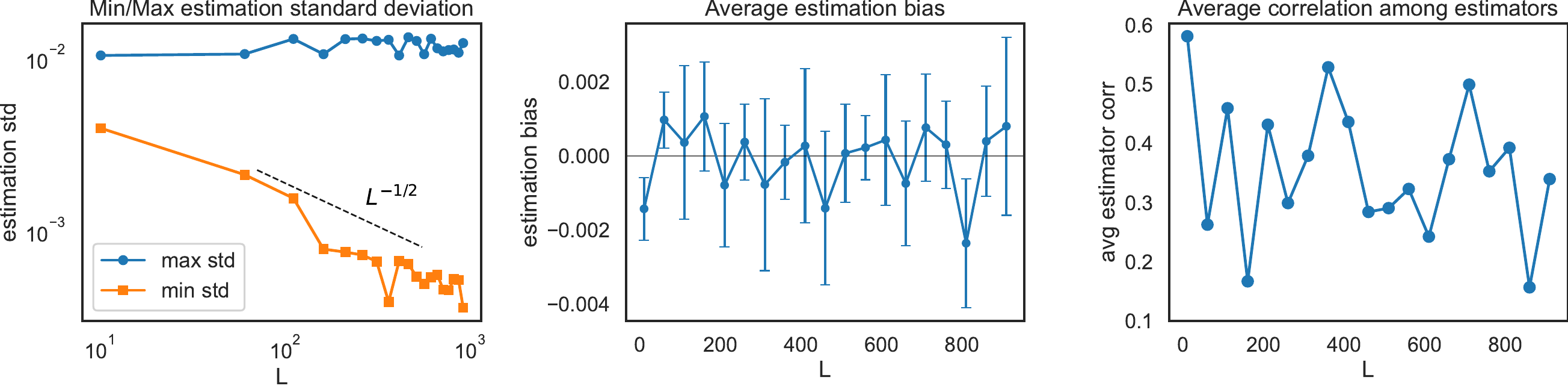}
    \caption{Numerical study of the phase-factor estimation method. The ground-truth phase factors are obtained by discretizing a bi-harmonic pulse. Each query to the time-evolution unitary is corrupted by additive Gaussian noise, and the number of measurement shots is set to $M = 10^4$. The estimation standard deviation, bias, and average correlation (see \cref{eqn:avg_corr_phase_factor}) are computed over 20 independent repetitions.}
    \label{fig:phase_factor_estimation_digital_surrogate}
\end{figure}

\section{Noise-Robust Preprocessing Through a Modified Unitary Tomography}\label{sec:app:robust_preproc_tomography}

In this section, we introduce a preprocessing subroutine for generating the unitary data for learning analog pulse function. This subroutine enhances the estimation robustness against depolarizing error and State Preparation and Measurement (SPAM) error.

\subsection{Unitary tomography on a single-qubit subspace}

Quantum channel tomography \cite{chuang_prescription_1997} is an important technique for understanding the dynamics of a quantum process. To prepare the unitary data for learning analog pulse function, we apply unitary tomography to each experiment with certain $\omega$ value, or equivalently $\theta := \omega T / L$ value. In this subsection, we restate the unitary tomography procedure for the completeness of notations. 

We assume the ability to prepare four pure states as input:
\begin{equation}
    \rho_{z, +} = \ket{0}\bra{0},\ \rho_{z, -} = \ket{1}\bra{1},\ \rho_x = \ket{+}\bra{+},\ \rho_y = \ket{\I}\bra{\I}.
\end{equation}

Let $\mathcal{E}(\rho) = U \rho U^\dagger$ be the unitary dynamics of interest. From experiments, we can measure the following quantities through expectation values
\begin{equation}
    \mathbf{r}_j = (\Tr(\sigma_i \mc{E}(\rho_j)) : i \in \{X, Y, Z\})^\top \in \mathbb{R}^3,\quad j \in \{(z,+), (z, -), x, y\}.
\end{equation}
Our goal is to assemble the Pauli Transfer Matrix (PTM) through experimental data. The PTM is defined as $T_{ij} = \frac{1}{2} \Tr(\sigma_i \mathcal{E}(\sigma_j)), i , j \in \{I, X, Y, Z\}$. Given a trace-preserving channel, we have $T_{00} = 1$ and $T_{0, j} = 0, j = 1, 2, 3$. The PTM be partitioned where each component can be derived from experimental data as follows
\begin{equation}
    T = \begin{pmatrix}
        1 & 0\\ \mathbf{t} & A
    \end{pmatrix},\ \mathbf{t} = \frac{\mathbf{r}_{z, +} + \mathbf{r}_{z, -}}{2},\ (A_{i, 3}) = \frac{\mathbf{r}_{z, +} - \mathbf{r}_{z, -}}{2},\ (A_{i, 1}) = \mathbf{r}_x - \mathbf{t},\ (A_{i, 2}) = \mathbf{r}_y - \mathbf{t}.
\end{equation}

Using Pauli matrices as a matrix basis, the density matrix can be written as a vector and the channel action can be represented as the matrix multiplication $\mathrm{vec}(\mathcal{E}(\rho)) = T \mathrm{vec}(\rho)$. Note that when the channel is unitary action, we have $\mathbf{t} = 0$ and $A \equiv R \in \mathrm{SO}(3)$. The unitary matrix can be recovered from the arithmetic in quaternions
\begin{equation}
    U = (\pm 1) \cdot \begin{pmatrix}
        w - \I z & -\I x - y \\ - \I x + y & w + \I z
    \end{pmatrix},\ w = \frac{1}{2} \sqrt{1 + \Tr(R)},\ x = \frac{R_{32} - R_{23}}{4w}, \ y = \frac{R_{13} - R_{31}}{4 w} , \ z = \frac{R_{21} - R_{12}}{4 w}.
\end{equation}

It is worth noting that the reconstructed operator is determined up to a global sign, i.e., $U$ and $-U$ are indistinguishable in the tomography procedure. 
This ambiguity originates from the fact that the undetermined sign disappears in the adjoint action $U \rho U^\dagger$. This reflects the well-known isomorphism $\mathrm{SO}(3)\cong \mathrm{SU}(2)/\{\pm I\}$, which is also known as the double-covering property of $\mathrm{SU}(2)$. We resolve this issue by introducing a global-phase alignment technique in \cref{sec:app:pm1_phases_continuity}.

\subsection{A modified unitary tomography procedure that is robust against SPAM and depolarizing error}

Note that SPAM error is trace-preserving. Then, the noisy channel is a composite of channels whose PTMs are lower triangular block matrices. Hence, the lower-right submatrix of the composite channel can be written as
\begin{equation}
    T_\text{tot} = T_\text{meas} T_\text{dplz} T_\text{unitary} T_\text{init} = \begin{pmatrix}
        1 & 0\\ * & \wt{A}
    \end{pmatrix},\ \wt{A} := \alpha M R S.
\end{equation}
Here, $M, S$ are the lower-right submatrices of the PTMs of measurement and state preparation errors, and $\alpha$ is the fidelity of the depolarizing channel. 

We can run another experiment with $\omega = 0$ and hence $U = I$ to set as a reference whose lower-right submatrix is $K := \alpha M S$. Let $M = I + \mf{g}_M + \Or(\delta^2)$ and $S = I + \mf{g}_S + \Or(\delta^2)$, where $\delta := \max\{\norm{\mf{g}_M}, \norm{\mf{g}_S}\}$. We have
\begin{equation}
    B := K^{-1/2} \wt{A} K^{-1/2} = R + \frac{1}{2} [\mf{g}_M - \mf{g}_S, R] + \Or(\delta^2),
\end{equation}
 We can perform polar decomposition on this matrix and get an approximation to the $\mathrm{SO}(3)$ rotation
\begin{equation}
    \wt{R} = \mathsf{Polar}(B) := B (B^\top B)^{-1/2} \quad \text{s.t.} \quad \norm{\wt{R} - R} \le \Or(\delta).
\end{equation}
Hence, it gives us a recovered unitary with a first-order SPAM error. The depolarizing error is compensated by introducing the reference $K$ for correction.

Specifically, when the difference of the SPAM error generators is symmetric, namely
\begin{equation}
    \Delta_\text{SPAM} = \frac{1}{2} (\mf{g}_M - \mf{g}_S)\quad \text{s.t.}\ \Delta_\text{SPAM}^\top = \Delta_\text{SPAM},
\end{equation}
this procedure recovers the rotational submatrix and hence the unitary with second-order error. To see it, we first rewrite the previous expression as $B = R(I + E)$ where $E = \frac{1}{2} R^\top[\Delta_\text{SPAM}, R] + \Or(\delta^2) = \frac{1}{2}(R^\top \Delta_\text{SPAM} R - \Delta_\text{SPAM}) + \Or(\delta^2)$. We have $\norm{E} = \Or(\delta)$. Furthermore, because $\Delta_\text{SPAM}$ is symmetric, we have $E$ is symmetric. Then, according to the polar decomposition, we have
\begin{equation}\label{eqn:spam_error_symmetrization}
    \wt{R} = B(B^\top B)^{-1/2} = R(I + E)(I + 2E + \Or(\delta^2))^{-1/2} = R(I + E)(I-E + \Or(\delta^2)) = R + \Or(\delta^2).
\end{equation}
Note that the symmetric property of $E$ is used in the second equality, which is crucial for getting the second-order accuracy in SPAM error. 

In real experiment, the measurement is subjected to sampling noise due to the use of $M$ measurement repetitions per experiment. Then, the assembled PTM is subjected to an additive Gaussian noise whose variance scales as $1/M$. Suppose the depolarizing fidelity is $\alpha$. Note that there are $\Theta(L)$ experiments in total, to ensure that the overall probability is at least constantly large, the measurement sample should be at least
\begin{equation}
    M \sim \Omega(\delta^{-2c}\alpha^{-2} \log(L)).
\end{equation}
Here, $c = 1, 2$ depending on whether the difference-in-generator $\Delta_\text{SPAM}$ is symmetric. Consequently, the total error in the recovered unitary is at the level of the spam error $\Or(\delta^c)$. 

We summarize our results as follows.

\begin{thm}[Robust preprocessing via a modified unitary tomography]\label{thm:app:unitary_tomography}
    Suppose the depolarizing fidelity is $\alpha$, the magnitude of SPAM error is $\delta$, the number of segments $L$, and a reference experiment implementing identity transformation is conducted. The followings hold at a constantly large probability in the presence of SPAM error and depolarizing error.
    \begin{enumerate}
        \item Using $M = \Omega(\delta^{-2} \alpha^{-2} \log(L))$ measurement repetitions per experiment, the unitary matrices can be recovered from the tomography procedure with $2$-norm error at most $\Or(\delta)$.
        \item When the difference-in-generator $\Delta_\mathrm{SPAM}$ is symmetric, using $M = \Omega(\delta^{-4} \alpha^{-2} \log(L))$ measurement repetitions per experiment, the unitary matrices can be recovered from the tomography procedure with $2$-norm error at most $\Or(\delta^2)$.
    \end{enumerate}
\end{thm}

We perform numerical simulations to validate our theoretical analysis. 
The top panels in \cref{fig:robust_unitary_tomography} confirm that the error scalings predicted in \cref{thm:app:unitary_tomography} are precisely reproduced when measurement noise is absent. 
However, as in other panels in \cref{fig:robust_unitary_tomography}, when a finite number of measurement repetitions $M$ is introduced, the overall reconstruction error saturates at a precision floor of $1/(\alpha\sqrt{M})$ imposed by measurement noise. Once the SPAM-induced recovery error falls below this threshold, further improvement requires increasing the number of measurement samples.

\begin{figure*}[htbp]
    \centering
    \includegraphics[width=\linewidth]{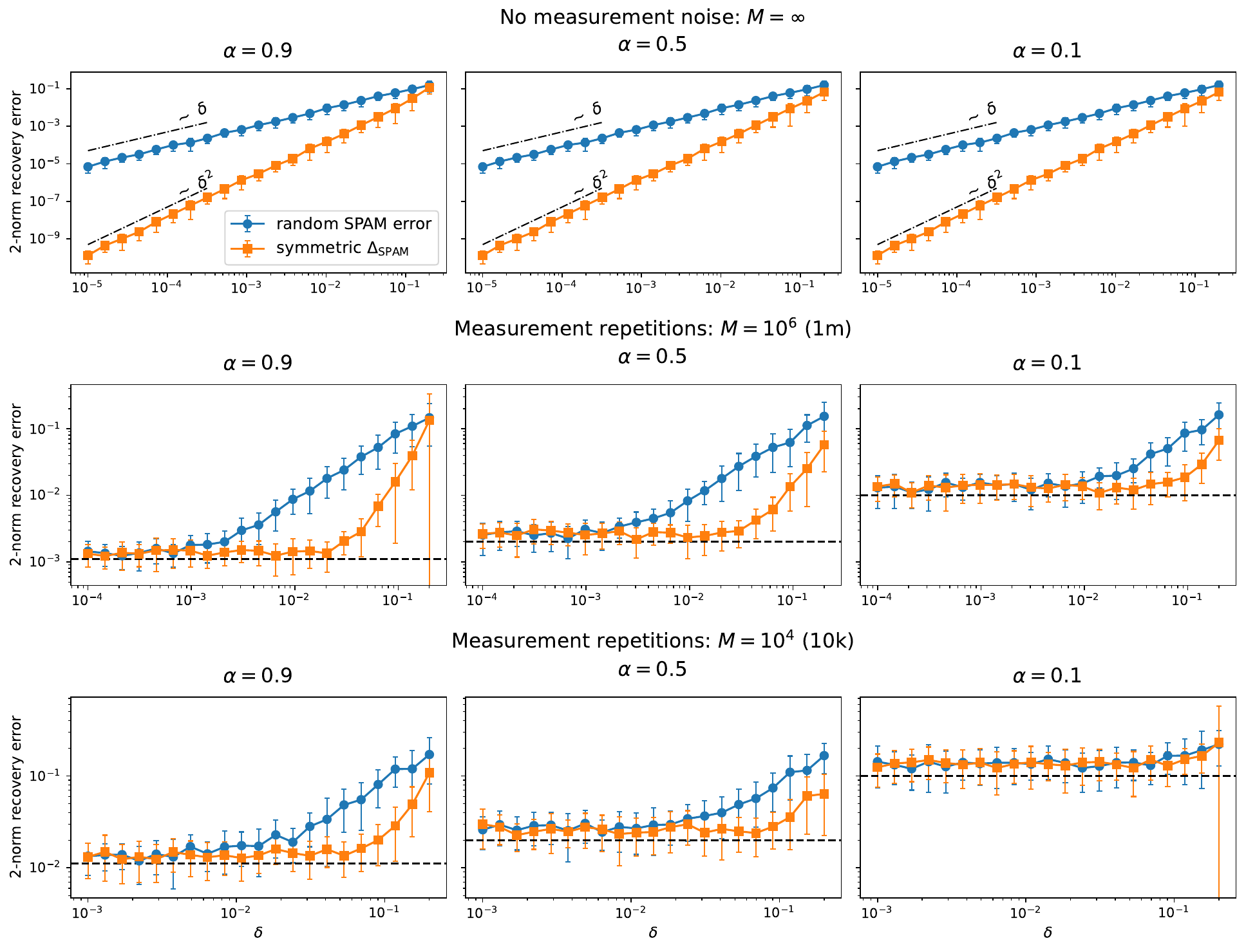}
    \caption{Robust unitary tomography. Measurement noise is modeled by sampling Bernoulli-distributed outcomes in simulated experiments used to reconstruct the Pauli Transfer Matrix. In the middle and bottom rows, the dashed lines indicate the estimated precision floor $1/(\alpha \sqrt{M})$ set by measurement noise. The error bars represent the standard deviation computed from 30 independent random simulations. Each row uses the same axis ranges.}
    \label{fig:robust_unitary_tomography}
\end{figure*}

\subsection{Resolving the sign problem in global phase alignment}\label{sec:app:pm1_phases_continuity}

A practical issue in unitary tomography is the ambiguity of the global phase due to the double-covering property of $\mathrm{SU}(2)$. That means that the tomography result cannot distinguish $U$ and $- U$. Though this does not affect typical applications, this issue may affect the simultaneous processing of a series of unitary matrices corresponding to multiple $\omega$ values unless these signs can be uniformly matched. That is, a mismatch in the sign corresponds to a sharp discontinuity from  some $\omega$ to another $\omega'$.  
We may resolve this by exploiting the continuity of the unitary with respect to $\omega$; i.e., when $|\omega' - \omega|$ is sufficiently small, so should be the corresponding unitaries.
This intuition helps resolve the sign ambiguity in the global phase, providing our sampling mesh has a sufficient number of points. We will quantify this requirement in the remainder of the section.

We first need to quantify the difference between two time evolution matrices.

\begin{lem}
    Let $U(t, 0; \omega)$ be the time evolution matrix defined in \cref{eqn:time_dep_H}. Given $\omega_1, \omega_2 \ge 0$, it holds that
    \begin{equation}
        \norm{U(t, 0; \omega_1) - U(t, 0; \omega_2)} \le \abs{\omega_1 - \omega_2} t.
    \end{equation}
\end{lem}
\begin{proof}
    Let $E(t) := U(t, 0; \omega_1) - U(t, 0; \omega_2)$. Applying \cref{eqn:time_dep_H} gives
    \begin{equation}
    \begin{split}
        \frac{\ud}{\ud t} E(t) &= A(t; \omega_1) U(t, 0; \omega_1) - A(t; \omega_2) U(t, 0; \omega_2)\\
        &= A(t; \omega_1) E(t) + \left(A(t; \omega_1) - A(t; \omega_2)\right) U(t, 0; \omega_2)\\
        &= A(t; \omega_1) E(t) - \I (\omega_1 - \omega_2) \left(\cos(\phi(t)) X + \sin(\phi(t)) Y\right) U(t, 0; \omega_2).
    \end{split}
    \end{equation}
    Applying Duhamel's principle, we have
    \begin{equation}
        E(t) = - \I (\omega_1 - \omega_2) \int_0^t U(s, t; \omega_1) \left(\cos(\phi(s)) X + \sin(\phi(s)) Y\right) U(0, s; \omega_2) \ud s.
    \end{equation}
    Note that the integrand has a unit matrix $2$-norm. Applying the triangle inequality, we have
    \begin{equation}
        \norm{E(t)} \le \abs{\omega_1 - \omega_2} t.
    \end{equation}
    The proof is complete.
\end{proof}

Suppose the unitary matrices derived from tomography are $U_1$ and $U_2$, respectively. When the experimental error in tomography is ignored, they are the exact matrices up to undetermined signs, i.e.,  $U_i = \iota_i U(T, 0; \omega_i)$ for $i = 1, 2$ and $\iota_i \in \{-1, 1\}$. When evaluating their difference, there are two cases:
\begin{itemize}
    \item Matched signs $\iota_1 = \iota_2$. We have 
    \begin{equation}
        \alpha := \norm{U_1 - U_2} = \norm{E(T)} \le \abs{\omega_1 - \omega_2} T.
    \end{equation}
    \item Unmatched signs $\iota_1 = - \iota_2$. We have 
    \begin{equation}
        \begin{split}
            \beta := \norm{U_1 - U_2} &= \norm{U(T, 0; \omega_1) + U(T, 0; \omega_2)} = \norm{2 U(T, 0; \omega_1) - E(T)}\\
            &\ge 2 - \norm{E(T)} \ge 2 - \abs{\omega_1 - \omega_2} T.
        \end{split}
    \end{equation}
\end{itemize}
When $\abs{\omega_1 - \omega_2} T \ll 1$, the difference in these two cases is either close to zero or two. Thus, selecting $\abs{\omega_1 - \omega_2}$ sufficiently small thereby ensures that we may appropriately resolve the unitary's sign. 
By computing the difference
\begin{equation*}
    q := \norm{U_1 - U_2},
\end{equation*}
we can know whether there exists an erroneous sign flip between the two unitaries.

When the number of measurement shots in each tomography experiment is $M$, the recovered unitary matrix is subjected to experimental noise. Suppose the noise is unbiased and has variation scaling as $\sigma_\alpha, \sigma_\beta = \Or(1/M)$. Then, the spectral norm difference $q$ is either drawn from $q \sim \mc{N}(\alpha, \sigma_\alpha^2)$ if the signs match and $q \sim \mc{N}(\beta, \sigma_\beta^2)$ if the signs differ. Therefore, the simple threshold estimator $q < 1$ suffices to decide whether to flip the sign. For any prior probability $\xi \in [0, 1]$, the error probability is: 
\begin{equation}
    \begin{split}
        P(\text{error}) &:= \xi P(q > 1 | q \sim \mc{N}(\alpha, \sigma_\alpha^2)) + (1-\xi) P(q < 1 | q \sim \mc{N}(\beta, \sigma_\beta^2))\\
        &= \xi P\left(z > \frac{1 - \alpha}{\sigma_\alpha} \right) + (1-\xi) P\left(z > \frac{\beta - 1}{\sigma_\beta}\right) \quad \text{ where } z \sim \mc{N}(0, 1)\\
        &\le \xi \exp\left(- \frac{(1 - \alpha)^2}{2 \sigma_\alpha^2}\right) + (1-\xi) \exp\left(- \frac{(\beta - 1)^2}{2 \sigma_\beta^2}\right)\\
        &\le \exp\left(\Or\left(- M(1 - \abs{\omega_1 - \omega_2}T)^2 \right)\right).
    \end{split}
\end{equation}
Here, the tail bound of normal distribution, $P(z > x) \le e^{- x^2 / 2}$ when $x > 0$, is used.

This algorithm trivially scales to sequences of $\omega_i$. Suppose $N$ different $\omega$ values are used in the estimation. Then, we may set $U(\omega_1)$ to be the reference and align the sign of each $U(\omega_i)$ after with $U(\omega_{i - 1})$. 
The overall algorithm is given in \cref{alg:resolve_global_phase}. To ensure the overall probability of success is at least $1 - \zeta$, we apply the union bound and require each individual sign determination step failing with probability at most $\zeta / (N - 1)$. Each tomography experiment requires measurement shots scaling at least as
\begin{equation}
    M \ge \Omega\left( \frac{1}{(1 - \max_{i = 1, \cdots, N - 1} \abs{\omega_i - \omega_{i + 1}}T)^2} \log\left( \frac{\zeta}{N - 1} \right) \right).
\end{equation}

We summarize it as the following theorem about the performance guarantee of \cref{alg:resolve_global_phase}.

\begin{thm}
    When $\max_{i = 1, \cdots, N - 1} \abs{\omega_i - \omega_{i + 1}}T \le 1/2$, and each tomography experiment uses $M = \Omega(\log(\zeta / N))$ number of measurement shots, the probability that \cref{alg:resolve_global_phase} successfully returns a sequence of unitary matrices with fully aligned signs is at least $1 - \zeta$.
\end{thm}

\begin{algorithm}[H]
\caption{Aligning global phases of unitary matrices recovered from tomography}\label{alg:resolve_global_phase}
\begin{algorithmic}
\State \textbf{Input:} An integer $N$, a sequence of unitary matrices recovered from tomography $\{(\omega_i, U_i) : i = 1, \cdots, N\}$, where $\omega_1 < \omega_2 < \cdots < \omega_N$.
\State \textbf{Output:} A sequence of unitary matrices with aligned signs $\{(\omega_i, V_i) : i = 1, \cdots, N\}$.
\State Initialize $V_1 = U_1$.
\For{$i = 2, \cdots, N$}
    \State Compute $q = \norm{U_i - V_{i - 1}}$
    \If{$q < 1$}
    \State Set $V_i = U_i$.
    \Else
    \State Set $V_i = - U_i$.
    \EndIf
\EndFor
\State \Return $\{(\omega_i, V_i) : i = 1, \cdots, N\}$
\end{algorithmic}
\end{algorithm}

\section{Structural Perturbations to Ideal Pulse Functions in Numerical Simulations}\label{sec:app:numerical_results_perturbation}

In our simulations, the implemented pulse $\tilde{\phi}(t)$ is constructed by adding a physically motivated perturbation to the ideal continuous control waveform $\phi(t)$. The goal is to emulate imperfections from classical control electronics. This implementation preserves the continuous form of the ideal pulse and perturbs it additively through a smoothed, piecewise-defined distortion.

We first generate a coarse perturbation profile by dividing $[0, 1]$ into $L_{\mathrm{perturb}}$ equal subintervals. A random perturbation amplitude $p_j \sim \mathrm{Unif}[-\eta, \eta], \eta = 0.5$ is assigned to each interval, producing a piecewise-constant perturbation function $p_{\mathrm{pc}}(t)$ where $p_{\mathrm{pc}}(t) = p_j$ for $t$ in the $j$-th segment. This represents low-resolution imperfections or drift in classical hardware.

Next, this coarse perturbation is smoothed to emulate bandwidth limits of microwave control. In the implementation, we convolve $p_{\mathrm{pc}}(t)$ with a Gaussian window of width $w = 0.02$, producing a differentiable smoothed perturbation $p_{\mathrm{smooth}}(t)$. The final implemented pulse is the ideal waveform corrupted by this filtered perturbation,
\[
\tilde{\phi}(t) = \phi(t) + p_{\mathrm{smooth}}(t),
\]
which retains the high-resolution structure of the target pulse while incorporating realistic low-frequency distortions.

This method yields perturbations that are (i) coarse in origin, reflecting classical resolution limits, yet (ii) smoothed prior to application, reflecting physical bandwidth constraints. Unlike models based on segment-averaged approximations of $\phi(t)$, our perturbation model preserves the continuous ideal pulse and only distorts it through an independent filtered noise profile.

%TC:endignore

\end{document}